\documentclass[conference]{IEEEtran}
\usepackage{cite}
\usepackage{amsmath,amssymb,amsthm}
\usepackage{algorithmic}
\usepackage{graphicx}
\usepackage{textcomp}
\usepackage{xcolor}
\usepackage{circledsteps}
\usepackage{tikz}
\usepackage{url}
\usepackage{mathabx}
\usepackage{tabularx}
\usepackage[capitalize]{cleveref}
\usepackage{csquotes}
\usepackage{tlatex}
\usepackage{pict2e}
\usepackage{pf2}
\def\BibTeX{{\rm B\kern-.05em{\sc i\kern-.025em b}\kern-.08em
    T\kern-.1667em\lower.7ex\hbox{E}\kern-.125emX}}
\usepackage{hyphenat}
\newcommand{\partitle}[1]{}
\newcommand{\sstep}[2]{\langle #1, #2 \rangle}

\newcommand{\arstep}[2]{#1 \rightarrow #2}
\newcommand{\astep}[3]{#1 \xrightarrow{#3} #2}

\newtheorem{assumption}{Assumption}

\makeatletter
\newcommand{\bigcomp}{%
  \DOTSB
  \mathop{\vphantom{\sum}\mathpalette\bigcomp@\relax}%
  \slimits@
}
\newcommand{\bigcomp@}[2]{%
  \begingroup\m@th
  \sbox\z@{$#1\sum$}%
  \setlength{\unitlength}{0.9\dimexpr\ht\z@+\dp\z@}%
  \vcenter{\hbox{%
    \begin{picture}(1,1)
    \bigcomp@linethickness{#1}
    \put(0.5,0.5){\circle{1}}
    \end{picture}%
  }}%
  \endgroup
}
\newcommand{\bigcomp@linethickness}[1]{%
  \linethickness{%
      \ifx#1\displaystyle 2\fontdimen8\textfont\else
      \ifx#1\textstyle 1.65\fontdimen8\textfont\else
      \ifx#1\scriptstyle 1.65\fontdimen8\scriptfont\else
      1.65\fontdimen8\scriptscriptfont\fi\fi\fi 3
  }%
}
\makeatother

\usepackage{color}
\definecolor{boxshade}{gray}{.9}
\setboolean{shading}{true}

\newtheorem{theorem}{Theorem}
\newtheorem{lemma}[theorem]{Lemma}
\newtheorem{definition}[theorem]{Definition}

\begin{document}

\title{Model Checking the Security of the Lightning Network
}

\author{\IEEEauthorblockN{Matthias Grundmann, Hannes Hartenstein}
    \IEEEauthorblockA{
    \textit{KASTEL Security Research Labs}\\
    \textit{Karlsruhe Institute of Technology (KIT)}\\ 
    Karlsruhe, Germany}
}

\maketitle

\thispagestyle{plain} 

\begin{abstract}
Payment channel networks are an approach to improve the scalability of blockchain-based cryptocurrencies.
The Lightning Network is a payment channel network built for Bitcoin that is already used in practice.
Because the Lightning Network is used for transfer of financial value, its security in the presence of adversarial participants should be verified.
The Lightning protocol's complexity makes it hard to assess whether the protocol is secure.
To enable computer-aided security verification of Lightning, we formalize the protocol in TLA\textsuperscript{+} and formally specify the security property that honest users are guaranteed to retrieve their correct balance.
While model checking provides a fully automated verification of the security property, the state space of the protocol's specification is so large that model checking becomes unfeasible.
We make model checking the Lightning Network possible using two refinement steps that we verify using proofs.
In a first step, we prove that the model of time used in the protocol can be abstracted using ideas from the research of timed automata.
In a second step, we prove that it suffices to model check the protocol for single payment channels and the protocol for multi-hop payments separately.
These refinements reduce the state space sufficiently to allow for model checking Lightning with models with payments over up to four hops and two concurrent payments.
These results indicate that the current specification of Lightning is secure.
\end{abstract}

\section{Introduction}
\label{sec-introduction}

Blockchain-based cryptocurrencies do not scale well with respect to their transaction throughput.
One approach to improve said scalability are Payment Channel Networks -- a second layer on top of a blockchain that processes payments without writing a transaction for each payment to the blockchain.
A \emph{payment channel} between two users is opened by publishing one transaction on the underlying blockchain.
Once a payment channel is open, it allows for performing an unlimited number of payments between its two users without writing to the blockchain.
Finally, a payment channel is closed by publishing a second transaction on the blockchain.
In a payment channel \emph{network}, the participating users are connected by payment channels and can perform multi-hop payments using a path between a payment's sender and recipient over a set of payment channels.
The first proposed payment channel network is the Lightning Network \cite{poon_bitcoin_2016} built on Bitcoin \cite{nakamoto_bitcoin:_2008} as the underlying layer. The Lightning Network is already used in practice and experiences a rising adoption \cite{barbaravicius_year-over-year_2024}.

Our goal is to verify that Lightning, the Lightning Network's protocol, is secure, i.e., an honest user finally retrieves on the blockchain the user's correct balance in the payment channel even if other users do not cooperate or are actively malicious.
A critical part in a payment channel protocol is the closing of a payment channel:
The balance that a user finally retrieves is determined by the transactions that are written to the blockchain during the closing of the payment channel.
A user can only be guaranteed to retrieve the user's balance from the payment channel if the user is able to close the payment channel without cooperation of the other party.
Lightning ensures that a user can close a payment channel independently of the other party by asserting that a user always has a valid closing transaction that could be published on the blockchain.
After a payment has been performed in a payment channel, previously valid closing transactions become outdated but these transactions could still be published on the blockchain.
To disincentivize a dishonest user Alice from publishing an outdated closing transaction, the other user, Bob, receives revocation secrets when the closing transaction is outdated.
These revocation secrets enable Bob to punish Alice by retrieving not only his assets but also the assets of Alice if Alice would publish an outdated closing transaction.

The revocation mechanism as well as Lightning's reliance on time and the number of involved parties make it difficult to assess whether Lightning actually fulfills the security property.
Such an assessment might be facilitated by computer-aided methods.
In particular, model checking can automatically verify that a protocol fulfills a property and provide a counterexample if the checked property does not hold for the protocol under investigation.
Lightning is defined by an official specification \cite{web-lightning-bolts} that describes all aspects of the protocol and partially the intuition behind the protocol.
The specification is not directly usable for a security analysis because the specification contains many implementation details and is not formalized.
To enable the use of model checking for Lightning, we contribute a \emph{specification of Lightning} in the formal language TLA\textsuperscript{+} \cite{lamport_temporal_1994,lamport_specifying_2002} that formalizes all protocol steps and messages that users send during opening, updating, and closing a payment channel as well as for multi-hop payments.
In particular, the specification models the revocation mechanism including how users exchange revocation secrets, build closing transactions, and find and react to outdated closing transactions.
We chose TLA\textsuperscript{+} because, as a general purpose language, TLA\textsuperscript{+} allows for specifying arbitrary protocols
and because there are tools supporting different verification methods, e.g. random state space exploration (simulation), explicit-state model checking \cite{web-tlc}, symbolic model checking \cite{web-apalache}, and theorem proving \cite{web-tlaps}.
We also contribute a \emph{specification of the security property} of a payment channel network by defining a secure payment system in TLA\textsuperscript{+}.
The security property formally defines the intuitive property that all honest users have received their correct balances when the protocol terminates.
Our security model allows parties of the protocol to become adversarial which means in our model that they omit sending messages or publishing transactions that are required to be sent by the protocol or that they publish transactions on the blockchain other than specified by the protocol.
Limitations of our adversary model are that an adversarial user does not send arbitrary messages and that adversarial users do not exchange information.

However, the state space of the specification of Lightning is so large that model checking is unfeasible.
We make model checking the Lightning Network possible using two refinement steps that we verify by proofs.
In a first step, we prove that the model of time used in the protocol can be abstracted using ideas from the research of timed automata \cite{alur_automata_1990}.
In a second step, we prove that it suffices to model check the protocol for single payment channels and the protocol for multi-hop payments separately.
These refinements reduce the state space sufficiently to allow for model checking Lightning with the explicit-state model checker TLC \cite{web-tlc}.
We use TLC to fully explore the state space of models with payments over up to four hops and with two concurrent payments.
To check the models also for larger scenarios as well as the whole  stepwise refinement, we use simulation which is a lightweight alternative to model checking in which not the complete state space but only random behaviors are explored.
While our approach does not give comprehensive formal correctness guarantees, it gives confidence that the current specification of Lightning is secure.
We leave writing a complete formal proof and verifying all refinements using a theorem prover for future work.

We describe Lightning in more detail and give an introduction to TLA\textsuperscript{+} in \cref{sec-fundamentals}.
Related work follows in \cref{sec-related-work}.
Our approaches for the specification of Lightning and for the specification of the security property are presented in in \cref{sec-lightning} and \cref{sec-security-property}.
In \cref{sec-proof-sketch}, we explain our approach for verifying that Lightning fulfills the security property and sketch the ideas behind each abstraction step.
We present the results of model checking in \cref{sec-model-checking-results} and discuss limitations of our approach and ideas for future work in \cref{sec-discussion}.
\Cref{sec-conclusion} concludes the paper.

\section{Fundamentals}
\label{sec-fundamentals}

In this section, we introduce Lightning and TLA\textsuperscript{+}.
We do not present all the details of Lightning that are part of our TLA\textsuperscript{+} formalization, but aim to give an overview of the main ideas of Lightning that contribute to the protocol's complexity.

\subsection{Payment Channels: Single-Hop Payments}

A payment channel is a protocol between two users that enables these two users to deposit coins into the payment channel during opening, to perform payments between the two users by updating the payment channel, and to retrieve their final funds by closing the payment channel.
The protocol has to guarantee that its users can close the channel at every point in the execution of the protocol to retrieve their current balance, independently of the cooperation of the other user. Even with an actively malicious channel partner, the protocol has to ensure that an honest user cannot lose funds as long as the user actively monitors the underlying blockchain and reacts to malicious closing attempts.

A payment channel is opened in Lightning by creating a funding transaction.\footnote{See BOLT 2 and BOLT 3 \cite{web-lightning-bolts}.}
The funding transaction has an input spending an output from the funding user (funder)\footnote{Until recently, Lightning supported only single-funded channels, i.e. only one user deposits coins into the channel during opening. Our work, therefore, considers only single-funded channels (see \cref{sec-discussion}).}
and the funding transaction has a multi-signature output that is spendable only by the two users in the channel together.
Just publishing the funding transaction on the blockchain would create a dependence of the funder on the other user for spending the funding transaction's output as the funding transaction's output can only be spent by the two users together.
To prevent such a dependence, an initial commitment transaction that spends the funding transaction's output is created by the two users and the non-funding user sends their signature for the initial commitment transaction to the funder who only publishes the funding transaction after receiving this signature.
Commitment transactions are the `main vehicle' of Lightning. They are cryptographically signed by both parties and encode the state of the payment channel and can be executed on the underlying blockchain.
A commitment transaction has at least two outputs: For each user, there is one output that has an amount that corresponds to the user's current balance and is redeemable only by this user.
In the initial commitment transaction, all funds are spendable by the funder.\footnote{This is a simplification; Lightning's specification allows the funder to send a small amount to the non-funding user already in the initial commitment transaction (see \cite[\texttt{push\_msat}]{web-lightning-bolts}).}

For a payment from one user to the other user in the channel, an HTLC (Hash Timelocked Contract) is added to the channel.
An HTLC is a contract that encodes the agreement that the recipient receives a specified amount if the recipient proves knowledge of a preimage to a specified hash before a specified time has passed.
The process to make a payment is as follows:
The recipient of the payment draws a random value $x$ and calculates the hash value $y = H(x)$ using a cryptographic hash function $H$.
The recipient sends $y$ to the sender of the payment. 
The sender of the payment adds an HTLC to the channel by including into a new commitment transaction a dedicated output which the receiver of the payment can spend by providing a preimage of $y$.
The amount of coins that are part of the HTLC are deducted from the payment's sender's output in the new commitment transaction.
After the HTLC is committed to the payment channel, the recipient of the payment fulfills the HTLC by sending the preimage $x$ to the sender of the payment.
Then, the channel is updated by creating a new commitment transaction without the HTLC output to remove the HTLC and, in the new commitment transaction, the HTLC's amount is added to the recipient's balance.
If the recipient does not fulfill the HTLC before the timelock, the HTLC is also removed but the HTLC's amount is added back to sender's balance.

To update the channel for adding or removing an HTLC, the sender of the payment creates a new commitment transaction and sends a signature for this new commitment transaction to the payment's recipient.
Now, the recipient has two valid commitment transactions, i.e., the current and the new commitment transaction, both signed by the payment's sender.
As a malicious user might publish an outdated commitment transaction, outdated commitment transactions need to be revocable.
As a signature to a commitment transaction cannot be undone, Lightning uses an approach that relies on incentives: A user can be punished for publishing an outdated commitment transaction.
For each commitment transaction, there exists a revocation key pair.
With knowledge of the private revocation key, one user can spend all outputs of the commitment transaction that the user's counterpart in the channel has published.
During an update of a channel, both users send each other their signature for the new commitment transaction and reply by sending the private revocation key for the then outdated commitment transaction.
As the users do not have the private revocation key for the current commitment transaction of their counterpart, they cannot punish each other for correct behavior like publishing the current commitment transaction.
For the security of the protocol, it is crucial that each user has the private revocation keys for revoking outdated commitment transactions and that the other user in the channel does not have the private revocation key for the latest commitment transaction.

\subsection{Payment Channel Networks: Multi-Hop Payments}

If two users do not have a common payment channel but they are connected over a path of payment channels of other users, they can make multi-hop payments between each other.
Intermediate users forward the payment over their channels and might receive a small fee for their service.
To prevent intermediaries from stealing or loosing coins, it must be guaranteed for a multi-hop payment that each intermediary receives an incoming payment on one channel if and only if the intermediary forwards the payment on another channel.
Also, the sender should send the payment to an intermediary if and only if the recipient receives the payment from an intermediary.
Lightning uses HTLCs for multi-hop payments to achieve these security properties.
As in a single-hop payment, the recipient of a payment draws a random value $x$, calculates the hash value $y = H(x)$, and sends $y$ to the sender of the payment.
The sender of the payment creates an HTLC with the first intermediary using $y$ as the hash condition for the HTLC.
The intermediary creates an HTLC with the next hop and each intermediary repeats this process until the last intermediary creates an HTLC with the recipient of the payment.
The recipient fulfills the HTLC by sending $x$ to the last intermediary.
By fulfilling the HTLC, the payment's recipient receives the payment's amount from the last intermediary.
Again, each intermediary forwards the secret value $x$ fulfilling the HTLCs along the route until the sender receives $x$ and pays the first intermediary.
The timelocks of the HTLCs are chosen in a descending order from the sender to the recipient, so that each intermediary has enough time to fulfill the incoming HTLC from the previous hop if the next hop fulfills the outgoing HTLC.

\subsection{TLA\textsuperscript{+}}

We specify Lightning and the security property in TLA\textsuperscript{+}.
The Temporal Logic of Actions (TLA) \cite{lamport_temporal_1994} is a temporal logic to reason about properties of a system.
The language TLA\textsuperscript{+} is based on TLA and used to formalize the behavior of a system.
Using the explicit-state model checker TLC \cite{web-tlc}, the symbolic model checker Apalache \cite{web-apalache}, or a theorem prover (see TLAPS \cite{web-tlaps}), invariants and properties can be verified.
TLA\textsuperscript{+} has been used repeatedly to reason about the properties of systems and protocols (see \cite{merz_formal_2019,bogli_systematic_2025-1}).
We chose TLA\textsuperscript{+} because, as a general purpose language, TLA\textsuperscript{+} allows for specifying arbitrary protocols although abstractions are required for modeling cryptographic primitives (see \cref{sec-lightning}) and because there are tools supporting different verification methods.

In TLA\textsuperscript{+}, the state of a system is described by a set of variables~$v$.
Formally, a state is an assignment of values to variables.
The state space
of a system is the set of all possible states.
A behavior is a sequence of states. A system is described by defining the set of valid behaviors of the system.
A step $\arstep{s}{t}$ is a pair of successive states $s$ and $t$ in a behavior.
An \emph{action} is a function that takes a step and assigns a boolean value to it.
If an action $A$ assigns \textsc{true} to a step $\arstep{s}{t}$, we say that the step $\arstep{s}{t}$ is an $A$ step or a step of $A$ and use the notation $\astep{s}{t}{A}$.
An action is enabled in a state $s$ if there exists an $A$ step starting at state $s$.

The definition of a system in TLA\textsuperscript{+} is the set of all valid behaviors of the system.
A system is described inductively by a set of initial states and an action that determines valid steps of the system.
The set of initial states is defined as the set of states for which a formula with the name $Init$ is valid.
The formula $Init$ has at least one conjunct for each variable of the state that describes the valid values of the variable in the initial state.
The action with the name $Next$ determines which steps are allowed for the system to change its state.
By starting in an initial state and performing steps allowed by the action $Next$, the behaviors of the system can be generated.
A system with variables $v$ is represented as a formula $Spec = Init \land \Box [ Next ]_{v}$ where $\Box$ is the \textit{always} operator of temporal logic and $[ Next ]_{v}$ means $Next$ or a stuttering step in which all variables $v$ are unchanged. 
An additional conjunct may be a fairness condition $WF_v(A)$ that asserts that an $A$ step is taken if the action $A$ is enabled continuously.
The $Next$ action is typically a disjunct of multiple subactions that define different ways for the system's state to be updated.
An action $A$ is described as a conjunction of multiple conjuncts that describe the state in which the action $A$ is enabled
and the new state that is reached by an $A$ step.
Primed variables (e.g., $v'$) are used to describe the values of the variables in the new state
and unprimed variables (e.g., $v$) describe the values of the variables in the current state.
The subactions of $Next$ can be grouped into modules. In the TLA\textsuperscript{+} formalization of Lightning, we use modules to structure the specification by grouping related functionality in a module.
Each module can be instantiated multiple times, parameterized with different sets of variables. We use such instantiations in the formalization to instantiate a module once for each user and channel.

The variables of a system can be internal or external.
Internal variables cannot be seen by an external observer.
From the outside, what a system does is described by the external variables.
In TLA\textsuperscript{+}, it can be expressed that one system implements another system.
A specification $S_1$ implements (or, equally, refines) specification $S_2$, if all external variables of specification $S_1$ are also external variables of specification $S_2$ and the way how the external variables of specification $S_1$ are changed in a behavior of specification $S_1$ also fulfills specification $S_2$.
Internally, the two specifications might work very differently and have different sets of internal variables.
If specification $S_1$ implements $S_2$, we use the notation $S_1 \Rightarrow S_2$.

To show that specification $S_1$ implements $S_2$, we specify a mapping from the state space of specification $S_1$ to the state space of $S_2$.
With a slight formal incorrectness, we also say that behaviors are mapped by the mapping whereby we mean that a behavior $\sigma_1$ of specification $S_1$ is mapped to a behavior $\sigma_2$ of specification $S_2$ by applying the mapping to all states in $\sigma_1$.
The mapping is a refinement mapping if it maps every behavior $\sigma_1$ of specification $S_1$ to a behavior $\sigma_2$ of specification $S_2$ so that the values of the external variables in both behaviors are equal and $\sigma_2$ is a valid behavior of specification $S_2$.
If there exists a refinement mapping from specification $S_1$ to specification $S_2$, then specification $S_1$ implements specification $S_2$ \cite{abadi_existence_1991}.

\subsection{Explicit Real-Time Specifications}
\label{sec_fundamentals_explicit_real_time_specifications}

Lightning depends on time to decide whether an HTLC's timelock has passed or not.
Thus, a specification of Lightning needs to be a real-time specification that models time.
While there are languages especially for modeling real-time systems (e.g., \textsc{Kronos} \cite{bozga_kronos_1998}, 
\textsc{Uppaal} \cite{larsen_uppaal_1997}), we use TLA\textsuperscript{+}, a general purpose language.
Real-time systems can be modeled in TLA\textsuperscript{+} using explicit real-time specifications \cite{lamport_real-time_2005} that we define as follows.
An \textit{explicit real-time specification} has a set of variables for clocks that model the value of time.
There must be at least one clock; however, there might also be multiple clocks.
Because time is defined in Lightning by the height of the blockchain, we restrict all clocks to have discrete values.
Progress of time is modeled by an \textit{AdvanceTime} action (also called \textit{Tick} in the literature) that advances each clock by the same non-negative integer value and leaves all other variables unchanged.
An explicit real-time specification is similar to the model of a timed automaton\footnote{One difference is that timed automata are typically defined with real numbered values for clocks while the time in Lightning is the height of the blockchain which is always an integer.} \cite{alur_automata_1990}.
In \cref{sec-proof-improved_model_of_time}, we apply research results from the area of timed automata to the Lightning specification.

\section{Related Work}
\label{sec-related-work}

Aspects of Bitcoin and Lightning were formally analyzed in previous work.
Andrychowicz et al. \cite{andrychowicz_modeling_2014} modeled Bitcoin contracts as timed automata and verified them using the \textsc{Uppaal} model checker \cite{larsen_uppaal_1997}.
Setzer \cite{setzer_modelling_2018} modeled Bitcoin transactions in Agda \cite{bove_brief_2009}.
Boyd et al. \cite{boyd_blockchain_2020-1} created a model of a blockchain in Tamarin and analyzed Hash Timelocked Contracts, a primitive that is used by Lightning.
Hüttel and Staroveški \cite{huttel_secrecy_2020,huttel_key_2022} formalized four subprotocols of Lightning and analyzed these protocols using ProVerif for secrecy and authenticity properties.
These works on Lightning
are complementary to the problem definition in \cref{sec-introduction} as they show lower level properties of subprotocols but not the security of the combined protocol.
Rain et al. \cite{rain_towards_2023,brugger_checkmate_2023} formalized two subprotocols of Lightning and analyzed these protocols mechanically for game-theoretic security.
Their work is complementary to the problem definition above as their formalization of the protocol assumes that an honest party actually can punish a dishonest party.
This assumption is a property that we aim to prove.

The security of Lightning was analyzed before by Kiayias and Thyfronitis Litos \cite{kiayias_composable_2020}.
Compared to our formalization, the protocol definition of \cite{kiayias_composable_2020} considers more details about the cryptographic aspects.
They specified an ideal functionality of Lightning which is a detailed description of how the parties in a payment channel network behave.
Compared to our definition of the security property, the complexity and the length of this ideal functionality make it hard to assess whether the ideal functionality corresponds to a user's intuitive expectation of security.
They used the UC framework \cite{canetti_universally_2001} to prove that Lightning securely implements this ideal functionality even in the presence of arbitrary behaving adversaries.
While working on our TLA\textsuperscript{+} formalization of Lightning, we found two subtle flaws in the formalization of \cite{kiayias_composable_2020} of Lightning that render the formalized protocol insecure.
However, we believe that these flaws can be corrected and that Lightning fulfills the ideal functionality formalized in \cite{kiayias_composable_2020}.
The first flaw concerns an incomplete description of how a user reacts to maliciously published outdated transactions.
The second flaw is more subtle and concerns how the data in an input is linked to the spending methods of an output that is spent by this input.
A detailed description of the flaws can be found in \cref{sec-appendix-composable}.
While we found the first flaw by comparison of our formalization to the formalization in the paper, we found the second flaw only by model checking when we had a similar flaw in a draft of our formalization.
While the specific flaws can be fixed with low effort, it is difficult and tedious to manually find such flaws in a proof.
Using our approach of model checking instead, such issues can be revealed automatically.

Weintraub et al. \cite{weintraub_payout_2024-1} analyzed the messages exchanged during a single-hop transaction in Lightning.
They inferred properties from the official Lightning specification and found by model checking an ambiguity in the official specification.
This shows the need for an unambiguous formal specification of Lightning.
Further, Weintraub et al. present an attack that relies on honest parties not following the Lightning protocol.
Specifically, they assume that parties consider a payment processed before the payment is actually processed and that, in case of a network partition, the channel parties agree out-of-band on the outcome of a payment instead of closing the channel on the blockchain after a reasonable timeout.
In this work, we assume that honest parties follow the protocol.
We specify explicitly at what point a payment can securely be considered processed (see \cref{sec-lightning}) and check that from this point on an honest party is guaranteed to receive the payment.
Moreover, our work is a step towards a formal and unambiguous official specification of Lightning.

Recently, Fabiański et al. \cite{fabianski_formally_2025} formally verified a simplified variant of Lightning in Why3.
They also formalized the informal security property that honest users do not loose money. In contrast to our approach of defining the security property by defining the behavior of a secure system, they use a game-based definition which is more complex.
While we assume an adversary with limited capabilities, they formally verified that the formalized protocol is secure even in the presence of arbitrary behavior.
Their work shows that a formal verification of even a variant of Lightning considering only single payment channels without HTLCs is a challenging effort.
Their work as well as our work could be used as starting point for a formal verification of a more complete model of the Lightning protocol.

\section{Formalization of Lightning}
\label{sec-lightning}

To be able to use model checking for Lightning, we need to formalize it first, because there exists only an informal official specification.
In this section, we explain important aspects of our TLA\textsuperscript{+} formalization of Lightning to show the assumptions and the abstraction layer of the formalization.
The complete TLA\textsuperscript{+} formalization is available on Github\footnote {\url{https://github.com/kit-dsn/lightning-tla}}.
A detailed presentation of the official specification and our formalization can be found in \cref{sec-appendix-formalization}.

The TLA\textsuperscript{+} formalization specifies a system with an arbitrary number of users.
The behavior of a user is specified as in the official specification \cite{web-lightning-bolts} with some simplifications that we detail below.
The key task of Lightning is to ensure that each user can spend the right transaction output at the right time.
Therefore, our model of Lightning focuses on the protocol logic in which users exchange data to build transactions, publish transactions, and observe transactions on the blockchain.
To keep the complexity at a manageable level, we do not model fees and we abstract cryptographic primitives like signatures and hash functions.

To ensure that our TLA\textsuperscript{+} formalization of Lightning captures the behavior of Lightning as closely as possible, the TLA\textsuperscript{+} formalization follows the structure of the official specification.
We make use of the same identifiers for messages as in the official specification, and the states of HTLCs in the formalization can be mapped to those used in Core Lightning\footnote{\url{https://github.com/ElementsProject/lightning/blob/master/common/htlc_state.h}}, an implementation of Lightning.

Lightning uses primitives such as signatures and hash functions that are not directly available in TLA\textsuperscript{+}.
For the formalization, we make the perfect cryptography assumption that the adversary cannot break cryptographic primitives, and we use a symbolic representation of cryptographic keys and signatures as done in previous work (e.g., \cite{andrychowicz_modeling_2014,boyd_blockchain_2020}).
We abstract these primitives by focusing on their relevant properties that are used in Lightning.
For example, Lightning is based on the assumption that a hash function is deterministic and easy to evaluate given the preimage but cannot be reverted given a hash value and that two different inputs to a hash function result in two different outputs.
In the TLA\textsuperscript{+} formalization, the preimages that are used for multi-hop payments are not randomly generated but are deterministically assigned based on the id of the associated payment.
We model the hash value of a preimage to be equal to the preimage itself and distinguish hashes and preimages by the names of the variables in which a preimage or hash value is stored.
In Bitcoin, transaction identifiers are defined as a hash over the transaction.
In the formalization, we model transaction identifiers by drawing a new unique value for each transaction when the transaction is created and including that identifier in the transaction.

The TLA\textsuperscript{+} formalization contains a model of the blockchain and all transaction types used in Lightning defining the conditions how each transaction output can be spent.
We model publishing a transaction on the blockchain by a single step that happens instantly, i.e., we assume that users have blockchain connectivity and we make the simplifying assumption that each transaction to be published is included in the next block being created.
For the communication between users, we model that messages are delivered reliably and in order but can be arbitrarily delayed.

In Lightning, the height of the blockchain is used as logical time that is relevant for the timeout of HTLCs.
We refer to the height of the blockchain as time, which is formalized as a variable that is advanced by arbitrary integer steps.
Thus, the specification is an explicit real-time specification (see \cref{sec_fundamentals_explicit_real_time_specifications}).
Some actions in the protocol are urgent (see \cite{bornot_modeling_1998}) meaning that they need to happen before a certain point in time, e.g., a user has to fulfill a HTLC before the HTLC's timeout.
We model this by letting each user specify deadlines and not letting time advance beyond a deadline until a step is taken that removes the deadline.

As required by the security model of Lightning, the formalization models that users can be adversarial.
In principle, there are four types of adversarial behavior:
(1)~Omitting sending a message to another user.
(2)~Omitting publishing a transaction on the blockchain.
(3)~Sending a message to another user with arbitrary content. 
(4)~Publishing transactions on the blockchain other than specified by the protocol.
The adversary model for the TLA\textsuperscript{+} formalization allows the adversary to omit sending messages~(1) or publishing transactions~(2).
Also, the adversary is allowed to create and publish certain transactions other than specified by the protocol~(4).
Transactions published by an adversary are only relevant if they spend an output of a transaction that is related to the payment channel.
In the formalization, an action models that the adversary publishes transactions in two ways:
First, by finding all outputs that are spendable for the adversary and publishing a new transaction that spends the output to an output that is owned by the adversary.
Second, by signing and publishing a transaction that the adversary has already received the other user's signature for (e.g., an outdated commitment transaction).
We do not model that the adversary sends messages with arbitrary content~(3) because this would significantly increase the specification's complexity.
In practice, the effect of the adversary sending messages with arbitrary content is limited because users verify every message they receive.
Messages that have an invalid payload or that are received at an invalid point in the protocol execution are ignored.
To check the verification of messages, we explicitly model the verification of every received message and the message's payload, e.g., the verification of signatures and preimages.

We model that any user in the specification can become adversarial.
However, a limitation of our adversary model is that we do not allow information exchange between adversarial users which would model a single adversarial entity controlling multiple users.
This restriction simplifies the verification of the specification and we consider it future work to extend the specification with a broader adversary model.

\section{Security Property}
\label{sec-security-property}

Our goal is to model check the security of Lightning.
Our notion of security is captured by the following informal definition. We define a user as being honest if the user behaves as required by the specification of Lightning.
\begin{definition}[informal security]
\label{def-informal-security}
An honest user will finally get paid out at least the user's correct balance.
\end{definition}
To make this definition precise, we formalize the security property.
One approach to define security are low-level invariants and properties such as `a user's balance is always positive' or `if a user has received the preimage for an outgoing HTLC, the associated incoming HTLC of the user will be fulfilled'.
While formalizing the specification of Lightning, we regularly used such properties to check the draft of the specification for flaws.
However, it is hard to tell for such properties whether they are actually required for the protocol to be secure.
Also, it is hard to tell whether these low-level properties imply what a user intuitively expects as a security property.

The informal definition implicitly concerns four variables:
1. Whether a user is honest. The security property only applies to a user who follows the protocol.
2. A user's balance in a channel which defines the correct balance that a user expects to have. This balance is affected by the payments that are sent and received.
3. A user's view on the status of the user's payments, i.e., whether a payment has been sent or received.
4. A user's balance on the blockchain which are the funds that a user can directly access. The amount that a user is finally paid out is added to this balance.
To formalize the informal definition, we use TLA\textsuperscript{+} to define how these four variables are allowed to change.

\begin{figure}
\small
\begin{tlatex}
\@x{}\moduleLeftDash\@xx{ {\MODULE}
 IdealPaymentNetwork}\moduleRightDash\@xx{}%
  \@x{ {\VARIABLES} BlockchainBalances ,\, ChannelBalances ,\,}%
  \@x{\@s{16.4} Payments ,\, Honest}%
  \@x{ {\CONSTANTS} UserIds ,\, InitialPayments ,\, Numbers}%
  \@pvspace{8.0pt}%
  \@x{ IdealUser ( user ) \.{\defeq} {\INSTANCE} IdealUser {\WITH}}%
  \@x{\@s{16.4} UserId \.{\leftarrow} user ,\,}%
  \@x{\@s{16.4} ChannelBalance \.{\leftarrow} ChannelBalances [ user ] ,\,}%
   \@x{\@s{16.4} BlockchainBalance \.{\leftarrow} BlockchainBalances [ user ]
   ,\,}%
  \@x{\@s{16.4} Payments \.{\leftarrow} Payments [ user ] ,\,}%
  \@x{\@s{16.4} Honest \.{\leftarrow} Honest [ user ]}%
  \@pvspace{8.0pt}%
  \@x{ IdealPayments \.{\defeq} {\INSTANCE} IdealPayments}%
\@pvspace{8.0pt}%
\@x{ Spec \.{\defeq}}%
 \@x{\@s{16.4} \.{\land} \A\, user \.{\in} UserIds \.{:} IdealUser ( user )
 {\bang} Spec}%
\@x{\@s{16.4} \.{\land} IdealPayments {\bang} Spec}%
\@x{}\bottombar\@xx{}%
\end{tlatex}
    \caption{Formal definition of the security property as an idealized payment network.
    The module IdealPaymentNetwork specifies that each user behaves as specified by the module IdealUser (see \cref{fig-tla-idealized-payment-network-user}) and that the users' views on which payments have been processed are consistent as specified in the module IdealPayments (see \cref{fig-tla-idealized-payment-network-payments}).}
    \label{fig-tla-idealized-payment-network}
\end{figure}

\begin{figure}
\small
\begin{tla}
----------------------------- MODULE IdealUser -----------------------------
EXTENDS Integers, SumAmounts
VARIABLES BlockchainBalance, ChannelBalance,
          Payments, Honest
CONSTANTS UserIds, UserId, InitialPayments, Numbers
ASSUME Numbers \subseteq Int

Init ==
    /\ BlockchainBalance \in Numbers
    /\ ChannelBalance = 0
    /\ Payments = {pmt \in InitialPayments :
                    \/ pmt.sender = UserId
                    \/ pmt.receiver = UserId}
    /\ Payments \in SUBSET [amount: Numbers, id: Numbers,
             sender: UserIds, receiver: UserIds,
             state: {"NEW", "ABORTED", "PROCESSED"} ]
    /\ Honest \in {TRUE, FALSE}

Deposit ==
    /\ \E amount \in 1..BlockchainBalance :
        /\ BlockchainBalance' = BlockchainBalance - amount
        /\ ChannelBalance' = ChannelBalance + amount
        /\ ChannelBalance' \in Numbers
    /\ UNCHANGED <<Payments, Honest>>
    
Pay ==
    /\ \E P \in SUBSET {pmt \in Payments : pmt.state = "NEW"} :
          \E r \in [P -> {"ABORTED", "PROCESSED"}] :
            /\ Payments' = (Payments \ P)
                      \cup {[p EXCEPT !.state = r[p]] : p \in P}
            /\ LET PP == {p \in P : r[p] = "PROCESSED"}
                   rBal == SumAmounts({p \in PP :
                              p.receiver = UserId})
                   sBal == SumAmounts({p \in PP :
                              p.sender = UserId})
               IN   /\ ChannelBalance - sBal >= 0
                    /\ ChannelBalance' = ChannelBalance
                                          + rBal - sBal
                    /\ ChannelBalance' >= 0
    /\ UNCHANGED <<Honest, BlockchainBalance>>

Withdraw ==
    /\ BlockchainBalance' \in Numbers
    /\ BlockchainBalance' >= BlockchainBalance
    /\ \E amount \in 0..ChannelBalance :
        /\ ChannelBalance' = ChannelBalance - amount
        /\ Honest => BlockchainBalance' >=
                        BlockchainBalance + amount
    /\ UNCHANGED <<Payments, Honest>>
    
Next == Deposit \/ Pay \/ Withdraw
vars == <<BlockchainBalance, ChannelBalance,
              Payments, Honest>>
Spec ==
    /\ Init
    /\ [][Next]_vars
    /\ WF_vars(ChannelBalance > 0 /\ Honest /\ Withdraw)
=============================================================================
\end{tla}
\begin{tlatex}
\@x{}\moduleLeftDash\@xx{ {\MODULE} IdealUser}\moduleRightDash\@xx{}%
\@x{ {\EXTENDS} Integers ,\, SumAmounts}%
\@x{ {\VARIABLES} BlockchainBalance ,\, ChannelBalance ,\,}%
\@x{\@s{50.74} Payments ,\, Honest}%
\@x{ {\CONSTANTS} UserIds ,\, UserId ,\, InitialPayments ,\, Numbers}%
\@x{ {\ASSUME} Numbers \.{\subseteq} Int}%
\@pvspace{8.0pt}%
\@x{ Init \.{\defeq}}%
\@x{\@s{16.4} \.{\land} BlockchainBalance \.{\in} Numbers}%
\@x{\@s{16.4} \.{\land} ChannelBalance \.{=} 0}%
\@x{\@s{16.4} \.{\land} Payments \.{=} \{ pmt \.{\in} InitialPayments \.{:}}%
\@x{\@s{93.87} \.{\lor}\@s{1.23} pmt . sender \.{=} UserId}%
\@x{\@s{93.87} \.{\lor}\@s{1.23} pmt . receiver \.{=} UserId \}}%
 \@x{\@s{16.4} \.{\land} Payments \.{\in} {\SUBSET} [ amount \.{:} Numbers ,\,
 id \.{:} Numbers ,\,}%
\@x{\@s{53.22} sender \.{:} UserIds ,\, receiver \.{:} UserIds ,\,}%
\@x{\@s{53.22} state \.{:} \{\@w{NEW} ,\,\@w{ABORTED} ,\,\@w{PROCESSED} \} ]}%
\@x{\@s{16.4} \.{\land} Honest \.{\in} \{ {\TRUE} ,\, {\FALSE} \}}%
\@pvspace{8.0pt}%
\@x{ Deposit \.{\defeq}}%
 \@x{\@s{16.4} \.{\land} \E\, amount \.{\in} 1 \.{\dotdot} BlockchainBalance
 \.{:}}%
 \@x{\@s{32.72} \.{\land} BlockchainBalance \.{'} \.{=} BlockchainBalance
 \.{-} amount}%
 \@x{\@s{32.72} \.{\land} ChannelBalance \.{'} \.{=} ChannelBalance \.{+}
 amount}%
\@x{\@s{32.72} \.{\land} ChannelBalance \.{'} \.{\in} Numbers}%
\@x{\@s{16.4} \.{\land} {\UNCHANGED} {\langle} Payments ,\, Honest {\rangle}}%
\@pvspace{8.0pt}%
\@x{ Pay \.{\defeq}}%
 \@x{\@s{17.63} \.{\land} \E\, P \.{\in} {\SUBSET} \{ pmt \.{\in} Payments
 \.{:} pmt . state \.{=}\@w{NEW} \} \.{:}}%
 \@x{\@s{36.52} \E\, r \.{\in} [ P \.{\rightarrow} \{\@w{ABORTED}
 ,\,\@w{PROCESSED} \} ] \.{:}}%
 \@x{\@s{44.72} \.{\land} Payments \.{'} \.{=} ( Payments \.{\,\backslash\,} P
 )}%
 \@x{\@s{85.64} \.{\cup} \{ [ p {\EXCEPT} {\bang} . state \.{=} r [ p ] ]
 \.{:} p \.{\in} P \}}%
 \@x{\@s{44.72} \.{\land} \.{\LET} PP \.{\defeq} \{ p \.{\in} P \.{:} r [ p
 ]\@s{20.70} \.{=}\@w{PROCESSED} \}}%
\@x{\@s{75.58} rBal \.{\defeq} SumAmounts ( \{ p \.{\in} PP \.{:}}%
\@x{\@s{126.72} p . receiver\@s{18.06} \.{=} UserId \} )}%
\@x{\@s{75.58} sBal\@s{0.12} \.{\defeq} SumAmounts ( \{ p \.{\in} PP \.{:}}%
\@x{\@s{126.72} p . sender \.{=} UserId \} )}%
 \@x{\@s{56.94} \.{\IN}\@s{8.2} \.{\land} ChannelBalance \.{-} sBal \.{\geq}
 0}%
\@x{\@s{83.78} \.{\land} ChannelBalance \.{'} \.{=} ChannelBalance}%
\@x{\@s{186.79} \.{+} rBal \.{-} sBal}%
\@x{\@s{83.78} \.{\land} ChannelBalance \.{'} \.{\geq} 0}%
 \@x{\@s{17.63} \.{\land} {\UNCHANGED} {\langle} Honest ,\, BlockchainBalance
 {\rangle}}%
\@pvspace{8.0pt}%
\@x{ Withdraw \.{\defeq}}%
\@x{\@s{16.4} \.{\land} BlockchainBalance \.{'} \.{\in} Numbers}%
\@x{\@s{16.4} \.{\land} BlockchainBalance \.{'} \.{\geq} BlockchainBalance}%
 \@x{\@s{16.4} \.{\land} \E\, amount \.{\in} 0 \.{\dotdot} ChannelBalance
 \.{:}}%
 \@x{\@s{32.72} \.{\land} ChannelBalance \.{'} \.{=} ChannelBalance \.{-}
 amount}%
 \@x{\@s{32.72} \.{\land} Honest \.{\implies} BlockchainBalance \.{'}
 \.{\geq}}%
\@x{\@s{115.58} BlockchainBalance \.{+} amount}%
\@x{\@s{16.4} \.{\land} {\UNCHANGED} {\langle} Payments ,\, Honest {\rangle}}%
\@pvspace{8.0pt}%
\@x{ Next \.{\defeq} Deposit \.{\lor} Pay \.{\lor} Withdraw}%
 \@x{ vars\@s{2.10} \.{\defeq} {\langle} BlockchainBalance ,\, ChannelBalance
 ,\,}%
\@x{\@s{60.39} Payments ,\, Honest {\rangle}}%
\@x{ Spec\@s{1.46} \.{\defeq}}%
\@x{\@s{16.4} \.{\land} Init}%
\@x{\@s{16.4} \.{\land} {\Box} [ Next ]_{ vars}}%
 \@x{\@s{16.4} \.{\land} {\WF}_{ vars} ( ChannelBalance \.{>} 0 \.{\land}
 Honest \.{\land} Withdraw )}%
\@x{}\bottombar\@xx{}%
\end{tlatex}
    \caption{
    Part of the security property that defines how a user's variables are allowed to change.
    The action $\mathrm{Deposit}$ specifies how a user deposits funds into the channel.
    $\mathrm{Pay}$ specifies that payments are either processed or aborted and the channel balance is updated accordingly.
    $\mathrm{Withdraw}$ specifies that the channel balance is withdrawn to the blockchain.
    The fairness condition ensures that a user will withdraw if the user's channel balance is positive.
    }
    \label{fig-tla-idealized-payment-network-user}
\end{figure}

\begin{figure}
\small
\begin{tla}
--------------------------- MODULE IdealPayments ---------------------------

EXTENDS Integers
VARIABLE Payments
CONSTANTS UserIds, Numbers

Init ==
    /\ \A user \in UserIds :
        Payments[user] \in SUBSET [amount: Numbers,
             sender: UserIds, receiver: UserIds, id: Numbers,
             state: {"NEW", "ABORTED", "PROCESSED"}]
Pay ==
    /\ \A user \in UserIds :
        \/ UNCHANGED Payments[user]
        \/ \E P \in SUBSET {p \in Payments[user] :
                                      p.state = "NEW"} : 
            /\ \E r \in [P -> {"ABORTED", "PROCESSED"}] :
                /\ \A p \in P :
                    (r[p] = "PROCESSED" /\ p.sender = user)
                        => \E rp \in Payments'[p.receiver] :
                                /\ rp.id = p.id
                                /\ rp.state = "PROCESSED"
                /\ Payments[user]' = (Payments[user] \ P)
                        \cup {[p EXCEPT !.state = r[p]] : p \in P}

Spec == Init /\ [][Pay]_Payments
=============================================================================
\end{tla}
\begin{tlatex}
\@x{}\moduleLeftDash\@xx{ {\MODULE} IdealPayments}\moduleRightDash\@xx{}%
\@pvspace{8.0pt}%
\@x{ {\EXTENDS} Integers}%
\@x{ {\VARIABLE} Payments}%
\@x{ {\CONSTANTS} UserIds ,\, Numbers}%
\@pvspace{8.0pt}%
\@x{ Init \.{\defeq}}%
\@x{\@s{16.4} \.{\land} \A\, user \.{\in} UserIds \.{:}}%
 \@x{\@s{32.72} Payments [ user ] \.{\in} {\SUBSET} [ amount \.{:} Numbers
 ,\,}%
 \@x{\@s{53.22} sender \.{:} UserIds ,\, receiver \.{:} UserIds ,\, id \.{:}
 Numbers ,\,}%
\@x{\@s{53.22} state \.{:} \{\@w{NEW} ,\,\@w{ABORTED} ,\,\@w{PROCESSED} \} ]}%
\@x{ Pay \.{\defeq}}%
\@x{\@s{17.63} \.{\land} \A\, user \.{\in} UserIds \.{:}}%
\@x{\@s{33.95} \.{\lor} {\UNCHANGED} Payments [ user ]}%
 \@x{\@s{33.95} \.{\lor} \E\, P \.{\in} {\SUBSET} \{ p \.{\in} Payments [ user
 ] \.{:}}%
\@x{\@s{149.13} p . state \.{=}\@w{NEW} \} \.{:}}%
 \@x{\@s{50.27} \.{\land} \E\, r \.{\in} [ P \.{\rightarrow} \{\@w{ABORTED}
 ,\,\@w{PROCESSED} \} ] \.{:}}%
\@x{\@s{66.59} \.{\land} \A\, p \.{\in} P \.{:}}%
 \@x{\@s{82.91} ( r [ p ] \.{=}\@w{PROCESSED} \.{\land} p . sender \.{=} user
 )}%
 \@x{\@s{101.28} \.{\implies} \E\, rp \.{\in} Payments \.{'} [ p . receiver ]
 \.{:}}%
\@x{\@s{143.36} \.{\land} rp . id \.{=} p . id}%
\@x{\@s{143.36} \.{\land} rp . state \.{=}\@w{PROCESSED}}%
 \@x{\@s{66.59} \.{\land} Payments [ user ] \.{'} \.{=} ( Payments [ user ]
 \.{\,\backslash\,} P )}%
 \@x{\@s{99.31} \.{\cup} \{ [ p {\EXCEPT} {\bang} . state \.{=} r [ p ] ]
 \.{:} p \.{\in} P \}}%
\@pvspace{8.0pt}%
\@x{ Spec \.{\defeq} Init \.{\land} {\Box} [ Pay ]_{ Payments}}%
\@x{}\bottombar\@xx{}%
\end{tlatex}
    \caption{Part of the security property that defines that the sender of a payment may only see the payment as processed if the receiver of the payment sees the payment as processed.}
    \label{fig-tla-idealized-payment-network-payments}
\end{figure}

We formalize the security property of Lightning by defining the behavior of a secure payment system.
The full specification of the security property is shown in \cref{fig-tla-idealized-payment-network,fig-tla-idealized-payment-network-user,fig-tla-idealized-payment-network-payments}.
The security property has four variables $\mathrm{Honest}$, $\mathrm{ChannelBalances}$, $\mathrm{Payments}$, $\mathrm{BlockchainBalances}$ matching the variables described above.
The security property describes which changes of these variables are allowed.
The definition of the security property is divided into three modules.
The module IdealUser (see \cref{fig-tla-idealized-payment-network-user}) defines four actions $\mathrm{Deposit}$, $\mathrm{Pay}$, and $\mathrm{Withdraw}$ that describe how the variables of single user can change.
The module IdealPayments (see \cref{fig-tla-idealized-payment-network-payments}) ensures that the users' views on which payments have been processed are consistent.
The module IdealPaymentNetwork (see \cref{fig-tla-idealized-payment-network}) connects the two other modules by defining that for all users in the payment network the specification of the module IdealUser must hold and that the specification of the module IdealPayments must hold.

Initially, the variables are initialized (see \cref{fig-tla-idealized-payment-network-user}) so that each user has a channel balance of zero, an integer balance on the blockchain, and a set of payments that the user wants to send or receive.
The variable $\mathrm{Honest}$ specifies for each user whether the user is honest or dishonest.

The action $\mathrm{Deposit}$ describes a deposit by defining that a user's blockchain balance may decrease by a certain amount and the user's channel balance increases by the same amount.
The action $\mathrm{Withdraw}$ describes a full withdraw by defining that a user's channel balance may be set to 0 and, for an honest user, the user's blockchain balance must increase by at least the user's channel balance. For dishonest users, it is only guaranteed that a dishonest user's blockchain balance does not decrease.
The action $\mathrm{Pay}$ of the module IdealUser describes the execution of a set of payments by defining that the sending users' channel balances are decreased by the amounts of sent payments and the receiving user's channel balances are increased by the respective amounts. Payments can be aborted keeping channel balances unchanged.
Intuitively, one expects from a secure payment network that the sender of a payment sees the payment as sent (and the payment's balance deducted from the user's channel balance) only if the receiver of the payment sees the payment as received.
This condition is enforced by the action $\mathrm{Pay}$ in the module IdealPayments which defines that a payment can only be sent if the receiver of the payment has already received the payment or receives the payment in the same step.
Further, the module IdealUser defines a fairness condition ensuring that the system does not terminate before all users have been paid out.

In the next section, we describe how we verify that Lightning implements the security property as defined above.
In our formalization of Lightning, the $\mathrm{BlockchainBalances}$ variable is refined as the sum of unspent transaction outputs on the blockchain that a user can exclusively spend.
In the view of each user, the state of a payment changes from \textsc{new} to \textsc{processed} when the corresponding HTLC is fulfilled.
Because we verify that using Lightning with this definition of a processed payment fulfills the security property, we also clarify how Lightning can securely be used:
Honest users who sell a product in exchange for a payment can securely deliver the product once they have fulfilled the corresponding incoming HTLC and do not need to wait for the fulfilled HTLC to be removed from the channel.

Because our specification of the protocol allows for adversarial behavior, the result that the protocol specification implements the secure payment system means that no modeled adversarial behavior can break the security property, i.e., the countermeasures implemented in the protocol are sufficient.

\section{Model Checking the Security of Lightning}
\label{sec-proof-sketch}

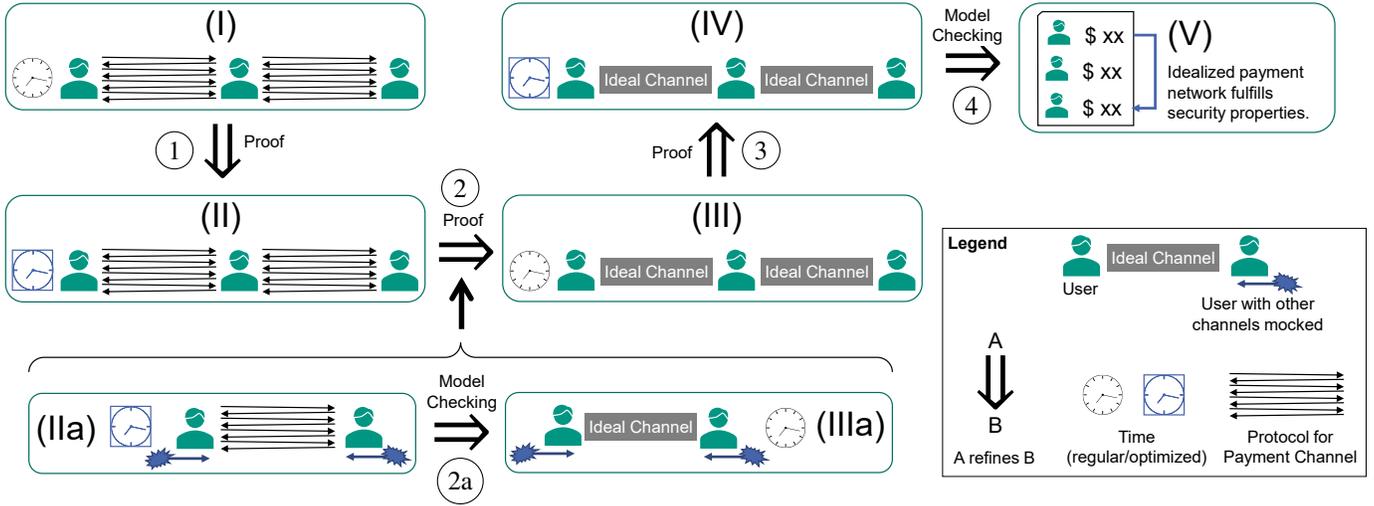
\begin{figure*}
\begin{tikzpicture}
\node[inner sep=0pt] () at (0,0)
    {\includegraphics[width=\linewidth]{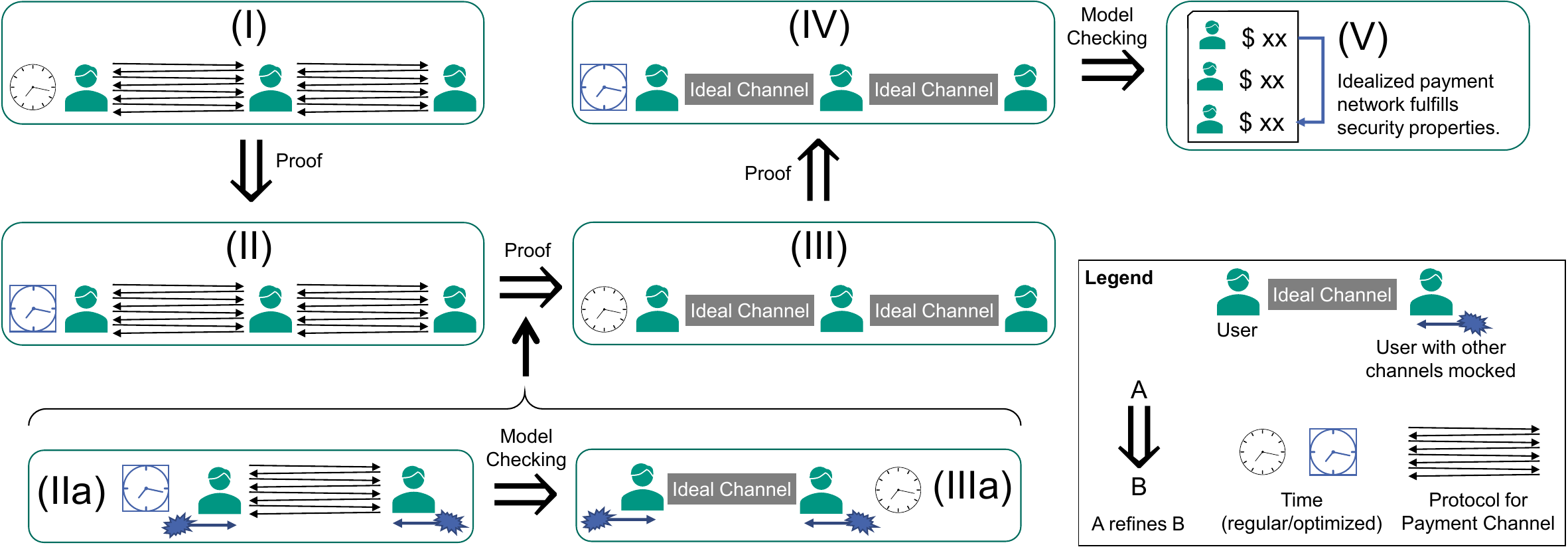}};
\node[shape=circle,draw,inner sep=2pt] (1) at (-6.8,1.2) {1};
\node[shape=circle,draw,inner sep=2pt] (2) at (-3,.7) {2};
\node[shape=circle,draw,inner sep=2pt] (2) at (-3,-3.2) {2a};
\node[shape=circle,draw,inner sep=2pt] (3) at (1,1.2) {3};
\node[shape=circle,draw,inner sep=2pt] (4) at (3.8,1.8) {4};
\end{tikzpicture}
\caption{Structure of the stepwise refinement to show that the Lightning protocol $(I)$ refines the security property $(V)$.
Specification $(I)$ is presented in \cref{sec-lightning}, specification $(V)$ in \cref{sec-security-property}, and the intermediate specifications in \cref{sec-proof-sketch}.
The correctness of each step is either proven or model checked.
Additionally, we check the refinement steps as well as the whole stepwise refinement using simulation, i.e., non-exhaustive model checking (see \cref{sec-model-checking-results}).
}
\label{fig-img-proof-sketch}
\end{figure*}

In this section, we give an overview of how we make it possible to model check that the TLA\textsuperscript{+} specification of Lightning (see \cref{sec-lightning}) fulfills the security property (see \cref{sec-security-property}).
To this end, we construct a stepwise refinement from the specification of Lightning to the security property.
The structure of this stepwise refinement is graphically shown in \cref{fig-img-proof-sketch}.

Our goal is to model check that the TLA\textsuperscript{+} formalization of Lightning~fulfills the security property, i.e., that specification $(I)$ implements specification $(V)$.
The idea of the stepwise refinement is to use multiple abstractions to show that, transitively, the specification of Lightning implements the security property. 
The structure of our approach as well as of this section is as follows.
Focusing on the first abstraction \Circled{1} in \cref{sec-proof-improved_model_of_time}, we prove that specification $(I)$ implements specification $(II)$ in which the progress of time is modeled more efficiently.
In \cref{sec-proof-abstract-protocol-steps}, we abstract from the details of the payment channel protocol by defining specification $(III)$ in which channels are updated in idealized steps (see \Circled{2}).
We prove the abstraction \Circled{2} by showing that each individual channel in specification $(II)$ refines a channel in specification $(III)$.
These individual channels are modeled in specifications $(IIa)$ and $(IIIa)$ and we check the refinement \Circled{2a} between these specifications by model checking.
In \cref{sec-proof-refine-sec-property}, after another abstraction \Circled{3} to optimize the modeling of time, we use model checking to check (see \Circled{4}) that specification $(IV)$ implements the security property defined in specification $(V)$.
In this section, we present the ideas behind the proofs and briefly sketch the proofs. The full proofs can be found in \cref{sec-appendix-proofs,sec-appendix-generalized-time-skip,sec-appendix-time-skip-proof,sec-appendix-time-skip-application-I-II,sec-appendix-time-skip-III-IV}.

\subsection{Improved Model of Time}
\label{sec-proof-improved_model_of_time}
The protocol as specified in specification $(I)$ is too complex even for model checking because of the large number of possible states of the protocol.
One reason for the huge state space is the modeling of time.
In Lightning, time is defined by the height of the blockchain.
There are many states that only differ by their value of time and are equivalent in the sense that they have the same futures.
For example, consider a state~$s$ in which the value of time is $1$ and there exists a single HTLC with a timelock of $10$.
Consider a state $s'$ that is equal to state $s$ except that the value of time in state $s'$ is $2$.
Further, assume that the only conditions in the specification that depend on the value of time are conditions that verify whether an HTLC's timelock has passed.
Then, the same future states can be reached from states $s$ and $s'$ because, for a condition whether the HTLC's timelock has passed, it does not matter whether the value of time is $1$ or $2$.
Consequently, it suffices to consider only one of the states $s$ and $s'$ during model checking.
This equivalence of the states $s$ and $s'$ is captured by the following definition of bisimilarity \cite{baier_principles_2008} for explicit real-time specifications as defined in \cref{sec_fundamentals_explicit_real_time_specifications}.

\begin{definition}[Bisimilarity]
A \emph{bisimulation} is a binary relation 
$\mathcal{R}$ over states so that for all $(s, s') \in \mathcal{R}$ it holds that:
\begin{enumerate}
  \item If $\astep{s}{t}{A}$ for some state $t$ and an action $A$, then $\astep{s'}{t'}{A}$ for some state $t'$ and $(t, t') \in \mathcal{R}$.
  \item If $\astep{s'}{t'}{A}$ for some state $t'$ and an action $A$, then $\astep{s}{t}{A}$ for some state $t$ and $(t, t') \in \mathcal{R}$.
\end{enumerate}
If $(s, s') \in \mathcal{R}$, we say that the states $s$ and $s'$ are \emph{bisimilar}.
\end{definition}

Intuitively, two states $s$ and $s'$ are bisimilar if for every step that is possible from state $s$ there exists a matching step of the same action starting at state $s'$ and the reached states are bisimilar.
As we defined the \textit{AdvanceTime} action of an explicit real-time specification so that it increases the clocks by any non-negative integer value, an AdvanceTime step from state $s$ to state $t$ might be matched by an AdvanceTime step from state $s'$ to state $t'$ that increases time by a different value.
In the area of timed automata \cite{alur_automata_1990}, this notion of bisimilarity defined above is usually referred to as untimed bisimilarity \cite{alur_observational_1994} or time-abstracting bisimilarity \cite{tripakis_analysis_2001}.

Prior work (e.g. \cite{alur_automata_1990,tripakis_analysis_2001}) has proposed to improve model checking of timed automata by exploiting this type of bisimilarity.
The idea behind the optimization is to group all states that are bisimilar in an equivalence class, referred to as a \textit{zone}.
A zone graph is constructed by connecting zone $z_1$ to zone $z_2$ if zone $z_1$ contains a state from which a step to a state of zone $z_2$ exists.
During model checking, it suffices to explore the zone graph instead of the possibly much larger original state graph.
In practice, the zone graph can be explored by considering only one representative state of each zone during model checking.
E.g., we define $AdvanceTime$ for exploring the zone graph to advance to the next zone by advancing the value of time to the lowest time of the states in the next zone.

\begin{figure}
\includegraphics[width=\linewidth]{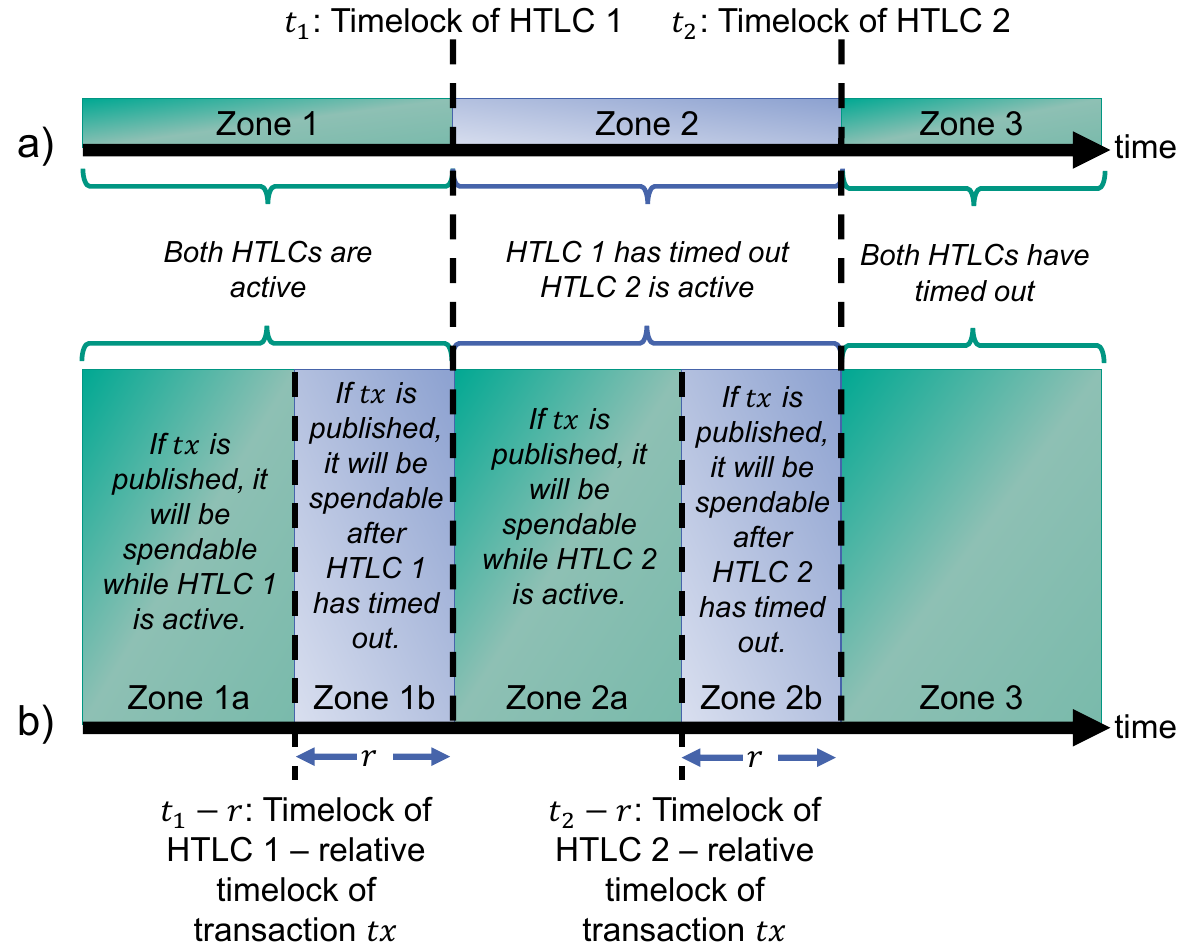}
\caption{
Zones in which bisimilar states are grouped in a scenario a) with two HTLCs with timelocks $t_1$ and $t_2$ and b) with two HTLCs and an unpublished transaction $tx$ with an output that has a relative timelock $r$.}
\label{fig-zones}
\end{figure}

In a simplified specification of the Lightning protocol without on-chain transactions, the only time-dependent conditions would check whether the timelock of an HTLC has passed.
Thus, states with time values on the same side of all HTLCs' timelocks would be in the same zone.
In a scenario with two HTLCs, states that are equal except for the value of time could be grouped into three zones depending on whether a state's value of time is below, between, or higher than both HTLCs' timelocks (see \cref{fig-zones}a).
In Lightning however, transactions can be published with outputs carrying a relative timelock that allows the output to be spent only a certain time after the transaction has been published.
Imagine a scenario with two HTLCs with the timelocks $t_1$ and $t_2$ and one unpublished transaction $tx$ that has a relative timelock $r$.
A state in this scenario will be in one of five zones (see \cref{fig-zones}b) because zones have to be split to distinguish different futures:
If the value of time in the state is below $t_1$ then it has to be distinguished whether the time is below $t_1 - r$ or above $t_1 - r$ because, if the transaction $tx$ is published before $t_1 - r$, there will be a future state in which the transaction $tx$ is spendable before time $t_1$.
However, if the transaction $tx$ is published after $t_1 - r$ there will be a future state in which the HTLC with the timelock $t_1$ has timed out and the transaction $tx$ cannot be spent yet.
This example shows that, to distinguish in which order the relative timelocks of transaction outputs and the absolute timelocks of HTLCs are reached, a higher number of zones is required.
The number of zones in the specification of Lightning is even higher than indicated in the simplified example above because, if multiple transactions with relative timelocked outputs can be published, the order of how these transaction outputs become spendable needs to be encoded in the zone graph, too.
We developed a refinement mapping that maps bisimilar states to the representatives of their respective zones.
Using this mapping, the state space is reduced. However, we found this reduction still not to be sufficient to allow for efficient model checking.

A source of complexity for the time-optimized specification is that the zone graph as described above encodes information about the order of how future timelocks are reached.
However, Lightning does not depend on the order of how relative timelocked outputs become spendable.
Therefore, we can further simplify the zone graph used by the time-optimized specification in the following way:
The subsequent zones of each zone are the zones in which the time has reached the next HTLC timelock or the age of a transaction has reached the relative timelock of one of the transaction's outputs.
For better comprehensibility of the specification, we model time using multiple clocks in the specification:
One variable represents the absolute value of time and, for each published transaction, a new clock is created that models the age of a transaction.
When a transaction is published, a new transaction clock is started that runs with the same speed as the clock for the absolute value of time.
Consequently, the $AdvanceTime$ action is defined in the time-optimized specification $(II)$ as follows:
Time advances to the next point in time at which a new zone for the time clock starts (i.e., HTLC timelock) or the transaction age clock of a transaction (or multiple transactions) advances to the next point in time at which a new zone for this transaction age clock starts (i.e., the relative timelock of one of the transaction's outputs).

With this change in the time-optimized specification, the state space is reduced to a greater extent than with approaches building the zone graph as the bisimulation quotient automaton (see \cite{tripakis_analysis_2001,baier_principles_2008}).
These approaches from prior work have the property that the zone graph is bisimilar to the original specification which means that for every behavior in the original specification there exists a behavior in the zone graph \emph{and vice versa}.
With our simplification of the zone graph, the time-optimized specification becomes more abstract and allows for behaviors that are not possible in the original specification.
For example, a relative timelocked output of a transaction $tx_2$ might be spendable before a relative timelocked output of a transaction $tx_1$ although transaction $tx_1$ was published before transaction $tx_2$.
For the stepwise refinement, we only need that the original specification implements the time-optimized specification or, in other words, that the original specification is similar to the time-optimized specification, i.e., that for every behavior in the original specification there exists a behavior in the time-optimized specification.

We prove that the original specification $(I)$ refines the time-optimized specification $(II)$ by defining a refinement mapping \Circled{1} from specification $(I)$ to specification $(II)$ and proving the correctness of this refinement mapping.
The proof can be found in the appendix. In \cref{sec-appendix-generalized-time-skip} we introduce notation and generalized approach for time optimization of explicit real-time specifications in TLA\textsuperscript{+}. We prove the generalized time optimization in \cref{sec-appendix-time-skip-proof}.
In \cref{sec-appendix-time-skip-application-I-II} we prove that the optimization can be applied to specification $(I)$.
The refinement mapping \Circled{1} maps a state $s$ to a state $s_\mathrm{R}$ that is the zone representative of the respective zone by setting each clock in state $s_\mathrm{R}$ to the lowest value that the clock can have in the respective zone.
In the example of \cref{fig-zones}b, a state in Zone 2a would be mapped to a state where the value of time is set to $t_1$.
The idea of the proof is to show that each step of specification $(I)$ starting in state $s$ is mapped to a step of specification $(II)$ starting in state $s_\mathrm{R}$.
An $AdvanceTime$ step starting in state $s$ is either a step within the same zone or a step that advances to a new zone.
A step that stays in the same zone is mapped to a stuttering step in specification $(II)$.
A step that advances to a new zone is mapped to a step in specification $(II)$ from $s_\mathrm{R}$ to a next zone representative.
Further, we prove for each non-$AdvanceTime$ action in the specification that all steps that are possible from state $s$ are possible  from all states that are in the same zone as state $s$.

\subsection{Abstraction of Protocol Steps in Payment Channels}
\label{sec-proof-abstract-protocol-steps}

Having applied the time optimization, the state space of specification $(II)$ is still too large to be explored by model checking.
Therefore, we further reduce the state space by consolidating effects of multiple protocol steps in idealized steps.
We implement this approach in specification $(III)$.
We prove that specification $(II)$ implements specification $(III)$ if the protocol for a single payment channel implements the specification of an idealized channel.
Thus, this abstraction step separates model checking that multi-hop payments using idealized channels refine the security property from model checking that the protocol for a single payment channel implements the specification of an idealized channel.

\partitle{Idealized Channel Specification}
The idealized channel specification omits some protocol details like messages that are exchanged between the parties, blockchain transactions, etc.
Instead, the idealized channel specification specifies mainly how the states of HTLCs are updated as these are required for modeling multi-hop payments.
For example, when opening a payment channel, the protocol specifies how messages are exchanged between the two parties; however,  in the idealized channel specification, opening is modelled as a single step.
Specifying the idealized channel represents a challenge:
On the one hand, the specification must allow all behaviors that are possible in the protocol. This includes behaviors in which one or both users are dishonest and, for example, an HTLC is timed out although the HTLC has previously been fulfilled.
On the other hand, while the specification may principally allow behaviors that are not possible in the protocol, the specification must not allow behaviors that are not possible in the protocol if these behaviors violate the security property.

\begin{figure}
\includegraphics[width=\linewidth]{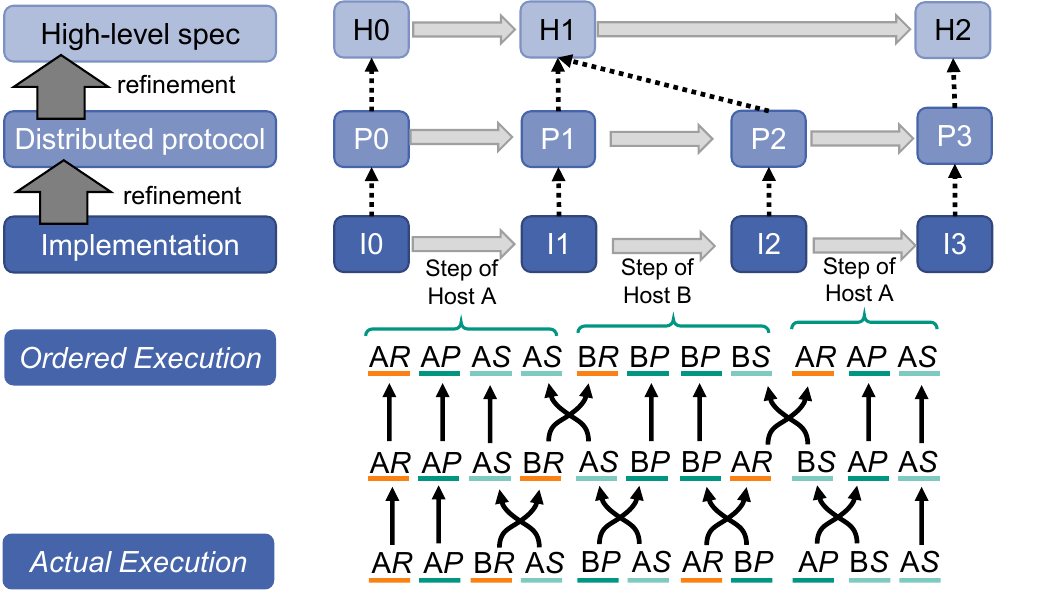}
\caption{
Overview of the IronFleet methodology.
The upper half shows the refinement steps from an implementation to a high-level specification.
Each state (I0 to I3) in a behavior of the implementation is mapped to a state of the distributed protocol (P0 to P3) which is mapped to a state of the high-level specification (H0 to H2).
The lower half shows how the low-level steps that each implementation step consists of are reordered.
Each low-level step is abbreviated where the first letter indicates the host (A or B) and the second letter what the step does.
An implementation step consists of low-level steps that receive messages (R), perform processing (P), or send a message (S).
While the steps of multiple hosts can be interleaved in an actual execution, IronFleet ensures that the steps can be reordered to an ordered execution in which contiguous low-level steps match implementation steps.
Figure from \cite[Figures 3 and 7]{hawblitzel_ironfleet_2015} (redacted).
}
\label{fig-ironfleet}
\end{figure}

\partitle{IronFleet}
A related idea has been used by Hawblitzel et al. \cite{hawblitzel_ironfleet_2015,hawblitzel_ironfleet_2017} for the IronFleet methodology to verify distributed systems.
In the following, we present ideas used by IronFleet and explain how our approach relates to these ideas.
Given an implementation for a distributed system written in Dafny \cite{leino_dafny_2010} and a high-level specification, IronFleet is a method to automatically create a machine-checked proof that the implementation meets the high-level specification.
IronFleet separates a distributed system into three layers: The high-level specification, a distributed protocol specification, and the implementation (see \cref{fig-ironfleet}).
IronFleet proves that the implementation refines the high-level specification by proving that the implementation refines the distributed protocol specification and that the distributed protocol specification refines the high-level specification.
The high-level specification specifies a centralized state machine that describes the expected externally visible behavior of a distributed system.
The distributed protocol specification defines a distributed state machine that runs on each host of the distributed system.
To keep the distributed protocol specification simple, the specification is abstract, e.g., it uses unbounded integers and unbounded sequences, and it is assumed that
each step is an atomic step in which messages are read from the network, the state is updated and messages are sent to the network.
The implementation layer is defined by imperative code that runs on each host.

\begin{figure*}
\includegraphics[width=\linewidth]{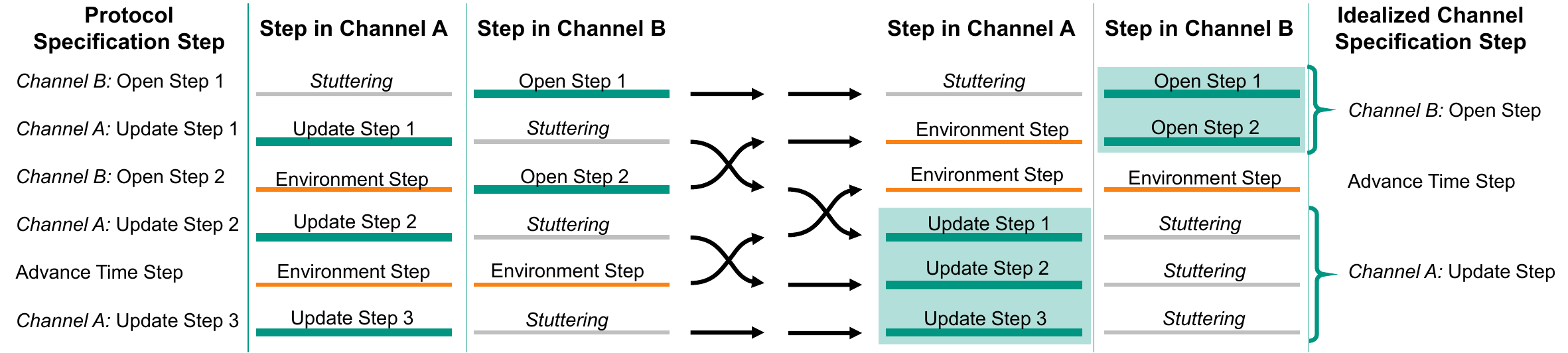}
\caption{
In the protocol specification $(II)$, steps for opening and updating two channels might be interleaved (left column). The refinement mapping \Circled{2} from specification $(II)$ to $(III)$ abstracts such interleaved steps to idealized steps (right column). Thereby, the specification's complexity is significantly reduced.
The second and third columns show how each protocol step is categorized for both channels as either a step in this channel (green), an environment step (orange) that is a step of time or another channel but has an effect on the channel, or a stuttering step (gray) that has no effect on the channel. Following the reordering rules, these steps can be reordered so that in each channel multiple steps are contiguous. These contiguous steps are replaced by an abstract step of an action of the idealized channel specification $(III)$ (right column).
Note that a reordering is shown also in \cref{fig-ironfleet} but the representation here is rotated.
}
\label{fig-interleavings}
\end{figure*}

To prove that the implementation of the distributed system refines the distributed protocol specification, Hawblitzel et al. first prove that the implementation for a single host refines the state machine specified by the distributed protocol specification.
Because the distributed protocol specification requires atomic steps where a single atomic step consists of receiving a message, updating the host's state, and sending a message, Hawblitzel et al. assume for this proof that each host step in the implementation is also such an atomic step.
In an actual execution, such implementation steps are not guaranteed to be atomic and the low-level steps of implementation steps of different hosts might be interleaved.
IronFleet uses an informal reduction argument to demonstrate that, for every actual execution, there exists an equivalent execution in which the low-level steps of each host's implementation step are contiguous as in the atomic step.
The reduction argument of Hawblitzel et al. is based on the insight that a reordering of events is valid if after reordering `(1) each host receives the same packets in the same order, (2) packet send ordering is preserved, (3) packets are never received before they are sent, and (4) the ordering of operations on any individual host is preserved' \cite[Section 3.6]{hawblitzel_ironfleet_2015}.
An example for such a reordering is shown in \cref{fig-ironfleet}.
In the following, we refer to these four rules as IronFleet reordering rules.
Using the proof that the implementation for a single host refines the state machine specified by the distributed protocol specification, Hawblitzel et al. prove that a distributed system composed of multiple host implementations refines the distributed protocol of multiple hosts \cite[Section 3.5]{hawblitzel_ironfleet_2015}.

\partitle{Relation of our approach to IronFleet}
The general approach of IronFleet is similar to our approach as our stepwise refinement from the security property $(V)$ to a detailed protocol specification $(I)$ corresponds to the layered proof of Ironfleet from a high-level specification to an implementation. The distributed protocol layer in IronFleet corresponds to our intermediate specification $(III)$ in which we abstract from protocol details.
In a distributed system as considered by IronFleet, each implementation host step consists of multiple lower-level steps that are independent of other hosts and hosts communicate by exchanging messages over the network.
In Lightning, payment channels are updated by multiple protocol steps that are independent of steps in other payment channels.
Although it might seem counter-intuitive at first sight, we identify payment channels (and not users) with hosts in the IronFleet methodology.
A payment channel is a two party protocol that does not `communicate' with other channels by exchanging messages but two payment channels affect each other by accessing the same shared variables if there is a user who participates in both payment channels.
Therefore, we identify shared variables of users who participate in multiple payment channels with the network in the IronFleet methodology.

We explain how we use the reduction and refinement ideas to abstract specification $(II)$ at an example:
\Cref{fig-interleavings} shows a simplified excerpt from a behavior in which a channel $B$ is opened, time is advanced, and a channel $A$ is updated.
In the idealized channel specification $(III)$, these steps are only three steps (see right side of \cref{fig-interleavings}).
In the protocol, the abstract steps for opening and updating a channel take multiple protocol steps which might be interleaved in an actual execution (see left side of \cref{fig-interleavings}).
We developed a refinement mapping \Circled{2} that maps protocol states to states of the idealized channel specification by implicitly applying two transformations:
Steps are reordered to get contiguous steps in a channel and the resulting batches of contiguous steps are replaced by idealized channel steps.
We check the correctness of the refinement mapping using model checking (see \cref{sec-model-checking-results}).

In the following, we explain the ideas behind the construction of the refinement mapping \Circled{2} that maps the protocol to idealized channels.
By definition of the protocol specification, each step in the protocol is a step in one of the payment channels or a step that advances time.
In the second and third column, \cref{fig-interleavings} shows how each protocol step is seen by each of the two channels in the example.
A protocol step of channel A is seen by channel A as the respective protocol step.
Depending on whether this protocol step has an effect on other channels, the step might be seen by another channel B as either a stuttering step (i.e., a step that leaves all variables of the channel unchanged) or as step that changes some variables of channel B.
We refer to the later steps as environment steps because, from the perspective of channel B, the variables of the channel B are changed not by the channel itself but by the channel's environment.
A step that advances time is an environment step in all channels.
For the reordering of steps, our refinement mapping considers the IronFleet reordering rules. However, as these rules are formulated for a message passing model, they need to be adjusted for our shared variable model.
The first rule `each host receives the same packets in the same order' can be rephrased as `all environment steps are in the same order'.
The third rule `packets are never received before they are sent' is not necessary because each step that changes a variable used by multiple channels is simultaneously both a sending step from the perspective of the channel performing the step and a receiving step from the perspective of an inactive channel whose variables are changed by the step.
The second rule `packet send ordering is preserved' and the fourth rule `the ordering of operations on any individual host is preserved' can be rephrased as `all protocol steps of a channel are in the same order'.
Consequences of these reordering rules are that, from the perspective of one channel, a stuttering step can be swapped with any other step and that a protocol step and an environment step can be swapped unless the protocol step affects another channel, i.e., the step is an environment step in another channel.
The reordering strategy used by our refinement mapping is that environment steps are swapped to happen before protocol steps until a protocol step is reached that ends a batch of protocol steps that is replaced by a single idealized channel step.

To verify the correctness of the refinement mapping, we use model checking and a proof (see \cref{sec-appendix-proofs}).
Comparable to the proof of the IronFleet methodology \cite{hawblitzel_ironfleet_2015}, we first check for a single payment channel that the protocol specification implements the idealized channel specification with a single channel.
We define specifications $(IIa)$ and $(IIIa)$ that describe the behavior of only a single payment channel but include a module that mocks the channel's environment, i.e. the module specifies the effects of environment steps that other channels can have on the specified channel.
Due to the environment mocking module, specifications $(IIa)$ and $(IIIa)$ describe all the steps the specified payment channel can take when the payment channel is part of a system with other payment channels.
We prove that the environment module mocks all possible steps by other channels in \cref{sec-appendix-proof-II-III}.
We specified a complex refinement mapping \Circled{2a} from specification $(IIa)$ to specification $(IIIa)$ for the reordering and abstraction of protocol steps to idealized channel steps.
We checked the refinement mapping \Circled{2a} using model checking (see \cref{sec-model-checking-results}).
We prove the refinement mapping \Circled{2} in the following way:
We define a refinement mapping \Circled{2} from specification $(II)$ to specification $(III)$ that uses the refinement mapping \Circled{2a} to map each individual channel from the protocol to an idealized channel.
We prove that by mapping each state of specification $(II)$ with the refinement mapping \Circled{2} to a state of specification $(III)$, each step of specification $(II)$ is mapped to a state of specification $(III)$.

\subsection{Refinement of Security Property}
\label{sec-proof-refine-sec-property}
To facilitate the refinement mapping \Circled{2} from specification $(II)$ to specification $(III)$, specification $(III)$ is defined as a real-time specification in which time can advance by arbitrary integer numbers instead of using the improved model of time.
To allow for an efficient model checking, we apply the same optimization for time as used above by defining specification $(IV)$ where bisimilar states are grouped in a zone (see \cref{sec-proof-improved_model_of_time}).
By a proof (see \cref{sec-appendix-time-skip-III-IV}) analogously to the proof of \cref{sec-proof-improved_model_of_time}, specification $(III)$ implements specification $(IV)$ and, by transitivity, specification $(I)$ implements specification $(IV)$.

Finally, we can show that specification $(IV)$ using idealized channels implements the idealized functionality defined in specification $(V)$.
We check this refinement by defining a refinement mapping \Circled{4} and checking the refinement mapping using model checking (see \cref{sec-model-checking-results}).
The definition of the refinement mapping \Circled{4} is straightforward because all variables of specification $(V)$  are also variables of specification $(IV)$.

In the following section, we present our results of model checking.
While the refinements introduced above significantly reduce the state space to explore, still only small finite models can be checked completely.
To the extent that we could verify the refinements by model checking, we conclude from the refinement steps described above that specification $(I)$ implements specification $(V)$, i.e. the specification of Lightning is an implementation of an idealized functionality for a payment network and fulfills the security property.

\section{Results of Model Checking}
\label{sec-model-checking-results}

We check the refinement mappings \Circled{2a} and \Circled{4} using model checking.
Additionally, we check our manual proof steps \Circled{1}, \Circled{2} and \Circled{3} and the whole stepwise refinement by simulation.

\setlength{\extrarowheight}{2pt}
\begin{table}
    \caption{Model checking of refinement mapping \Circled{2a} from specification $(IIa)$ to specification $(IIIa)$}
    \label{tab_model_checking_IIa_IIIa}
    \begin{tabularx}{\linewidth}{lXrr}
        \hline
        ID & Model & \# States & Runtime \\
        \hline
        C1 & Payment from user A to user B & $\sim$ $10^5$ & $\sim$ 3\,min \\
        C2 & Payment from user A over B to C & $\sim$ $10^5$ & $\sim$ 11\,min \\
        C3 & Payment from user C over A and B to D & $\sim$ $10^5$ & $\sim$ 11\,min \\
        C4 & Two payments: Payment from user A to B and payment from user B over A to C & $\sim$ $10^7$ & $\sim$ 9\,h \\ 
        C5 & Two concurrent payments from user A to B & $\sim$ $10^9$ & $\sim$ 1\,mo \\ 
        \hline
    \end{tabularx}

\end{table}

To model check the refinement mapping \Circled{2a} from specification $(IIa)$ to specification $(IIIa)$, we use the explicit state model checker TLC that explores all reachable states, calculates the refinement mapping on these states and verifies that the mapped states and steps fulfill specification $(IIIa)$.
Specification $(IIa)$ models two users and a payment channel and is parameterized by the information about the context of this payment channel, i.e., the other users in the payment channel network, and the payments to be processed.
As there are infinitely many possible ways to parameterize specification $(IIa)$, we only check a small selection of models that we deem representative.
We model check five different models that are listed in \cref{tab_model_checking_IIa_IIIa}.
To give an impression, the table also shows the magnitude of the number of distinct states that were explored and the time used by TLC (measured on a server with 40 physical cores).
In each model, the execution starts with two users (A and B) prepared to open a payment channel and TLC explores all possible behaviors of the two users to open the channel, communicate with mocked users where applicable, process payments, and close the channel. Each checked behavior ends with the channel being closed and the two users having their funds paid out on the blockchain.
The simplest model listed in \cref{tab_model_checking_IIa_IIIa} is a payment from user A who funded to channel to the other user.
Models C2 and C3 are models in which the channel between users A and B is an intermediate hop on a payment that includes mocked users.
Model C4 models two payments: A payment from user A to user B and a payment that user B sends to user C over user A as an intermediate.
There are many more states to explore in model C4 as for the previous models as the two payments can partially interleave: User B has received the payment from A to B by fulfilling the corresponding HTLC. While the fulfilled HTLC is removed, User B can already send the payment to user C.
Model C5 models two payment from user A to user B which is an even larger model as the two payments can interleave from the beginning.
By taking about a month to model check, this model is at the limits of what we can model check in reasonable time.
The main factor for this long runtime are states in which both users are dishonest. Starting from such states, many more behaviors are possible than from states in which at least one user is honest.

\begin{table}
    \caption{Model checking of refinement mapping \Circled{4} from specification $(IV)$ to secure payment system $(V)$}
    \label{tab_model_checking_IV_V}
    \begin{tabularx}{\linewidth}{Xrr}
        \hline
        Model & \# States & Runtime \\
        \hline
        Payment from user A over B to user C & $\sim$ $10^5$ & $\sim$ 1\,min \\
        Two payments: Payment from user A over B to C and payment from user C over B to A & $\sim$ $10^7$ & $\sim$ 45\,min \\
        Two concurrent payments: Payment from user A over B to C and payment from user A to  B & $\sim$ $10^7$ & $\sim$ 1\,h \\
        Three payments: Payment from user A over B to C, payment from user B to A, and payment from user B to C & $\sim$ $10^8$ & $\sim$ 13\,h \\
        Payment from user A over B and C to user D & $\sim$ $10^8$ & $\sim$ 12\,h \\
        \hline
    \end{tabularx}
\end{table}

The models that we model check to verify the refinement mapping \Circled{4} from specification $(IV)$ to specification $(V)$ are listed in \cref{tab_model_checking_IV_V}.
In all models except the last one, we model three users (A, B, and C) and two payment channels: one channel between user A and user B and the other between user B and user C.
In the last model, we model four users (A to D) and three payment channels so that a payment from user A to user D is possible.
In all these models, we model check actual multi-hop payments and thereby verify that Lightning implements specification $(V)$, the idealized functionality of a payment network.

Using TLC to explore the whole state space of larger models than the models we have just described becomes impractical.
However, we can partially verify larger models by using TLC's simulation mode in which the model checker starts in an initial state and chooses each next state randomly.
Recent work has shown that using simulation as a `lightweight' verification where more rigorous methods are not practical can be successful at finding critical flaws \cite{howard_smart_2024}.
Using simulation, we verify the abstractions \Circled{2a} and \Circled{4} for larger models.
We use also simulation to verify our structured proofs for the abstractions \Circled{1}, \Circled{2}, and \Circled{4} as well as to verify the whole proof that specification $(I)$ refines the secure payment system in specification $(V)$.

\section{Discussion}
\label{sec-discussion}

In this section, we review our approach and discuss limitations and future work.

\subsection{Choice of Formalization Language and Tools}

Protocol verifiers such as Tamarin \cite{meier_tamarin_2013} and ProVerif \cite{blanchet_modeling_2016} have successfully been used for unbounded verification of a number of security protocols \cite{basin_tamarin_2022}.
These verifiers could also be used to model aspects of Lightning and reason about the protocol's properties without the limitations of finite model checking.
These tools support modeling cryptographic primitives and allow for stronger adversary models.
However, it is challenging to model natural numbers with addition and subtraction of variables (see \cite[page 35]{the_tamarin_team_tamarin-prover_2024} and \cite[page 18]{blanchet_security_2022}) which is required for modeling blockchain transactions with amounts as in our TLA\textsuperscript{+} specification.
The generality of TLA\textsuperscript{+} allowed us to model all relevant aspects of Lightning.
However, we had to abstract cryptographic primitives (see \cref{sec-lightning}).

A benefit of TLA\textsuperscript{+} is that TLA\textsuperscript{+} does not restrict the way properties are proven, whether manually, using an explicit-state \cite{web-tlc} or a symbolic model checker \cite{web-apalache}, or a tool for automated reasoning \cite{web-tlaps}.
We used the explicit-state model checker TLC.
The automated model checking process and the generation of counterexamples facilitated our process of defining the intermediate specification $(III)$ and the complex refinement mapping from specification $(II)$ to specification $(III)$.
Because we used model checking, we could only verify the security for a number of four users.
Theoretically, there might be attacks that only apply when there are more than four users; these attacks would not be discovered by our approach.
A further consequence of the choice for model checking was that we had to restrict the adversary model to restrict the messages that an adversary can send.
However, the TLA\textsuperscript{+} specification could also be verified using unbounded verification with a theorem prover \cite{web-tlaps}.
Currently, writing such a proof seems to be too effortful. However, it will be facilitated by future advancements in assistance and automation for theorem proving.

\subsection{Limitations and Future Work}
While our work is a step towards a complete formal verification of Lightning, there are limitations not just with respect to the verification methodology as discussed above but also with respect to the model of the protocol.
To formalize Lightning, we left out aspects that are not required for security. E.g., the fees that the sender of a payment pays to intermediate hops, key rotation and onion-routing for increased privacy, and route finding for multi-hop payments.
Also, the model of the blockchain could be augmented by considering reorganizations and delays for transaction inclusion.
Further, we assume that all parties participating in the protocol are known from the beginning and that channels are opened and then closed. We do not model that channels are reopened and that new parties join the network.

Our adversary model restricts the capabilities of an adversary by disallowing the sending of messages with arbitrary content and the exchange of information between adversarial users.
While we had to make these restrictions to keep the specification's state space at a manageable level, follow up work could find optimizations that allow for making the adversary stronger.

While this work was in progress, the official Lightning specification was extended to allow for dual-funding of channels, i.e. both parties may deposit coins into a channel during opening (see \cite[Channel Establishment v2]{web-lightning-bolt2}).
The method we presented in this work can be used to formalize and analyze the security of this advancement as well.

The property that we model checked in this paper focuses on security.
Future work could also include other properties in the idealized functionality and adapt the proof.
For example, it could be shown that, assuming honest and cooperating users and timely delivery of messages, payments are guaranteed to succeed.

\subsection{Known Attacks on Lightning}

While we prove that Lightning is secure, prior work has identified several attacks on the assumptions of Lightning and properties that are not included in our security definition.
Several works discuss griefing \cite{perez-sola_lockdown_2020,mizrahi_congestion_2021-1,rohrer_discharged_2019-1,lu_general_2022} and other denial-of-service attacks \cite{tochner_route_2020} in which no funds are stolen but the regular operation is disturbed.
In the wormhole attack \cite{malavolta_anonymous_2019,tikhomirov_quantitative_2020-2}, an attacker reroutes a payment and receives the fees intended for other intermediate hops, however, the actual amount of the payment is unconcerned.
Extending the TLA\textsuperscript{+} formalization of Lightning with fees would allow for modeling the wormhole attack.
Because fees would also needed to be considered in the security property, this change would make the security property more complicated and, thus, we decided to leave modeling of fees out of scope for this paper.
Further, there are attacks on privacy \cite{herrera-joancomarti_difficulty_2019-2,van_dam_improvements_2020,rohrer_counting_2020-1,kappos_empirical_2021-1,romiti_cross-layer_2021-1,kumble_how_2021-1,biryukov_analysis_2022,ndolo_not_2024} which is a property that is not included in the security definition used in this work.
Other works \cite{harris_flood_2020-2,riard_time-dilation_2020,nadahalli_timelocked_2021,sguanci_mass_2023} discuss the violation of the assumption that users can timely publish a transaction on the blockchain.
In practice, security flaws are also based on implementations not following the specification, e.g., by missing verification checks \cite{web-cve-2019-email}. 

\subsection{Evaluation of Protocol Modifications}

Besides proving that the formalization of Lightning fulfills the security property, the TLA\textsuperscript{+} formalization of Lightning can also be used to test proposed modifications of Lightning.
To quickly find flaws, it suffices often to model check only a subset of the specification (e.g., only a single channel) and verify just lower-level invariants and temporal properties.
While such an approach cannot prove that a modification of Lightning is secure, it can accelerate protocol development by providing a short feedback loop to developers.
To evaluate this idea, we introduced flaws by adapting the formalization of Lightning and verified that the introduced flaws are detected by model checking.
As an example, we tested whether the, so called, second-stage transactions for HTLCs can be removed by including the conditions of the outputs of HTLC second-stage transactions directly in a commitment transaction's outputs. Verification with the model checker showed within a few minutes that this makes the protocol insecure and, thus, Lightning cannot be simplified in this way.

\subsection{Connecting the Specification to an Implementation}

There exist multiple implementations of Lightning that are actively used.
While the TLA\textsuperscript{+} specification of Lightning is not an executable implementation, it can be used to validate the correctness of existing implementations.
Cirstea et al. \cite{cirstea_validating_2025} have recently shown a lightweight way to connect implementations in imperative languages to a TLA\textsuperscript{+} specification.
Their approach is to collect traces of program executions and to use TLC to check these traces against traces described by the corresponding TLA\textsuperscript{+} specification.
Transferring their approach to Lightning, is an opportunity for follow up work.

\section{Conclusion}
\label{sec-conclusion}

We have formalized Lightning and a secure payment system that captures the security property of Lightning in TLA\textsuperscript{+}.
Using stepwise refinement, we were able to model check small models which showed that Lightning implements the secure secure payment system.
Therefore, the formalization can serve as a starting point for future work towards a formally verified reference implementation of Lightning in TLA\textsuperscript{+}.
Furthermore, the approach could also be used to enable specifications of other protocols to be model checked.
In particular, the abstraction of the model of time can be generalized as well as the general approach of separating model checking of local behavior and behavior in a network.
Thus, our approach can be a valuable tool to analyze new versions of Lightning or similar protocols.
We hope that this approach contributes to making future protocols for payment channel networks and related protocols more secure.

\bibliographystyle{IEEEtran}
\bibliography{library,websites}

\appendix

\subsection{On the Formalization of \cite{kiayias_composable_2020}}
\label{sec-appendix-composable}

While working on the formalization of Lightning in TLA\textsuperscript{+}, we found the following two flaws in the formalization of \cite{kiayias_composable_2020}.
While these flaws render the formalized protocol insecure, they are easy to fix and it seems that the security proof could work for the corrected protocol.
The following references to figures and page numbers refer to the paper's version on ePrint \cite[version 20220217:205237]{kiayias_composable_2019-1}.

The first flaw concerns the punishment of the publication of an outdated commitment transaction for which the protocol is specified in Fig. 37, lines 21-25 (page 64).
A problem arises, for example, in the following situation:
Before the current time, user Alice has sent an outgoing HTLC to user Bob. The HTLC was committed and has been fulfilled. Now, the HTLC's absolute timelock has passed.
Now, Alice has an outdated commitment transaction that commits the HTLC and Alice has Bob's signature on the HTLC timeout transaction corresponding to that HTLC.
Alice is malicious and publishes this outdated commitment transaction together with the HTLC timeout transaction which is valid because the HTLC's absolute timelock has passed.
Bob runs the protocol specified in Fig. 37 and arrives at line 22.
In line 22, a revocation transaction is created whose inputs spend all outputs of the outdated commitment transaction.
In the situation described, such a revocation transaction cannot be valid because the HTLC output in the outdated commitment transaction is already spent.
Instead of an input referencing the outdated commitment transaction's HTLC output, the revocation transaction must have an input that references the output of the HTLC timeout transaction.
While the protocol as formalized in Fig. 37 is incorrect, the security proof on page 90 does not mention the case that a second-stage (timeout or success) HTLC transaction might have been published for an outdated commitment transaction and, thus, the protocol seems to be correct.
While the protocol can be corrected by adding a specification of how such cases are handled, it is difficult to detect such flaws by inspecting the proof manually.

For a scenario that shows the impact of the second flaw, assume that in the payment channel between users Alice and Bob there is currently an unfulfilled HTLC for a payment from Alice to Bob.
The HTLC's absolute timelock passes and the HTLC times out.
Bob unilaterally closes the payment channel by publishing the latest commitment transaction.
The commitment transaction contains an output for the HTLC with the spending method $pt_{\mathrm{rev},n+1} \lor (pt_{\mathrm{htlc}, n+1}, \mathrm{\mathtt{CltvExpiry}~absolute}) \lor (pt_{\mathrm{htlc}, n+1} \land ph_{\mathrm{htlc}, n+1}, \mathrm{on~preimage~of~} h)$ (see Fig. 40, line 8) where $pt$ are public keys for which Alice has the private key, $ph$ are public keys for which Bob has the private key, and $\mathtt{CltvExpiry}$ is the HTLC's absolute timelock.
Now, Alice could spend the output of the commitment transaction corresponding to the HTLC by creating a transaction with an input that uses the disjunct $(pt_{\mathrm{htlc}, n+1}, \mathrm{\mathtt{CltvExpiry}~absolute})$ because the absolute timelock has passed.
Bob holds the HTLC success transaction that was signed by Alice with the private key corresponding to $pt_{\mathrm{htlc}, n+1}$ (Fig. 43, line 13).
If Bob has the preimage for the HTLC, Bob can add the preimage to the HTLC success transaction and can spend the HTLC's output in the commitment transaction using the disjunct $(pt_{\mathrm{htlc}, n+1} \land ph_{\mathrm{htlc}, n+1}, \mathrm{on~preimage~of~} h)$ of the spending method.
However, the HTLC success transaction is also valid without the preimage as it fulfills the conditions of the disjunct $(pt_{\mathrm{htlc}, n+1}, \mathrm{\mathtt{CltvExpiry}~absolute})$ because the HTLC's absolute timelock has passed.
Because Bob published his latest commitment transaction, Alice cannot revoke the transaction and this would result in Bob receiving the amount of the HTLC without releasing (or even without having) the preimage.
One way to correct this problem is to use the possibility that the transaction model of the paper \cite[Section 12]{kiayias_composable_2019-1} allows an output to specify a list of spending conditions and an input spending this output to reference a specific spending condition.
The correction would be to transform the disjunction in Fig. 40, line 8 into a list of spending methods and add the corresponding indices to the inputs in Fig. 43, line 13.
Another way is taken by the Lightning Network's specification which uses in the output's spending method for a timeout the operator \texttt{CHECKLOCKTIMEVERIFY} that verifies that a spending transaction has a certain timelock set (\texttt{locktime}). As Bob's HTLC success transaction has the \texttt{locktime} set to 0, the success transaction cannot fulfill this spending method.
We found this flaw by model checking when we had a similar flaw in a draft of our formalization. We fixed the flaw in our formalization by modeling the \texttt{locktime} field for transactions and adding a validity condition modeling the operator \texttt{CHECKLOCKTIMEVERIFY}.

\subsection{Lightning and Formalization}
\label{sec-appendix-formalization}

In this section, we present Lightning according to the official specification \cite{web-lightning-bolts} and describe what aspects we consider for the formalization and which aspects we leave out of scope.
We also explain how we model the protocol in TLA\textsuperscript{+}.
We start by explaining how Lightning implements multi-hop payments in \cref{sec-lightning-multi-hop-payments}.
Later, we explain how payment channels are opened, updated, and closed.
We give an overview of the TLA\textsuperscript{+} formalization of Lightning in \cref{sec-lightning-formalization-overview}.
We present the basic ideas behind the protocol's formalization in TLA\textsuperscript{+} at the example of multi-hop payments in \cref{sec-lightning-formalization-multi-hop-payments}.
This includes an explanation how sending and receiving messages and the state of an HTLC is modeled.
In this section, we also discuss how random values, hash functions, transactions, and signatures are modeled.
Because TLA\textsuperscript{+} does not offer a native way to model cryptographic primitives such as hash functions and signatures, we model those primitives in an abstract way.
In the real world, Lightning relies on security properties of these primitives, e.g., that it is practically impossible to find a preimage to a given hash value or that it is practically impossible to create a valid signature for a public key without knowing the associated private key.
In the formalization, we assume that these security properties of the primitives hold, e.g. that an adversary is not capable of reverting a cryptographic hash function.
Therefore, we can model the primitives via abstractions.

\subsubsection{Multi-Hop Payments}
\label{sec-lightning-multi-hop-payments}

This section explains how Lightning implements multi-hop payments.
If a user wants to receive a payment, the user creates an invoice and sends it to the user who wants to send the payment.
The invoice (see \cref{tab_fields_invoice}) contains a payment hash for which the receiver knows the preimage and a payment secret.
Further, the invoice contains additional data to describe the purpose of the payment for example.
This additional data is not included in the TLA\textsuperscript{+} formalization because it is not required for the security of the protocol.
The model of an invoice used in the TLA\textsuperscript{+} formalization contains just the hash and the payment secret.

After the payment's sender has received the invoice, the sender extracts the payment hash and the payment secret.
The payment's sender chooses a route through the Lightning Network, i.e. a list of users that are connected through payment channels and whose last hop is the payment's receiver. For this paper, we leave route finding out of scope and assume that the sender receives the route as external input.

The payment's sender sends an `update\_add\_htlc' message (see \cref{tab_fields_update_add_htlc}) to the first hop on the path.
The `update\_add\_htlc' message contains an id for the HTLC and the HTLC's amount, hash, and timelock. Further, the message contains an `onion\_routing\_packet' that contains the information who the next hop is and an encrypted package for the next hop with the information what the hop after the next hop is and so on.
The next hop receives the `update\_add\_htlc' message and verifies that the amount and timelock are in expected ranges and that the onion\_routing\_packet can be decrypted.
If a verification fails, the channel is closed. The principle that the channel is closed when a verification fails, is used throughout Lightning.
If the hop receiving the `update\_add\_htlc' message is not the payment's receiver but an intermediate hop: When the incoming HTLC for the payment is `irrevocably committed', the intermediate hop sends an `update\_add\_htlc' message to the next hop on the route.
An HTLC is irrevocably committed for a user A if user A has received the other user's signature on the commitment transaction containing the HTLC and has received the secrets required for revocation for all transactions that do not contain the HTLC and were signed by user A. \label{explanation-irrevocably-committed}
The forwarding of the `update\_add\_htlc' message is repeated until the payment's receiver receives the `update\_add\_htlc' message.
If the hop receiving the `update\_add\_htlc' message is the payment's receiver, the receiver sends the payment's preimage in an `update\_fulfill\_htlc' message (see \cref{tab_fields_update_fulfill_htlc}) to the previous hop.
In the TLA\textsuperscript{+} formalization, the `update\_fulfill\_htlc' message contains only the preimage because the associated HTLC can be found by hashing the preimage and looking up the HTLC by its hash value.
If the hop receiving the `update\_fulfill\_htlc' message is not the payment's sender, the hop forwards the payment's preimage in an `update\_fulfill\_htlc' message to the previous hop.
The forwarding of the `update\_fulfill\_htlc' message is repeated until the payment's sender receives the `update\_fulfill\_htlc' message.
If the hop receiving the `update\_fulfill\_htlc' message is the payment's sender, the payment has been successfully performed.

If an incoming HTLC that is committed times out, a user sends an `update\_fail\_htlc' message (see \cref{tab_fields_update_fail_htlc}) to the previous hop.
While the message contains fields for the reason of failure in Lightning, the TLA\textsuperscript{+} formalization contains only the id of the failed HTLC because learning the reason for failure is relevant for debugging but not for the protocol's functionality or security.
If the hop receiving the `update\_fail\_htlc' message is not the payment's sender: When the removal of the hop's outgoing HTLC is irrecoverably committed or the appropriate HTLC timeout transaction is confirmed on-chain, the hop sends an `update\_fail\_htlc' message to the previous hop.
The forwarding of the `update\_fail\_htlc' message is repeated until the payment's sender receives the `update\_fail\_htlc' message.
If the hop receiving the `update\_fail\_htlc' message is the payment's sender: When the removal of the hop's outgoing HTLC is irrecoverably committed or the appropriate HTLC timeout transaction is confirmed on-chain, the payment is finally cancelled.

\begin{table}
    \caption{Fields of a Lightning `invoice'.}
    \label{tab_fields_invoice}
    \begin{tabularx}{\linewidth}{lXX}
        \hline
        Variable & Description & Formalization \\
        \hline
        timestamp & Unix Timestamp & Not required. \\
        p & Payment hash & `hash' field of invoice \\
        s & Payment secret & `paymentSecret' field of invoice \\
        d & Description of purpose of payment & Not required. \\
        m & Additional metadata & Not required. \\
        n & Public key of the payment's receiver. & Not required. \\
        h & Description of purpose of payment & Not required. \\
        x & Expiry time & Not required. \\
        c & Minimal delta between HTLC timelocks for last HTLC & Not required because formalized as a constant default value. \\
        f & Fallback on-chain Bitcoin address & Not required. \\
        r & Routing information & Not required. \\
        9 & Feature bits & Not required. \\
        signature & Signature of above fields & Not included. We instead assume that the receiver can verify the integrity of the invoice. \\
        \hline
    \end{tabularx}
\end{table}

\begin{table}
    \caption{Fields of `update\_add\_htlc' message.}
    \label{tab_fields_update_add_htlc}
    \begin{tabularx}{\linewidth}{lXX}
        \hline
        Variable & Description & Formalization \\
        \hline
        channel\_id & ID of the channel derived from funding transaction & Not required. \\
        id & ID of the HTLC & `id' field of HTLC \\
        amount\_msat & Amount of HTLC in millisatoshi & `amount' field of HTLC \\
        payment\_hash & Hash of HTLC & `hash' field of HTLC \\
        cltv\_expiry & Timelock of HTLC & `absTimelock' field of HTLC \\
        onion\_routing\_packet & Data for forwarding to next hop including encrypted payload for next hop. & `dataForNextHop' field of HTLC \\
        \hline
    \end{tabularx}
\end{table}

\begin{table}
    \caption{Fields of `update\_fulfill\_htlc' message.}
    \label{tab_fields_update_fulfill_htlc}
    \begin{tabularx}{\linewidth}{lXX}
        \hline
        Variable & Description & Formalization \\
        \hline
        channel\_id & ID of the channel derived from funding transaction & Not required. \\
        id & ID of the HTLC & Not required. \\
        payment\_preimage & Preimage for HTLC & `preimage' field \\
        \hline
    \end{tabularx}
\end{table}

\begin{table}
    \caption{Fields of `update\_fail\_htlc' message.}
    \label{tab_fields_update_fail_htlc}
    \begin{tabularx}{\linewidth}{lXX}
        \hline
        Variable & Description & Formalization \\
        \hline
        channel\_id & ID of the channel derived from funding transaction & Not required. \\
        id & ID of the HTLC & `id' field of HTLC \\
        len & Length of the reason field & Not required. \\
        reason & Reason why HTLC failed & Not required because the only modeled reason is timeout. \\
        \hline
    \end{tabularx}
\end{table}

\subsubsection{Formalization Overview}
\label{sec-lightning-formalization-overview}

We formalize Lightning in TLA\textsuperscript{+}.
The formalization describes all possible actions how a user of the payment channel initiates transactions or reacts to messages or events.
In its structure, the formalization of the protocol specification follows the informal specification of the Lightning Network \cite{web-lightning-bolts}.
The formalization abstracts, however, multiple implementation details and parts that are not part of the main functionality such as fees and error messages. 

We formalize Lightning in an event-based specification.
In this manner, the protocol can be implemented and this approach allows for validating that the protocol is secure for every possible order of events.

The TLA\textsuperscript{+} specification of the protocol consists of three modules:
Two modules concern the specification of actions that a user performs for the execution of the payment channel protocol:
The module HTLCUser specifies the actions concerning HTLCs for multi-hop payments, e.g., sending an invoice, creating an HTLC, fulfilling an HTLC.
The module PaymentChannelUser specifies how the payment channel is created, how the payment channel is updated when a new HTLC is added or a fulfilled HTLC is persisted, how the payment channel is closed, how an adversarial user can cheat by publishing transactions on the blockchain, and how an honest user punishes a cheating user.
For example, the module PaymentChannelUser includes actions for creating and sending a signature of a new commitment transaction to the other user, processing messages from the other user, or publishing a commitment transaction on the blockchain to close the channel.
The third module is LedgerTime, the clock that increases the current time. Time is measured in Lightning by the block count of the Bitcoin blockchain. Thus, it is represented as an integer number and increased in integer steps.
The specification puts these three modules together by instantiating the LedgerTime module, the HTLCUser module, and the PaymentChannelUser module.
While the action of the LedgerTime module is a global action to advance time, the actions of the HTLCUser module and the PaymentChannelUser modules are parameterized by a user and, if applicable, a channel and the other user in the channel.
Formally, the TLA\textsuperscript{+} specification is defined by a set of initial states and a $Next$ action that describes possible steps that can lead from one state to a new state.

The TLA\textsuperscript{+} formalization models a system that is comprised of the users of a payment channel network and the payment channels between them.
A state of the modeled system is defined by the variables that are shown in \cref{tab_variables_spec_I}.
Some of these variables describe aspects of the system in general and some model variables of a specific user of which some are for a specific payment channel of that user.
The meaning and use of these variables are explained in the following sections.

The TLA\textsuperscript{+} formalization expects as external input the three constants NameForUserID, a sequence of modeled users, $u$InitialPayments, a set of payments to be sent by user $u$, and $u$InitialBalance, the initial balance of user $u$.
Each record in $u$InitialPayments describes a payment by an id, an amount, a point in time until the payment should be processed, and a path from the sender to the recipient.
In the initial state, the variable $u$NewPayments is initialized to the value of $u$InitialPayments for a user $u$ and to each payment record fields are added for storing values later, e.g., for the hash associated with the payment or a boolean value whether an invoice for this payment has been requested.

\begin{table}
    \caption{Variables that are part of the state of the TLA\textsuperscript{+} formalization.}
    \label{tab_variables_spec_I}
    \begin{tabularx}{\linewidth}{lX}
        \hline
        Variable & Description \\
        \hline
        \multicolumn{2}{l}{Global variables} \\
        \hline
        LedgerTime & Models the current time as an integer value representing the current height of the blockchain. \\
        Messages & Set of messages of which each user can process messages that are sent to the user. \\ 
        LedgerTx & Models the blockchain as the set of all published transactions. \\
        TxAge & Models the age of published transactions, i.e. how many blocks have been created since the transaction was published. \\
        \hline
        \multicolumn{2}{l}{Variables for user $u$:} \\
        \hline
        $u$Payments & Set of payments of user $u$ for which user $u$ is either sender or recipient. Each payment has a status that indicates whether the payment is new, processed, or aborted.\\
        $u$ExtBalance & External balance of user $u$. External refers to the balance on the blockchain, i.e. outside the payment channel network. \\
        $u$ChannelBalance & Balance of user $u$ inside the payment channel network. This equals the sum of the user's balance in all payment channels. \\
        $u$Honest & Specifies whether the user $u$ is honest. The value is initially set to either true or false and does not change. \\
        $u$NewPayments & Set of records for payments that the user $u$ wants to send. The module HTLCUser processes these payments by creating HTLC records that are added to variable Vars. \\
        $u$PreimageInventory & Set of the preimages that the user $u$ knows. \\
        $u$LatePreimages & Preimages that have been received after the associated HTLC has timed out. \\
        $u$PaymentSecretForPreimage & Function that maps for each incoming payment for which an invoice has been sent the preimage to the payment secret. \\
        \hline
        \multicolumn{2}{l}{Variables for channel $c$:} \\
        \hline
        $c$Messages & Sequence of messages sent between two users of a channel. Each user can process the first message that is sent to the user. \\
        $c$UsedTransactionIds & Helper variable that stores all ids that have already been used for creating transactions. This set helps to ensure that each new transaction id is unique. \\
        \hline
        \multicolumn{2}{l}{Variables for user $u$ of each channel $c$:} \\
        \hline
        $c,u$Balance & Integer that indicates the current balance of the user $u$ in channel $c$. \\
        $c,u$State & Protocol state the user $u$ is in in the channel $c$. This implies what messages the user $u$ expects to receive next. \\
        $c,u$Vars & Record with variables of the user $u$ in the channel $c$. Most importantly Vars contains the fields `IncomingHTLCs' and `OutgoingHTLCs' that store sets that contain a record for each incoming resp. outgoing HTLC. \\
        $c,u$DetailVars & Contains variables that are only used inside the module PaymentChannelUser. \\
        $c,u$Inventory & Contains keys and a set of transactions that the user $u$ can sign and, if the transaction is valid after signing, publish on the blockchain. \\
        \hline
    \end{tabularx}
\end{table}

\subsubsection{Formalization of Multi-Hop Payments}
\label{sec-lightning-formalization-multi-hop-payments}

The actions of a user as described in \cref{sec-lightning-multi-hop-payments} are specified as actions of the HTLCUser module.
The actions in the HTLCUser module are the actions of one specific user $u$ for a specific channel $c$ and are parameterized by the id of channel $c$, the id of user $u$, and the id of the other user in channel $c$.

The action `RequestInvoice' of the module HTLCUser is enabled if there is a payment for the user $u$ in $u$NewPayments and no invoice has been requested for this payment.
The action chooses such a payment, sends a message of type `RequestInvoice' to the payment's recipient and updates the payment to store that an invoice has been requested for this payment.
\emph{Sending of a message} is modeled by appending a record to a global variable `Messages'.
The `Messages' variable contains a list of all messages that are in transit.
The appended message has a field that describes the message's type, a field that states the message's sender, and a field that states the message's recipient and additional fields for payloads.
The payload for a `RequestInvoice' message is the payment id for which an invoice is requested.

The action `GenerateAndSendPaymentHash' of the module HTLCUser models the reception of a `RequestInvoice' message and the reply of sending an invoice to the sender of the `RequestInvoice' message.
The \emph{reception of a message} is modeled by a condition that inspects the records in `Messages' that have a recipient that equals the user $u$.
If the first message that is sent to the user $u$ has the type `RequestInvoice', then the action `GenerateAndSendPaymentHash' is enabled.
The action has to draw random values for the preimage and the payment secret.
TLA\textsuperscript{+} does not offer a native way to model randomly drawn values.
In Lightning, the values for the preimage and the payment secret should be unpredictable and hard to guess for an adversary.
In the TLA\textsuperscript{+} formalization, it suffices to model the adversary in a way that the adversary cannot guess the preimage and the payment secret.
Random values drawn in Lightning for the preimage and the payment secret are most likely unique, i.e. drawing these values a second time randomly will most likely lead to different values.
This uniqueness property of drawing random values is modeled in the following way in the TLA\textsuperscript{+} formalization:
A random preimage and payment secret are deterministically derived from the payment's id for which the preimage is used by adding constants to the payment's id.
Because the ids of payments are unique, this approach ensures that the preimages and payment secrets are unique as well.
In Lightning, the receiver of a payment sends as part of the invoice the hash of the preimage to the payment's sender.
However, in TLA\textsuperscript{+} there is no native way to calculate the hash value of a variable.
The reason why Lightning uses a hash function is that it should be simple to check whether a preimage and a hash correspond to each other but it should be impossible to calculate the preimage given its hash value.
As we can restrict the adversary in the specification to not calculate preimages from hashes, we use a simple approach to model a \emph{hash function}: We use the Identity function as hash function, i.e. the hash value of a preimage equals the preimage.
With this approach, it is trivial to check whether a hash corresponds to a preimage.
The distinction whether a given value is a preimage or a hash, depends on the variable that a value is stored in. The formalization is carefully written so that a value of a hash is never written to a variable that contains a preimage. However, writing a hash value given the knowledge of a preimage is possible as this models the execution of the hash function.
The action `GenerateAndSendPaymentHash' sends a message modeling the invoice that contains the hash value and the payment secret to the sender.
Further, the preimage and the payment secret are stored in the variables $u$PreimageInventory and $u$PaymentSecretForPreimage respectively.

The request for an invoice can also be ignored without further changes.
This is modeled by the action `IgnoreInvoiceRequest'.

The reception of the invoice by the payment's sender is modeled by the action `ReceivePaymentHash'.
The payment's sender stores the hash value and the payment secret in the record that describes the respective payment and calculates the onion package that is sent with the payment along the route.
The onion package contains for each hop a value that determines the next hop, the absolute timelock for the HTLC to use with the next hop and an onion package for the next hop.
For the last hop, the onion package contains the payment secret and the payment's amount.

If the payment channel $c$ is open and ready to operate, the payment channel $c$'s state stored in the variable $c,u$State of user $u$ is `rev-keys-exchanged'.
The action `AddAndSendOutgoingHTLC' sends an `update\_add\_htlc' message if the channel is ready and there is a payment that fulfills the following conditions:
The action is enabled for user $u$ if the user $u$ is the payment's sender or if the payment's incoming HTLC has been irrevocably committed (see \cref{explanation-irrevocably-committed}).
The hash for the payment must be known and the payment's timelock must be in the future.
The next hop for this payment must be the other user of the payment channel $c$.
An HTLC with the same hash must not already exist.
The current balance of user $u$ in the channel $c$ must be at least the payment's amount.
The effects of the action are that the payment is removed from the $u$NewPayments variable, the HTLC is added to the $c,u$Vars variable, and the `update\_add\_htlc' message is sent.
This message is sent with relation to channel $c$.
We use the $c$Messages variable to model the messages exchanged between two users of a channel $c$.
Thus, the `update\_add\_htlc' message is stored in $c$Messages.

The reception of the `update\_add\_htlc' message is modeled by the action `ReceiveUpdateAddHTLC'.
This action creates a new payment record in the user's $u$NewPayments variable if the payment is to be forwarded and, otherwise, verifies that the payment secret is correct and that the HTLC has the correct amount.

If an incoming HTLC is irrevocably committed and the user $u$ is the payment's receiver, the user $u$ can fulfill the incoming HTLC.
Whether an HTLC is irrevocably committed, is encoded in the state of the HTLC:
The \emph{state of an HTLC} is initially set to \textsc{new} when the HTLC is created in the `AddAndSendOutgoingHTLC' action or the `ReceiveUpdateAddHTLC' action (see \cref{fig-HTLC-graph}).
Once the HTLC is irrevocably committed, an action of the PaymentChannelUser module sets the HTLC's state to \textsc{committed}.
If an HTLC is fulfilled, the HTLC's state is advanced to \textsc{fulfilled}. Once the HTLC is irrevocably removed, an action of the PaymentChannelUser module sets the HTLC's state to \textsc{persisted}.
If an HTLC is failed, the HTLC's state is updated to \textsc{off-timedout}. Once the HTLC is irrevocably removed, an action of the PaymentChannelUser module sets the HTLC's state to \textsc{timedout}.
To fulfill the HTLC, a user sends an `update\_fulfill\_htlc' message to the other user in the payment channel.
The sending of the `update\_fulfill\_htlc' message is modeled by the `SendHTLCPreimage' action.
The formalization models that an HTLC can be fulfilled even after the HTLC's timeout during a grace period of a fixed length.
The length of the grace period, i.e., the number of blocks to wait after an HTLC's timeout, is defined by the constant $G$ which we set to 3 by default.
Adding a grace period is suggested by the official Lightning specification but not required.
We included the grace period in the formalization because the grace period creates an interesting situation that is relevant for the protocol's security because during the grace period the HTLC can be timed out as well as fulfilled.
If an HTLC is fulfilled after its timeout, the hash of the HTLC is stored in the field `FulfilledAfterTimeoutHTLCs' of the variable $c,u$Vars.

\begin{figure}
\includegraphics[width=\linewidth]{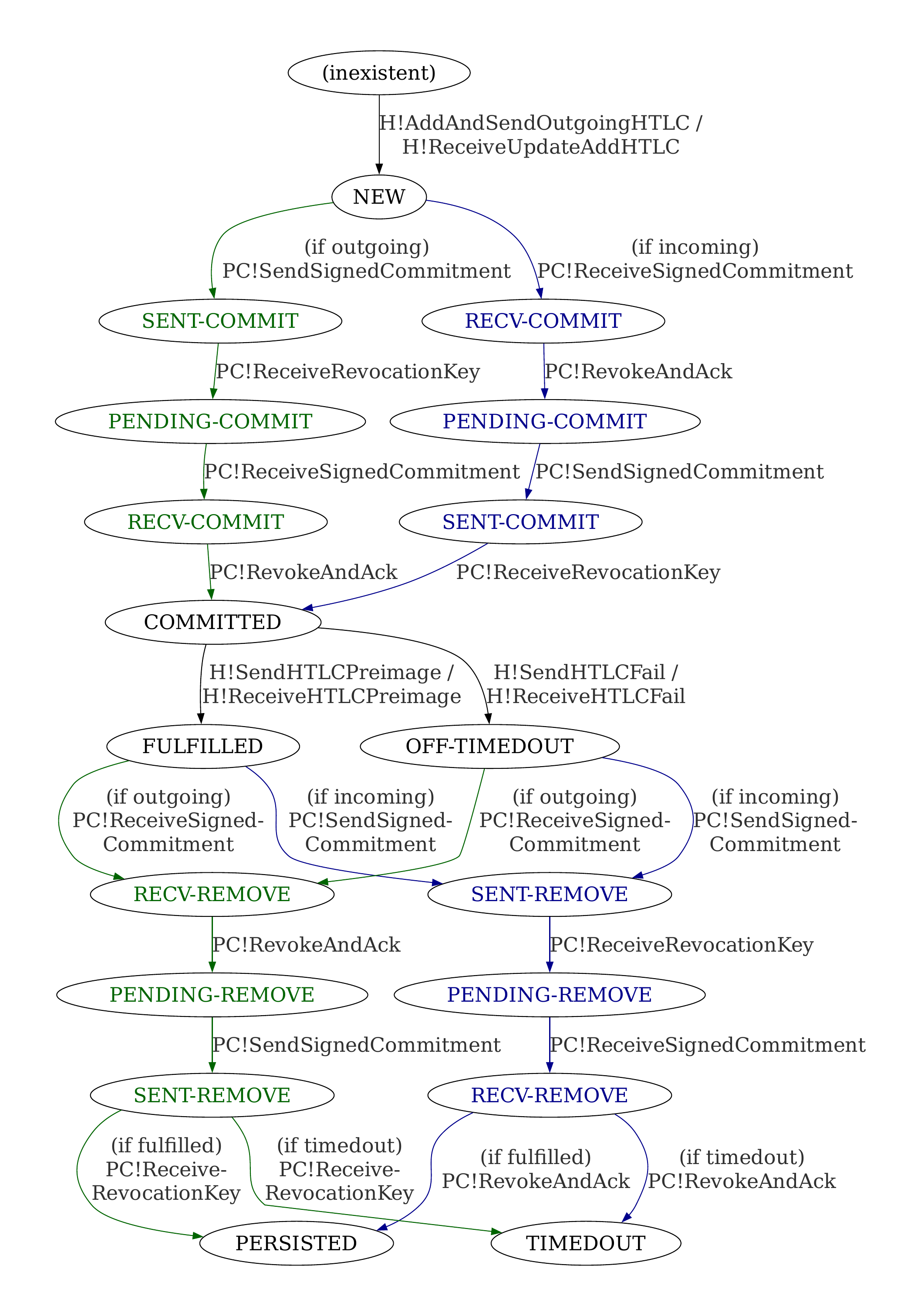}
\caption{Flow chart of HTLC states. The actions with the prefix H! are actions of the HTLCUser module (see \cref{sec-lightning-multi-hop-payments,sec-lightning-formalization-multi-hop-payments}); those prefixed with PC! are actions of the PaymentChannelUser module (see \cref{sec-lightning-opening,sec-lightning-updating,sec-lightning-closing}). Outgoing HTLCs follow the path printed in green; incoming HTLCs follow the path printed in blue.}
\label{fig-HTLC-graph}
\end{figure}

The action `ReceiveHTLCPreimage' models the reception of an `update\_fulfill\_htlc' message.
If the preimage is received after the HTLC's timeout and the grace period have passed, the payment might still be successful but it can also be aborted because the preimage reached the user too late.
The fact that the HTLC's preimage was received late is stored by adding the preimage to the set stored in the variable $u$LatePreimages.

Failing an HTLC is modeled by the action `SendHTLCFail'.
This action sends a `update\_fail\_htlc' message and updates the failed HTLC's state to \textsc{off-timedout}.
The action `ReceiveHTLCFail' receives the `update\_fail\_htlc' message and updates the state of the receiver accordingly.

\subsubsection{Keys and Funding-, Commitment- and HTLC-Transactions}
\label{sec-lightning-keys-transactions}

In this section, we present the keys that are used in Lightning and what they are used for.
To reduce the risk of being tracked by third parties, each user in Lightning has a set of private and public keys of which each key is used for one specific purpose.
Further, keys are rotated with every commitment transaction to prevent leaking information to a third party that gets to know multiple commitment transactions.
As our analysis focuses on the security property that users finally receive their correct balance and not on privacy leaks, we model these keys by a single key pair per user.

\subsubsection{Keys and Key Derivation}

Lightning makes use of multiple keys to build transactions.
Here, we explain these keys as they are used in Lightning.
For the formalization, we use a simplified model that we will present below.
Each user has the following set of keys according to the Lightning specification.
From one user's perspective, the user's own keys are prefixed by `local' and the other user's keys are prefixed by `remote'.
\begin{itemize}
  \item funding\_pubkey: A user's public key that is used to lock the output of the funding transaction. The appropriate private keys of both users are required to create a signed commitment transaction.
  \item localpubkey (remotepubkey): Public key used to lock an output that is spendable by the user (other user) in the commitment transactions created by the other user (user).
  \item local\_delayedpubkey (remote\_delayedpubkey): Public key used to lock an output that is spendable by the user (other user) in the commitment transactions created by the user (other user).
  \item local\_htlcpubkey (remote\_htlcpubkey): Public key used to lock an output that is spendable by the user (other user) in HTLC transactions.
  \item revocationpubkey: Public key used to lock an output that is spendable when the transaction should be revoked. 
\end{itemize}
Most keys are rotated with every commitment transaction so that third parties (e.g., watchtowers) cannot link multiple commitment transactions to the same channel.
To simplify the key management, the different public keys are derived from a basepoint per type of key and a point per transaction (see \cref{fig-key-derivation}\,a).
For each of the keys prefixed by local or remote, there is a basepoint for which the secret is kept by the local party.
The keys (local/remote)pubkey, (local/remote)\_delayedpubkey, (local/remote)\_htlcpubkey are derived from the from the per\_commitment\_point and from the (local/remote)\_payment\_basepoint, (local/remote)\_htlc\_basepoint, (local/remote)\_delayed\_payment\_basepoint respectively.
The private keys are derived from appropriate basepoint secrets and the per\_commitment\_point (see \cref{fig-key-derivation}\,b).
The revocationpubkey is derived from the local user's revocation\_basepoint and from the remote user's per\_commitment\_point (see \cref{fig-key-derivation}\,c).
The associated private key, called revocationprivkey, can be derived from the local user's revocation\_basepoint\_secret and the remote user's per\_commitment\_secret. The idea behind this is that both users can derive the revocationpubkey but initially no user can derive the revocationprivkey. However, the per\_commitment\_secret can be shared so that the local user can derive the revocationprivkey to spend outputs of published outdated commitment transactions.
However, as all other keys that are derived using the per\_commitment\_secret depend on an additional basepoint secret, the local user cannot derive other private keys of the remote user.
The funding\_pubkey is a regular key pair.

\begin{figure}
\centering
\includegraphics[width=0.8\linewidth]{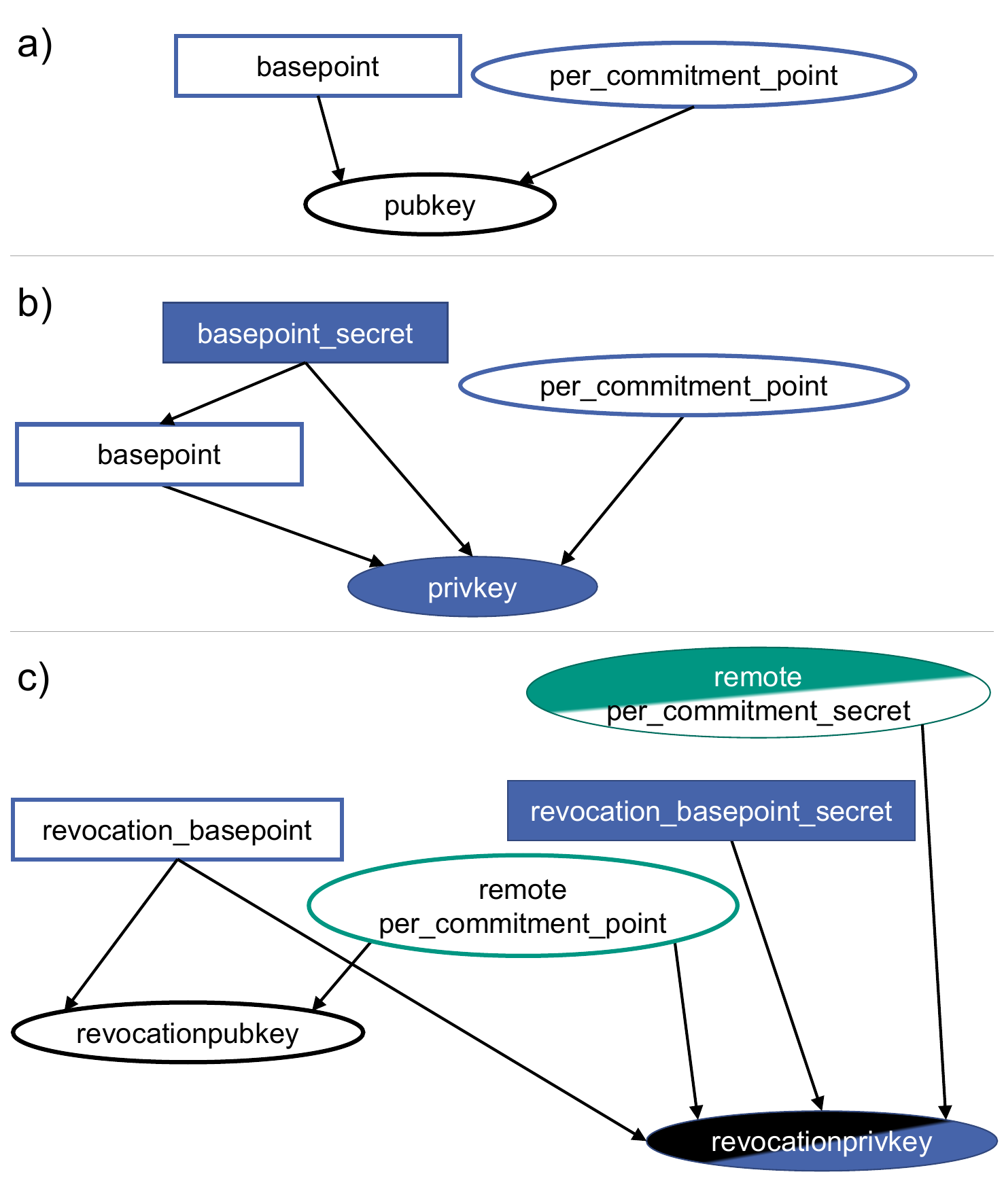}
\caption{Derivation of keys in Lightning. Each box shows a key or piece of information that is required to derive a key.
From the perspective of a user $u$, the boxes filled blue are pieces of information that only the user $u$ has.
The boxes filled green are pieces of information that only the other user has.
The boxes filled white are pieces of information that are shared between both users. The color of the boxes' borders indicates the user that generates a piece of information. A black border indicates that both users can derive a piece of information.
Arrows represent the relation `is required to derive'.
Rectangular boxes are pieces of information that are constant during the lifetime of the protocol.
Ellipses represent pieces of information that are rotated with every commitment transaction.
The remote per\_commitment\_secret and the revocationprivkey are special because the revocationprivkey cannot be generated by neither party until the remote per\_commitment\_secret is shared which enables the user $u$ to generate the revocationprivkey.
}
\label{fig-key-derivation}
\end{figure}

\subsubsection{Formalization of Keys}

To keep the specification simple, we model the funding\_pubkey and the various keys that are derived from a basepoint of a user and the user's per\_commitment\_point as the same key pair.
What is left are the revocation public keys that are derived from the revocation\_basepoint of one user and the \emph{other} user's per\_commitment\_point. These key pairs need to be rotated with every commitment transaction because the users exchange the secrets required to derive the private keys.
For modeling \emph{asymmetric key pairs}, we use a similar approach as for modeling preimages and hashes:
We use the same value for the private and the public key and distinguish between them by the name of the variable that assumes a value.
This user key which models all user specific key pairs used in Lightning is modeled as a symbolic value per user.
We model revocation keys by a record that contains a symbolic value that specifies the creator of the key and and a numerical index value that is incremented for each new commitment transaction which models the rotation of the per\_commitment\_point.
With this approach, the revocation keys can be rotated by modifying the index of the key but they are still regular key pairs.
We model the construction that the revocation keys are derived from secrets of both users by changing the conditions in transaction outputs where a revocation key is required:
A condition in the output of a transaction that requires a signature of the revocation key derived from user A's per\_commitment\_point and from user B's revocation\_basepoint is modeled as a condition that requires a signature of user A's revocation key and user B's user key.
This model leads to the following difference in transaction construction between Lightning and the formalization:
In Lightning, a transaction that can be published by user A is revocable with a signature that corresponds to the public revocation key derived from user B's revocation\_basepoint and the respective per\_commitment\_point of user A.
In the formalization, a transaction that can be published by user A is revocable with a signature that corresponds to the public revocation key of user A and the public user key of user B.
Both approaches have the same effect: Only user B can revoke a transaction published by user A and user B can revoke the transaction only after having received a secret from user A.
The message exchange differs in the following way:
In Lightning, user A sends the per\_commitment\_point to user B so that user B can derive the public revocation key.
In the formalization, user A sends the public revocation key to user B.
When revoking a transaction in Lightning, user A sends the A's per\_commitment\_secret to user B.
In the formalization, user A revokes a transaction by sending the private revocation key to user B.
\Cref{tab_keys_comparison} gives an overview of the keys used in Lightning and the corresponding formalization.

\begin{table}
    \caption{Comparison of keys in Lightning and the TLA\textsuperscript{+} formalization.}
    \label{tab_keys_comparison}
    \begin{tabularx}{\linewidth}{XX}
        \hline
        Lightning & Formalization \\
        \hline
        funding\_pubkey, localpubkey, local\_delayedpubkey, local\_htlcpubkey & User key modeled by one symbolic value per user. \\
        funding\_privkey, localprivkey, local\_delayedprivkey, local\_htlcprivkey & Modeled by the same symbolic value as the public key. \\
        payment\_basepoint, delayed\_payment\_basepoint, htlc\_basepoint and associated secrets & Not required because user key is a symbolic value instead being calculated. \\
        revocationpubkey & Modeled as record that contains a base and an index. The base specifies by a symbolic value the user that created the key. The index is a number that can be incremented to create a new key.  \\
        revocation\_basepoint, revocation\_basepoint\_secret & Corresponds to the base of the revocation key. \\
        per\_commitment\_point, per\_commitment\_secret & Corresponds to the index of the revocation key. \\
        \hline
    \end{tabularx}
\end{table}

\subsubsection{Transactions}

The funding transaction contains an output whose amount equals the capacity of the payment channel and which can be spent using signatures for both users' funding\_pubkey.

The commitment transaction has one input that references the funding transaction's output. Thus, the commitment transaction must be signed by signatures for both users' funding\_pubkey.
The outputs of a commitment transaction are as follows:
\begin{itemize}
  \item to\_local:
  The first output contains the funds of the party who can publish this commitment transaction. The output can be spent using a signature by local\_delayedprivkey after a timeout of length to\_self\_delay or it can be spent using revocationprivkey.
  \item to\_remote:
  The output that contains funds for the remote party can be spent using remoteprivkey.
  \item Outgoing HTLCs:
  For each outgoing HTLC, there exists an output that can be spent either using the revocationprivkey or using remote\_htlcprivkey and either using local\_htlcprivkey or by providing the HTLC's preimage.
  \item Incoming HTLCs: 
  For each incoming HTLC, there exists an output that can be spent either using revocationprivkey or using remote\_htlcprivkey and either the HTLC's preimage and local\_htlcprivkey or after the HTLC's timeout.
\end{itemize}

HTLC Transactions:
The HTLC success transaction has a locktime value of 0 which means that it is immediately valid.
The HTLC timeout transaction has a locktime value of cltv\_expiry (the HTLC's timeout), i.e. it is valid only after the point in time specified by the HTLC's timeout.
An HTLC's transaction input references the respective HTLC output of a commitment transaction. The HTLC transaction must be signed using remote\_htlcprivkey and local\_htlcprivkey.
For an HTLC success transaction, the input must contain the HTLC's preimage.
The output of both HTLC transactions is spendable either by the revocationprivkey or after to\_self\_delay using local\_delayedprivkey.

\subsubsection{Formalization of Transactions}

The formalization of transactions follows the UTXO (unspent transaction output) model of Bitcoin: A transaction is a record that contains a set of inputs, a set of outputs and additional data like the transaction's id, and an optional absolute timelock.
An output of a transaction is a record that contains an id, an amount, and a set of conditions of which one needs to be fulfilled to spend the output.
A condition consists of a type, a set of keys, an optional hash value, an absolute timelock and a relative timelock.
The condition can be one of the symbolic values SingleSignature, AllSignatures, SingleSigHashLock, AllSigHashLock.
SingleSignature specifies that this output can only be spent by a transaction signed using one of the keys stored in the condition. AllSignatures specifies that a spending transaction must be signed using all of the keys stored in the condition.
The types with the suffix `HashLock' specify the additional constraint that a spending transaction must provide a preimage to the hash value specified in the condition.
Additionally, a spending transaction must specify timelocks that have values set to a value that is greater than or equal to the values specified in the condition. The absolute timelock enforces that a spending transaction is only valid if the transaction is published at a point in time that is greater than or equal to the value of the absolute timelock. The relative timelock has the same effect but the transaction must be published at a point in time at which the age of the transaction that the output is contained in is greater than or equal to the relative timelock.
An input of a transaction consists of a reference to an output that is being spent by this input, a relative timelock, and witness data which contains of a set of signatures and an optional preimage.
Transactions in Bitcoin have ids which are calculated by hashing a transaction. For our model, we require the property that each transaction has a unique id.
We model transaction ids by choosing for each transaction an id as an arbitrary integer number that has not yet been used as a transaction id.
The variable $c$UsedTransactionIds stores a set of all transaction ids drawn to ensure that each new transaction id is unique.
To reduce the number of possible states of the formalization, we specify a unique range of integer numbers for each channel $c$.
In the constant AvailableTransactionIds, for each channel $c$ this set of numbers is stored from which transaction ids created by PaymentChannelUser for channel $c$ can be drawn.

Lightning requires users to exchange signatures on transactions.
Because TLA\textsuperscript{+} does not provide a native way to model signatures, we use an approach that meets the following requirements for signatures:
A signature must be bound to the transaction that is signed by the signature so that a user is able to check whether a signature was created for a specific transaction or not.
The validity of a signature can be checked using the corresponding public key.
However, a signature must only be creatable using the corresponding private key.
To model these requirements in a simple way, we model signatures of transactions by using the whole transaction with the signing key added to the inputs of the transaction.
This makes it simple to exchange signed transactions by exchanging the whole transaction and it is trivial to verify the validity of a signature by comparing the signature to the respective public key.
A caveat here is that when in Lightning only a signature is exchanged (e.g., in the `funding\_signed' and `commitment\_signed' messages), both users must be able to create the respective transaction by themselves.
To ensure that both users have all the information necessary to create the transactions in the formalization, we model the reception of a message that contains a signed transaction by comparing the received message to a record that contains the expected transaction built from locally available information.

We model the blockchain as a set of transactions. The variable LedgerTx is a set that contains the transactions that have been published on the blockchain.
For all transactions that are added to LedgerTx, an entry in TxAge is added that models the time since the transaction was included in the blockchain.
If a transaction is published, the transaction's entry in TxAge is initialized to 0.
When time is increased, all entries in TxAge are increased by the same time difference.

\subsubsection{Opening a Payment Channel}
\label{sec-lightning-opening}

To open a payment channel, the funder of the payment channel sends an open\_channel message to the user to whom the channel should be opened.
The open\_channel message contains several fields for the parameterization of the channel that we ignore for the TLA\textsuperscript{+} formalization (see \cref{tab_fields_open_channel}).
If the receiver wants the channel to be opened, the receiver replies with an `accept\_channel' message (see \cref{tab_fields_accept_channel}) in which the receiver sends their parameters for the channel to the funder. This includes the basepoints for the keys and the first commitment point to derive the public keys for the first commitment transaction.
With this information, the funder creates the funding transaction and the first commitment transaction.
The funder sends the `funding\_created' message (see \cref{tab_fields_funding_created}) that contains the funder's signature of the first commitment transaction and the id of the funding transaction which is required for creating commitment transactions.
The other user receives the `funding\_created' message, validates and stores the signature.
The other user creates the version of the first commitment transaction that the funder can publish and sends a signature for the first commitment transaction to the funder in the `funding\_signed' message (see \cref{tab_fields_funding_signed}).
Having verified and stored the signature of the first commitment transaction, the funder publishes the funding transaction.
After both users have noted that the funding transaction was published and confirmed on the blockchain, both users send each other a `channel\_ready' message which indicates that the channel can now be used to process payments.
The `channel\_ready' message (see \cref{tab_fields_channel_ready}) contains the second per\_commitment\_point that is required to derive the public keys required to create the second commitment transaction.

\begin{table}
    \caption{Fields of `open\_channel' message.}
    \label{tab_fields_open_channel}
    \scriptsize
    \begin{tabularx}{\linewidth}{XXX}
        \hline
        Variable & Description & Formalization \\
        \hline
        chain\_hash & Hash of the blockchain underlying the payment channel. & Not required because formalization considers only one blockchain. \\
        temporary\_channel\_id & Temporary id to describe the channel. & Not required because channel is uniquely identified by the pair of the channel parties. \\
        funding\_satoshis & Number of satoshis that the funder deposits into the channel. & Field `Capacity' \\
        push\_msat & Number of millisatoshi that the funder is giving to the other party. & Not included for simplicity. \\
        dust\_limit\_satoshis & No output is created with an amount less than this value. This prevents outputs from being created that cannot be spent economically because the amount of fees required to spend the output would be higher than the value of the output. & Not required because fees are not modeled. \\
        max\_htlc\_value\_in\_flight\_msat & Maximum amount of millisatoshi that can be part of HTLCs. & Not included for simplicity. \\
        channel\_reserve\_satoshis & Amount of satoshis that both users need to keep as their balance and cannot spend to disincentivize cheating attempts.  & Not required because the formalization does not consider the game theoretic aspect of whether a user might want to cheat but instead shows that cheating never succeeds. \\
        htlc\_minimum\_msat & Amount that any HTLC must at least have. & Not required for security properties. \\
        feerate\_per\_kw & Fee rate for commitment and HTLC transactions. & Not required as blockchain fees are not modeled. \\
        to\_self\_delay & Number of blocks that an output of the counterparty must be locked until the counterparty can spend it. & Modeled as constant `TO\_SELF\_DELAY'. \\
        max\_accepted\_htlcs & Maximum number of HTLCs that are accepted at the same time to ensure that messages do not grow too large. & Not required. \\
        funding\_pubkey & Public key used for locking the funding output that can only be spent using signatures of keys of both users. & Modeled as the user's public key. \\
        revocation\_basepoint & Point used to calculate revocation keys. & Modeled in form of a public/private revocation key pair. \\
        payment\_basepoint & Point used to calculate the payment keys. & Modeled as the user's public key. \\
        delayed\_payment\_basepoint & Point used to calculate the delayed payment keys. & Modeled as the user's public key. \\
        htlc\_basepoint & Point used to calculate the HTLC keys. & Modeled as the user's public key. \\
        first\_per\_commitment\_point & Point to calculate keys for the first commitment transaction. & Only revocation key is rotated. Modeled as key index variable as part of the key. \\
        channel\_flags & Flags, e.g. whether this channel is to announced publicly to the P2P network. & Not required. \\
        \hline
    \end{tabularx}
\end{table}

\begin{table}
    \caption{Fields of `accept\_channel' message. With the exception of minimum\_depth all fields are also described in \cref{tab_fields_open_channel}.}
    \label{tab_fields_accept_channel}
    \begin{tabularx}{\linewidth}{XXX}
        \hline
        Variable & Description & Formalization \\
        \hline
        temporary\_channel\_id & Temporary id to describe the channel. & Not required because channel is uniquely identified by the pair of the channel parties. \\ 
        dust\_limit\_satoshis & No output is created with an amount less than this value. This prevents outputs from being created that cannot be spent economically because the amount of fees required to spend the output would be higher than the value of the output. & Not required because fees are not modeled. \\ 
        max\_htlc\_value\_in\_flight\_msat & Maximum amount of millisatoshi that can be part of HTLCs. & Not included for simplicity. \\ 
        channel\_reserve\_satoshis & Amount of satoshis that both users need to keep as their balance and cannot spend to disincentivize cheating attempts.  & Not required because the formalization does not consider the game theoretic aspect of whether a user might want to cheat but instead shows that cheating never succeeds. \\ 
        htlc\_minimum\_msat & Amount that any HTLC must at least have. & Not required for security properties. \\ 
        minimum\_depth & Number of blocks created after a transaction has been published until the transaction is considered confirmed. & Not required as transactions are modeled as being finally confirmed immediately. \\
        to\_self\_delay & Number of blocks that an output of the counterparty must be locked until the counterparty can spend it. & Modeled as constant `TO\_SELF\_DELAY'. \\ 
        max\_accepted\_htlcs & Maximum number of HTLCs that are accepted at the same time to ensure that messages do not grow too large. & Not required because message sizes are theoretically not  limited. \\ 
        funding\_pubkey & Public key used for locking the funding output that can only be spent using signatures of keys of both users. & Modeled as the user's public key. \\ 
        revocation\_basepoint & Point used to calculate revocation keys. & Modeled in form of a public/private revocation key pair. \\ 
        payment\_basepoint & Point used to calculate the payment keys. & Modeled as the user's public key. \\ 
        delayed\_payment\_basepoint & Point used to calculate the delayed payment keys. & Modeled as the user's public key. \\ 
        htlc\_basepoint & Point used to calculate the HTLC keys. & Modeled as the user's public key. \\ 
        first\_per\_commitment\_point & Point to calculate keys for the first commitment transaction. & Only revocation key is rotated. Modeled as key index variable as part of the key. \\ 

        \hline
    \end{tabularx}
\end{table}

\begin{table}
    \caption{Fields of `funding\_created' message.}
    \label{tab_fields_funding_created}
    \begin{tabularx}{\linewidth}{XXX}
        \hline
        Variable & Description & Formalization \\
        \hline
        temporary\_channel\_id & Temporary id to describe the channel. & Not required because channel is uniquely identified by the pair of the channel parties. \\
        funding\_txid & ID of the funding transaction & . \\
        funding\_output\_index & Output ID for the funding output of the funding transaction. & Not required as funding transaction contains only one output. \\
        signature & Signature of the first commitment transaction & First commitment transaction signed by inserting private key. \\
        \hline
    \end{tabularx}
\end{table}

\begin{table}
    \caption{Fields of `funding\_signed' message.}
    \label{tab_fields_funding_signed}
    \begin{tabularx}{\linewidth}{XXX}
        \hline
        Variable & Description & Formalization \\
        \hline
        channel\_id & ID of the channel derived from funding transaction & Not required. \\
        signature & Signature of the first commitment transaction & First commitment transaction signed by inserting private key. \\
        \hline
    \end{tabularx}
\end{table}

\begin{table}
    \caption{Fields of `channel\_ready' message.}
    \label{tab_fields_channel_ready}
    \begin{tabularx}{\linewidth}{XXX}
        \hline
        Variable & Description & Formalization \\
        \hline
        channel\_id & ID of the channel derived from funding transaction & Not required. \\
        second\_per\_commitment\_point & Point to calculate keys for the second commitment transaction & Only revocation key is rotated. Modeled as key index variable as part of the key. \\
        \hline
    \end{tabularx}
\end{table}

\subsubsection{Formalization of Opening a Payment Channel}
\label{sec-lightning-formalization-opening}

The module PaymentChannelUser contains the actions to open, update, and close a payment channel.
As in the module HTLCUser, the actions of the module PaymentChannelUser are parameterized for a specific user $u$ and a payment channel $c$ of user $u$.
In this section, we use $u$ and $c$ to refer to the user and channel passed as parameters to an action.

To encode in the variables of a user $u$, in which state of the protocol the user is, the variable $c,u$State contains a string of characters that describes the state.
The state is initialized to `init' and after sending an `open\_channel' message, the state advances to `open-sent-open-channel' which means that the user $u$ expects to receive an `accept\_channel' message in the next step.
Sending and receiving of messages is modeled as described above in \cref{sec-lightning-formalization-multi-hop-payments}.
The `open\_channel' and `accept\_channel' messages contain many fields that are not relevant for the TLA\textsuperscript{+} formalization (see \cref{tab_fields_open_channel,tab_fields_accept_channel}).
On reason for fields not being  modeled is that they are used for features that are not included in the TLA\textsuperscript{+} formalization because they are optional features (e.g, `push\_msat').
Further, the TLA\textsuperscript{+} formalization does not model transaction fees for on-chain transactions. We leave this aspect out of scope to keep the model and security analysis focused.
However, modeling transaction fees could be included in the formalization as part of future work. 
The basepoints and first commitment point are modeled by sending the public key as explained in \cref{sec-lightning-keys-transactions}.

After having received an `accept\_channel' message, a funding user creates the funding transaction.
A transaction is created as a record consisting of a set of inputs, a set of outputs, an id and an absolute timelock (see \cref{sec-lightning-keys-transactions}).
The id of the funding transaction is chosen as a unique unpredictable identifier.
Once created, the funding transaction is stored in the field `transactions' of the variable `$c,u$Inventory'.
The sending of the signature of the first commitment transaction in the `funding\_created' message is modeled by the action `SendSignedFirstCommitTransaction'.
As explained in \cref{sec-lightning-keys-transactions}, sending a signature is modeled by sending the whole transaction with they funder's key added to the transaction's input.
In the same way, the non-funding user responds with the signed first commitment transaction for the funder in the `funding\_signed' message.
Once the funder has received the other user's signature on the first commitment transaction, the funder publishes the funding transaction by adding it to the set of published transactions in the variable `LedgerTx'.
Together with adding the transaction to LedgerTx, a new entry for the transaction is added to the variable TxAge which maps transaction ids to clocks.
A new clock for the published transaction is created and initialized to 0. This clock models the age of the transaction, i.e. how many blocks have been created since the transaction was published.
This information is relevant for spending conditions with timelocks.

After the users have noticed that the funding transaction has been published, both  users exchange new revocation public keys which models the exchange of per\_commitment\_points (see \cref{sec-lightning-keys-transactions}).
Because each user can send the new revocation public key to the other user after noticing that the funding transaction has been published, the order in which the `channel\_ready' messages are sent is not deterministic and a user might first send the `channel\_ready' message and then receive the other user's `channels\_ready' message or the other way around.
In the formalization, this has the effect that the actions `SendNewRevocationKey' and `ReceiveNewRevocationKey' are enabled in two different states (e.g., `open-funding-pub' and `open-new-key-received') and update the state accordingly (e.g., `open-new-key-sent' and `rev-keys-exchanged').
After the new revocation public keys are exchanged, the channel is ready to operate, i.e., to add HTLCs to the commitment transaction.

\subsubsection{Updating a Payment Channel}
\label{sec-lightning-updating}

After a user has sent at least one `update\_add\_htlc' message (see \cref{tab_fields_commitment_signed}) to inform the other user about an HTLC, the user sends a `commitment\_signed' message.
The `commitment\_signed' message contains a signature for the new commitment transaction in which all outgoing HTLCs for which an `update\_add\_htlc' message was sent are included and incoming HTLCs that have been fulfilled are not included.
The receiving user responds with a `revoke\_and\_ack' message (see \cref{tab_fields_revoke_and_ack}) to acknowledge the reception of the new commitment transaction signature and to revoke the old commitment transaction by sending the per\_commitment\_secret for the old commitment transaction.

\begin{table}
    \caption{Fields of `commitment\_signed' message.}
    \label{tab_fields_commitment_signed}
    \begin{tabularx}{\linewidth}{XXX}
        \hline
        Variable & Description & Formalization \\
        \hline
        channel\_id & ID of the channel derived from funding transaction & Not required. \\
        signature & Signature of the new commitment transaction & New commitment transaction signed by inserting private key. \\
        num\_htlcs & Number of HTLCs in the new commitment transaction & Not required because the field is redundant. \\
        num\_htlcs * signature & Signature of each HTLC transaction & HTLC transaction signed by inserting private key. \\
        \hline
    \end{tabularx}
\end{table}

\begin{table}
    \caption{Fields of `revoke\_and\_ack' message.}
    \label{tab_fields_revoke_and_ack}
    \begin{tabularx}{\linewidth}{XXX}
        \hline
        Variable & Description & Formalization \\
        \hline
        channel\_id & ID of the channel derived from funding transaction & Not required. \\
        per\_commitment\_secret & Secret for key derivation of keys for the previous commitment transaction. & Only revocation key is rotated. Modeled as sending the private revocation key. \\
        next\_per\_commitment\_point & Point to calculate keys for the new commitment transaction & Only revocation key is rotated. Modeled by sending public revocation key. \\
        \hline
    \end{tabularx}
\end{table}

\subsubsection{Formalization of Updating a Payment Channel}
\label{sec-lightning-formalization-updating}

When the user $u$ is in state `rev-keys-exchanged', the action `SendSignedCommitment' can be enabled if there is at least one HTLC to add or remove and the action `ReceiveSignedCommitment' can be enabled if there is a `commitment\_signed' message in the variable $c$Messages to be received by user $u$.
The actions find HTLCs to be updated by the states of the HTLCs and update the states of the HTLCs according to \cref{fig-HTLC-graph}.
The action `SendSignedCommitment' sends the complete new signed commitment transaction and signed HTLC transactions to model the sending of signatures (see \cref{sec-lightning-keys-transactions}).
If an outgoing HTLC has timed out, it will not be added by `SendSignedCommitment' and a commitment transaction that commits to an incoming HTLC that has timed out will not be accepted by `ReceiveSignedCommitment'. 

The formalization keeps track of the \emph{balance} that a user should have in the channel.
This value is stored in the variable `$c,u$Balance'.
During the opening of the payment channel, the value of $c,u$Balance is set to the channel's capacity for the funder of the channel and to 0 for the other user.
When a user commits to a new HTLC in the action `SendSignedCommitment', the action decrements the user's balance by the amount of the HTLC. 
The variable $c,u$Balance models for each user the balance that the user is guaranteed to receive as long as the user follows the protocol.
The amount of an HTLC is added to a user's balance when the user fulfills the HTLC by sending the preimage to the other user which is modeled by the `SendHTLCPreimage' action of the module HTLCUser.

Additionally, there is a variable $u$ChannelBalance which models the balance that a user has in all of the user's channels.
This variable is part of the security property and updated when a payment is processed.
The differences between the variables $c,u$Balance and $u$ChannelBalance is that $c,u$Balance represents only the balance that a user has in a channel that is not part of an HTLC while $u$ChannelBalance is the sum of all channels including the amount that is locked in HTLCs.

\subsubsection{Closing a Payment Channel}
\label{sec-lightning-closing}

There are two ways to close a payment channel:
The simplest way is to close the channel by publishing the latest commitment transaction on the blockchain.
Both parties can also create a dedicated closing transaction that cannot be revoked and, thus, does not require timeout and which is smaller than a commitment transaction and, thus, costs less fees to publish on the blockchain.
For the formalization, we chose to leave this type of closing out of scope because it is an optimization that is useful but not required for a functional and secure protocol.
Closing might, however, be dishonest if a party publishes a commitment transaction that is not the latest commitment transaction.

Each party has to watch the blockchain for published commitment transactions and react accordingly. 
When a channel is closed with an outdated commitment transaction, the honest party has to spend the commitment transaction's revocation outputs. 
When a channel is closed with a latest commitment transaction that contains HTLCs, these HTLCs need to be resolved on the blockchain.
If an incoming HTLC can be fulfilled, an HTLC success transaction must be published and, if an outgoing HTLC is timed out, an HTLC timeout transaction must be published. 
HTLCs that are not part of the latest commitment transaction, are aborted if they have not been committed, timed out or persisted.

\subsubsection{Formalization of Closing a Payment Channel}
\label{sec-lightning-formalization-closing}

In the formalization, a payment channel can only be closed by publishing a commitment transaction on the blockchain.
This can either be done honestly modeled by the action `CloseChannel' or dishonestly modeled by the action `Cheat'.
We model that a user observes a commitment transaction on the blockchain using different actions depending on whether the other party closed honestly or dishonestly.
When the payment channel is closed honestly, the published commitment transaction defines which HTLCs are committed and which are not.
The action `NoteThatOtherPartyClosedHonestly' updates the states of the HTLCs accordingly:
The state of HTLCs that were in the process of being committed but are not committed is set to \textsc{aborted}.
The state of those HTLCs that are committed is set to \textsc{committed} and the state of HTLCs that were fulfilled and are not in the commitment transaction is set to \textsc{persisted}.
When the payment channel is closed dishonestly, the action `Punish' models that the user $u$ notices the outdated commitment transaction and publishes a transaction that uses the revocation keys to punish the cheating party.

A challenge for the formalization is that a channel might be closed while the channel is in the process of being updated.
A dishonest user might revoke a commitment transaction and publish the revoked commitment transaction.
If the honest user observes the published commitment transaction on the blockchain before the honest user receives the revocation key, the honest user treats the publication of the commitment transaction as an honest closing.
After having received the revocation key, the honest user has to react to the published commitment transaction as a cheating attempt.

\subsubsection{Formalization of Messages} 
\label{sec-lightning-messages-formalization}

The exchange of messages in a channel is modeled by letting the users write messages to a message queue per channel from which each user can read the user's first message.
A message is sent by an action that specifies that the channel's message queue is extended by the message that is sent.
A message contains a field for the recipient, the sender, the message's type, and the payload.
For each type of message, the formalization includes an action that expects such a message and reacts to the message.
An action of user $u$ that reads a message contains a condition that checks that the first message sent to user $u$ has the expected type.
The action updates the user's state as required and removes the message from the queue.
There are two types of variables used for message queues.
The global variable Messages is a set that contains all messages for requesting and sending an invoice.
There is variable $c$Messages for each channel $c$ that models a FIFO queue for messages exchanged between the two users of channel $c$.
Because of this approach to model messages, the messages in $c$Messages are delivered in order and the messages in the global Messages variable can be delivered in an arbitrary order.
Messages can be directly received or arbitrarily delayed because after a step in which a message was sent there can be no or an unlimited number of steps that advance time until a step is taken in which the message is received.

\subsubsection{Formalization of Time Flow}
\label{sec-lightning-time-formalization}

HTLCs as well as commitment transactions use timelocks to enforce that certain actions cannot be done before a certain point in time.
The values used in Lightning for timelocks are numbers that indicate a specific height of the Bitcoin blockchain.
Because the blockchain grows on average at a constant rate, the height of the blockchain represents a logical time that grows on average linear to the clock time.
Because the height of the blockchain represents a logical time, we refer to the height of the blockchain also simply as time.
We model the time, i.e., height of the blockchain, as an integer number that can increase in integer steps.
The current value of the time is stored in the variable LedgerTime.
A step that increases LedgerTime is modeled as an action of the LedgerTime module.
A behavior of the system described by the TLA\textsuperscript{+} specification is a sequence of steps of actions of the users (modules PaymentChannelUser and HTLCUser) and of steps that advance the time, i.e., increase the value of LedgerTime.
Depending on where in a behavior the steps that advance the time are taken between steps of actions of users, the actions of users are modeled to happen slowly or quickly relative to the growth of the blockchain.

In some situations, the protocol requires a user to take an action before the blockchain has reached a certain height.
An example for such a situation is that, if a user knows the preimage for an HTLC, the user must fulfill the HTLC before the HTLC times out.
Consider the following scenario: There is an HTLC from user A to user B with a timeout at time 10. User B knows the preimage and the current time is 8.
While the action for fulfilling the HTLC is enabled, the action for increasing the LedgerTime is also enabled.
It is acceptable that the value of the variable LedgerTime is increased to 9. However, if the value of the variable LedgerTime was increased to 10 before user B fulfills the HTLC, user B would not have followed the protocol which requires user B to fulfill the HTLC before the HTLC's timeout.
Therefore, to model honest behavior of user B, we need to model that user B performs an action before the variable LedgerTime reaches the value 10.
If the value of the variable LedgerTime is 9, the action that increases the value of LedgerTime should not be enabled until user B has fulfilled the HTLC.
To model this urgency requirement, we let the each module specify a set of points in time called \emph{TimeBounds} that are the heights of the blockchain at which the user needs to perform an action at the latest.
The time (resp. height of the blockchain) will not advance further than the minimal height specified by all TimeBounds.
For the previous example, the value of TimeBounds of the module PaymentChannelUser would include the HTLC's timelock - 1 for each incoming HTLC that the user has the preimage for and that is not fulfilled.
After an action that removes the condition for a time bound has been taken, the height of the blockchain can advance further.
We check that no execution of the protocol is stuck because of a time bound but no action is possible and the height of the blockchain cannot advance by using a liveness property that checks that the time finally reaches a specified maximal value.

\subsubsection{Liveness of Users}

The Lightning protocol requires that honest users perform certain steps if they can.
For example, a user must respond to a `commitment\_signed' message with a `revoke\_and\_ack' message.
In the TLA\textsuperscript{+} specification, we model these requirements using a weak fairness condition that specifies that, if for an honest user an action is continuously enabled, the user has to eventually take a step of this action.
In general, we assume that dishonest users take steps in which they read from the environment but not steps in which they actively change their environment, i.e. dishonest users retrieve messages and read the blockchain but are not required to send messages or publish transactions.
However, to a certain degree, we also need to assume liveness for dishonest users.
We make the following exceptions:
A first exception is that, to achieve progress during channel opening, we specify a weak fairness condition that dishonest users actively participate in the opening of the channel.
This simplifies the formalization as we can assume that channels are actually opened.
The requirement is not a practical limitation of the adversary as, until a channel has been opened, there is nothing to loose or gain.
If the protocol execution terminates before the funding transaction has been published, the execution had no effect on the blockchain and, thus, on the balances of the users.
A second exception is that we specify that, once the other user in the channel has terminated, even dishonest users must publish transactions if they can.
This exception simplifies the definition of idealized channels because each user ends in a state in which the user has spent all outputs on the blockchain that the user can spend.
This exception is also not a practical limitation of the adversary as the other user who has already terminated will not be able to profit from actions of the adversary anymore.
We make a third exception for one specific situation:
Assume that an HTLC is part of a commitment transaction that has been honestly published on-chain, i.e. the commitment transaction cannot be revoked.
The timeout of the HTLC has already passed.
The user for whom the HTLC is incoming is honest but does not have the preimage.
The user for whom the HTLC is outgoing is dishonest.
Now, the dishonest user could spend the HTLC output and the honest user would note that the HTLC has timedout on-chain.
If the dishonest user never spends the HTLC output, the honest user stays in a state in which the HTLC might be resolved in two different ways: Either the dishonest user spends the HTLC output on-chain or the honest user might receive the preimage for the HTLC and spend the HTLC output.
Because we specify that users terminate only if they know how all HTLCs have been resolved, this situation prevents honest users from terminating.
We decide this situation by specifying that, in this specific situation, a dishonest user publishes a transaction on-chain to timeout the HTLC and retrieve the HTLC's amount.
This exception is only required because of the specific termination condition in our formalization.
Because the issue is about an output that is expected to be spent by the dishonest user, in practice, this is not a real problem because an honest user retrieves the honest user's balance and it is acceptable if the HTLC is never finally timedout.
An alternative to this exception would be to change the protocol so that the honest user marks the HTLC as timedout although the HTLC has not been timedout on-chain.

\subsection{Generalized Time Skip Theorem}
\label{sec-appendix-generalized-time-skip}

Given a real-time specification $Spec_S$, we define a specification $Spec_{\hat{S}}$ that is implemented by $Spec_S$ and potentially has fewer states.
The two specifications differ only in how time is advanced.
While in a common real-time specification $Spec_S$ time is advanced by steps of one time unit, the definition of specification $Spec_{\hat{S}}$ allows time to only advance to points in time at which a new step becomes possible. We say that a step \emph{becomes possible} in a state with time $t$ if the step would not be possible at the directly preceding point in time $t-1$.
The optimized specification does not allow points in time at which every step that is possible could also have been taken at an earlier point in time.
For a specification in which at many points in time the same steps are possible, this optimization reduces the state space and, thus, reduces the time required for model checking.

In this section, we define a general real-time specification $Spec_S$ and the corresponding optimized specification $Spec_{\hat{S}}$. We formulate that the original real-time specification implements the optimized specification in \cref{theorem-general-time-skip} and  prove \cref{theorem-general-time-skip} in \cref{sec-appendix-time-skip-proof}.

\subsubsection{Real-Time Specification $Spec_S$}

The real-time specification $Spec_S$ to be optimized is defined as $Spec_S = Init_S \land \square [Next_S]_{v_S} \land Liveness_S$ with the set of variables $v_S$.
The specification $Spec_S$ must be an explicit-time real-time specification with a set of clocks $\mathcal{X}$ and an $AdvanceTime_S$ action.
We assume that the specification has the following properties:
The clocks $\mathcal{X}$ are modeled as variables that are not externally visible.
An action may not read a clock and write the clock's value to another variable.\footnote{We see no reason why this would be a practical restriction. To measure time differences, new clocks can be created.}
In all initial states, the clocks have the same value.
By the $AdvanceTime_S$ action, the values of all clocks $x \in \mathcal{X}$ are advanced by $1$. This assumption facilitates the proof, however, it is not a real restriction because a specification that allows only $+1$ advancements of time implements a specification that allows advancements by any natural number.
The $AdvanceTime_S$ action might be disabled if a time bound is reached.
Time bounds are defined for each clock $x \in \mathcal{X}$ by a mapping $B^x$ from states to a set of natural numbers.
Time may not advance from a state $s$ if any clock $x$ has a value that is equal to one of the time bounds in the set $B^x(s)$ for clock $x$ and state $s$.
The $Next_S$ action of the specification $S$ is a disjunct of an internal next action $NextI$ and the $AdvanceTime_S$ action: $Next_S = NextI_S \lor AdvanceTime_S$.
$Liveness_S$ is the conjunction of formulas of the form $WF_{v_S}(A)$ and $SF_{v_S}(A)$ for subactions $A$ of $Next_S$.
In this paper, we assume that all values of the clock $x$ are elements of the set $\mathbb{N}_0$.

\subsubsection{Optimized Specification $\hat{S}$}

We define the optimized specification $\hat{S}$ with the variables $v_{\hat{S}} = v_S$ as: $Spec_{\hat{S}} = Init_{\hat{S}} \land \square [Next_{\hat{S}}]_{vars_{\hat{S}}} \land Liveness_{\hat{S}}$ with $Init_{\hat{S}} = Init_S$ and $Next_{\hat{S}} = AdvanceTime_{\hat{S}} \lor NextI_{\hat{S}}$ and $NextI_{\hat{S}} = NextI_S$ and $Liveness_{\hat{S}} = Liveness_S$.
Thus, the only difference between specifications $S$ and $\hat{S}$ is how time is advanced by the $AdvanceTime$ actions.
To define $AdvanceTime_{\hat{S}}$, we introduce the following definitions.
We refer to the state space of the optimized specification $\hat{S}$ as $\hat{\Sigma}$.

We define a function $\hat{T}_d^x$ that translates a state $s$ to another time by returning a state $s'$ as a copy of state $s$ with the clock $x$ set to $d$, i.e. $s'.x = d$.

\begin{definition}[$\hat{T}_d^x$]
\begin{alignat*}{3}
\hat{T}_{d}^x: \hat{\Sigma} &\to \hat{\Sigma} \\
        s &\mapsto s' \text{ so that } &&(\A &&v \in v_S : s'.v = s.v) \\
        & &&\land &&s'.x = d
\end{alignat*}
\end{definition}

Using the definition of $\hat{T}_d^x$, we can express an assumption on $B^x(s)$.
We assume that $B^x(s)$ is independent of the value of the clock $x$. Formally:

\begin{assumption}
\label{assumption-b-independent}
\[
  \A x \in \mathcal{X}, d \in \mathbb{N}_0, s \in \hat{\Sigma} : B^x(s) = B^x(\hat{T}^x_d(s))
\]
\end{assumption}

We extend the definition of $\hat{T}_d^x$ from a function of states to a function $\hat{\mathcal{T}}^{x}_d$ on behaviors:

\begin{definition}[$\hat{\mathcal{T}}^{x}_d$]
\begin{align*}
\hat{\mathcal{T}}^{x}_{d}: \hat{\Sigma}^* &\to \hat{\Sigma}^* \\
        \langle \sigma_0, \sigma_1, \sigma_2, ... \rangle &\mapsto \langle \hat{T}_d^x(\sigma_0), \hat{T}_d^x(\sigma_1), \hat{T}_d^x(\sigma_2), ... \rangle
\end{align*}
\end{definition}

Define the set $N_{\hat{S}} \subseteq \hat{\Sigma}^*$ as the set of all behaviors that consist of $NextI_{\hat{S}}$ or stuttering steps.

\begin{definition}[$N_{\hat{S}}$]
\begin{align*}
N_{\hat{S}} = \{ \sigma = \langle \sigma_0, \sigma_1, ... \rangle \in \hat{\Sigma}^* :  \forall i \in \mathbb{N}_0 : \astep{\sigma_i}{\sigma_{i+1}}{[NextI_{\hat{S}}]_{v_{\hat{S}}}} \}
\end{align*}
\end{definition}

and accordingly for specification $S$

\begin{definition}[$N_{S}$]
\begin{align*}
N_{S} = \{ \sigma = \langle \sigma_0, \sigma_1, ... \rangle \in \Sigma^* :  \forall i \in \mathbb{N}_0 : \astep{\sigma_i}{\sigma_{i+1}}{[NextI_{S}]_{v_{S}}} \}
\end{align*}
\end{definition}

Using this definition of the set $N_{\hat{S}}$, we can express that a behavior $\sigma \in \hat{\Sigma}^*$ consists of $NextI_{\hat{S}}$ or stuttering steps by stating that $\sigma \in N_{\hat{S}}$.

We define a \emph{newly possible behavior} as a behavior that contains a step that was not possible if the behavior started at the directly preceding point in time.
Formally, a newly possible behavior is a behavior $\sigma \in N_{\hat{S}}$ so that there exists a clock $x \in \mathcal{X}$ so that, while the behavior starts at time $t = \sigma_0.x$, it holds that $\hat{\mathcal{T}}^{x}_{t-1}(\sigma) \notin N_{\hat{S}}$, i.e., if $\sigma$ is translated to the directly preceding point in time, then the resulting behavior $\hat{\mathcal{T}}^{x}_{t-1}(\sigma)$ contains a step that is not allowed by action $NextI_{\hat{S}}$.

A newly possible behavior \emph{for state $s$}, is a newly possible behavior $\sigma$ that starts in a state created by translating state $s$ to a time $t \in \mathbb{N}$ for a clock $x \in \mathcal{X}$, i.e., $\sigma_0 = \hat{T}^x_{t}(s)$.

We define the function $ETP^x$ (EnablingTimePoints) that maps a state $s$ to all points in time at which a newly possible behavior for state $s$ exists.

\begin{definition}[$ETP^x(s)$]
\begin{alignat*}{2}
ETP^x : \Sigma &\to \mathcal{P}(\mathbb{N}) \\
      s &\mapsto \{ t \in \mathbb{N} : \exists &&\sigma = \langle \sigma_0, \sigma_1, ... \rangle \in N_{S} : \\
      & && \sigma_0 = T^x_{t}(s) \land \mathcal{T}^{x}_{t-1}(\sigma) \notin N_{S} \}
\end{alignat*}
\end{definition}

We assume a function $relETP^x(s)$ that is defined in specification $Spec_S$ and maps a state to a set of natural numbers so that the following conditions are met:

\begin{assumption}
\label{assumptions-relETP}
\begin{alignat*}{2}
      \A x \in \mathcal{X}, s \in \Sigma : &\ ETP^x(s) \subseteq relETP^x(s) \\
      &\land \{ b +1 \mid b \in B^x(s) \} \subseteq relETP^x(s) \\
\end{alignat*}
\end{assumption}

For a state $s$ for each value $t$ in the set $ETP^x(s)$, there exists a $NextI_{\hat{S}}$ step in a behavior that starts at state $\hat{T}^x_{t}(s)$ (state $s$ where clock $x$ is translated to time $t$), so that the step is not possible when translating the behavior to time $t-1$.
The set $relETP^x(s)$ contains the points in time that are in $ETP^x(s)$ and may contain additional points.

Using the definition of $relETP^x$, the action $AdvanceTime_{\hat{S}}$ can be defined as an action that sets the new value of each clock $x \in \mathcal{X}$ to a point in time that is in the set $relETP^x$. All other variables are left unchanged:

\begin{definition}[$AdvanceTime_{\hat{S}}$]
\label{def_advance_time_hat}
For two states $s,t \in \hat{\Sigma}$, it holds that $\astep{s}{t}{AdvanceTime_{\hat{S}}}$ iff
\begin{align*}
t.&time > s.time \\
&\land \A v \in v_{\hat{S}} : t.v = s.v \\
&\land \A x \in \mathcal{X} : t.x = s.x \lor t.x \in relETP^x(s) \\
&\land \A x \in \mathcal{X} : \A b \in B^x(s) : s.x \leq b \Rightarrow t.x \leq b
\end{align*}
\end{definition}

With the definition of $AdvanceTime_{\hat{S}}$, the specification $Spec_{\hat{S}}$ is fully defined and we can state the theorem that
$Spec_S$ implements $Spec_{\hat{S}}$:

\begin{theorem}
\label{theorem-general-time-skip}
\[
Spec_S \Rightarrow Spec_{\hat{S}}
\]
\end{theorem}

\subsection{Proof of Generalized Time Skip Theorem}
\label{sec-appendix-time-skip-proof}

Based on the definitions given above in \cref{sec-appendix-generalized-time-skip}, we prove \cref{theorem-general-time-skip}.

\subsubsection{Extended Real-Time Specification $Spec_{S'}$}

Because the proof that specification $Spec_S$ implements specification $Spec_{\hat{S}}$ needs auxiliary variables for defining how a clock is mapped to specification $Spec_S$, we define an extended specification $Spec_{S'}$ that wraps specification $Spec_S$ and adds the auxiliary variables.

Specification $Spec_{S'}$ uses the variables $v_{S'}$ defined as the variables $v_S$ of $Spec_S$ and the auxiliary variables $mappedClock^x$ for each clock $x \in \mathcal{X}$:
\[
v_{S'} = v_S \cup \bigcup_{x \in \mathcal{X}} \{mappedClock^x\}
\]
Specification $Spec_{S'}$ is defined by $Spec_{S'} = Init_{S'} \land \square [Next_{S'}]_{v_{S'}} \land Liveness_{S'}$ where $Init_{S'}$ is defined as:
\begin{align*}
  Init_{S'} = &Init_S \\
    &\land \A x \in \mathcal{X} : mappedClock^x = x
\end{align*}
In TLA\textsuperscript{+}, $x \colongt y$ is the notation for a function that maps $x$ to $y$.
Note, that we do not distinguish between a clock $x$ and the value of the clock $x$ to simplify notation.

$Liveness_{S'}$ is defined as $Liveness_S$.

$Next_{S'}$ is defined as $Next_{S'} = AdvanceTime_{S'} \lor NextI_{S'}$ where $NextI_{S'}$ is defined as:
\begin{align*}
  NextI_{S'} = &NextI_{S} \\
    &\land \A x \in \mathcal{X} : mappedClock'^x = mappedClock^x
\end{align*}
$AdvanceTime_{S'}$ is defined as:
\begin{alignat*}{2}
  Advanc&eTime_{S'} = \\
  Ad&vanceTime_{S} \\
        &\land \A x \in \mathcal{X} : &&mappedClock'^x = max( \{mappedClock^x\} \\
        & && \cup \{n \in relETP^x : n \leq x' \})
\end{alignat*}
$relETP^x$ is used without a state parameter here. The parameter of $relETP^x$ is the `current' state, i.e., the state that is described by the unprimed variables in $AdvanceTime_{S'}$.

\begin{lemma}
$Spec_S \Rightarrow Spec_{S'}$
\end{lemma}

\begin{proof}
\pf\ The lemma follows directly from the definition of $Spec_{S'}$ because the added variables $mappedClock^x$ are auxiliary variables and, in particular, history variables, that do not affect whether a $Next_S$ step is enabled or not.
\end{proof}

\subsubsection{Proof of \cref{theorem-general-time-skip}}

In the proof, we use the following notation: $\overline{F}$ is formula F in which each clock $x \in \mathcal{X}$ is replaced by the variable $mappedClock^x$.
Formally, this can be expressed using the notation $F \textsc{ with } v_{1} \leftarrow e_1, v_{2} \leftarrow e_2$ to describe the expression $F$ where variable $v_{1}$ is substituted by expression $e_1$ and variable $v_{2}$ is substituted by expression $e_2$.
With $x_1, x_2, x_3, ...$ being the clocks in $\mathcal{X}$:
\begin{alignat*}{2}
\overline{F} = F \textsc{ with } x_1 & \leftarrow && mappedClock^{x_1}, \\
x_2 & \leftarrow && mappedClock^{x_2}, \\
x_3 & \leftarrow && mappedClock^{x_3}, \\
  & ...
\end{alignat*}

Analogously to the definition of $\hat{T}^x_d$ above, we define $T'^x_d$ to be a function that, given a state $s \in \Sigma'$, returns a state $s$ that equals the given state $s$ except that the value of the clock $x$ is set to $d$.

\begin{definition}[$T'^x_d$]
\begin{alignat*}{3}
T'^x_{d}: \Sigma' &\to \Sigma' \\
        s &\mapsto s' \text{ so that } &&(\A &&v \in v_S' : s'.v = s.v) \\
        & &&\land &&s'.x = d \\
\end{alignat*}
\end{definition}

We extend the definition of $T'^x_d$ from a function of states to a function of behaviors:

\begin{definition}[$\mathcal{T}'^x_d$]
\begin{align*}
\mathcal{T}'^x_{d}: \Sigma'^* &\to \Sigma'^* \\
        \langle \sigma_0, \sigma_1, \sigma_2, ... \rangle &\mapsto \langle T'^x_d(\sigma_0), T'^x_d(\sigma_1), T'^x_d(\sigma_2), ... \rangle
\end{align*}
\end{definition}

Analogously to $N_{\hat{S}}$, we define $N_{S'} \subseteq \Sigma'^*$ to be the set of all behaviors in specification $S'$ that consist of $NextI_{S'}$ and stuttering steps.

\begin{definition}[$N_{S'}$]
\begin{align*}
N_{S'} = \{ \sigma = \langle \sigma_0, \sigma_1, ... \rangle \in \Sigma'^* :  \forall i \in \mathbb{N}_0 : \astep{\sigma_i}{\sigma_{i+1}}{[NextI_{S'}]_{v_{S'}}} \}
\end{align*}
\end{definition}

We define $n^x(s)$ to equal the minimal $n \in \mathbb{N}_0$ so that there exists a newly possible behavior for given state $s$ where the clock $x$ is set to $s.x - n$.
We assume that the function $min$ is defined so that if the parameter of $min$ is the empty set, $min$ equals $\infty$.

\begin{definition}[$n^x$]
\begin{align*}
n^x(s) \coloneq min(\{ n \in \mathbb{N}_0 : \E \sigma \in N_{S'} : &\  \sigma_0 = T'^x_{s.x - n}(s) \\
 &\land \mathcal{T}'^x_{s.x - n -1}(\sigma) \notin N_{S'} \} ) 
\end{align*}
\end{definition}

\begin{lemma}
\label{lemma-advance-time}
Given two states $s, t \in \Sigma'$ and $\astep{s}{t}{AdvanceTime_{S'}}$, it holds for all $x \in \mathcal{X}$ that:
\[
t.mappedClock^x \geq max(s.mappedClock^x, t.x - n^x(t))
\]
\end{lemma}

We write the proof of this lemma in form of a structured proof as introduced in \cite{lamport_how_2012}.
Given a state $s \in \Sigma'$, we use the notation $\ddot{s}$ to refer to a copy of state $s$ of which the variables $mappedClock^x$ were removed so that $\ddot{s} \in \Sigma = \hat{\Sigma}$.
\begin{proof}
\pf\ We prove the implication of \cref{lemma-advance-time}, by assuming that $s, t \in \Sigma' \land \astep{s}{t}{AdvanceTime_{S'}}$ and proving that $\A x \in \mathcal{X} : t.mappedClock^x = max(s.mappedClock^x, t.x - n^x(t))$.
Let $x \in \mathcal{X}$ be arbitrary but fixed.
\step{2}{$t.mappedClock^x = max( \{s.mappedClock^x\} \cup \{n \in relETP^x(\ddot{s}) : n \leq t.x \})$}
  \begin{proof}
  \pf\ By the assumption that $s, t \in \Sigma' \land \astep{s}{t}{AdvanceTime_{S'}}$, by definition of $AdvanceTime_{S'}$, and by writing the primed variables $v'$ as $t.v$ and the unprimed variables $v$ as $s.v$.
  \end{proof}
\step{3}{$t.mappedClock^x = max( s.mappedClock^x, \\ max(\{n \in relETP^x(\ddot{s}) : n \leq t.x \}))$}
  \begin{proof}
  \pf\ By \stepref{2} and rearranging the values to take the maximum from.
  \end{proof}
\step{4}{$max(\{n \in relETP^x(\ddot{s}) : n \leq t.x \}) \geq t.x - n^x(t)$}
  \begin{proof}
  \pf\ 
  \step{4-1}{$max(\{n \in relETP^x(\ddot{s}) : n \leq t.x \}) = t.x - min(\{ n \in \mathbb{N}_0 : (t.x - n) \in relETP^x(\ddot{s}) \}) $}
    \begin{proof}
    \pf\ $max(\{n \in relETP^x(\ddot{s}) : n \leq t.x \})$ can be written as $max(\{n \in \mathbb{N}_0 : n \in relETP^x(\ddot{s}) \land n \leq t.x \})$ because $relETP^x(\ddot{s}) \subseteq \mathbb{N}_0$.
    The greatest value $n$ so that $n \leq t.x$ can be written as $t.x - n_0$ for the minimal value $n_0$ so that $t.x - n_0 \leq t.x$.
    Then, $max(\{n \in relETP^x(\ddot{s}) : n \leq t.x \}) = t.x - min(\{n \in \mathbb{N}_0 : (t.x - n) \in relETP^x(\ddot{s}) \land (t.x - n) \leq t.x \})$.
    The statement follows because $t.x - n \leq t.x$ holds for all $n \in \mathbb{N}_0$.
    \end{proof}
  \step{4-2}{$min(\{ n \in \mathbb{N}_0 : (t.x - n) \in relETP^x(\ddot{s}) \}) \geq min(\{ n \in \mathbb{N}_0 : (t.x - n) \in ETP^x(\ddot{s}) \}) = min(\{ n \in \mathbb{N}_0 : \exists \sigma \in N_{S} : \sigma_0 = \hat{T}^x_{t.x - n}(\ddot{s}) \land \hat{\mathcal{T}}^x_{t.x - n-1}(\sigma) \notin N_{S} \} \})$}
    \begin{proof}
    \pf\ By definition of $relETP^x$ and $ETP^x$.
    \end{proof}
  \step{4-4}{$min(\{ n \in \mathbb{N}_0 : \E \sigma \in N_{S'} : \sigma_0 = T'^x_{t.x - n}(t) \land \mathcal{T}'^x_{t.x - n -1}(\sigma) \notin N_{S'} \}) = n^x(t)$}
    \begin{proof}
    \pf\ Definition of $n^x$.
    \end{proof}
  \step{4-3}{$\A n \in \mathbb{N}_0 : (\exists \sigma \in N_{S} : \sigma_0 = \hat{T}^x_{t.x - n}(\ddot{s}) \land \hat{\mathcal{T}}^x_{t.x - n-1}(\sigma) \notin N_{S} \iff \E \sigma' \in N_{S'} : \sigma'_0 = T'^x_{t.x - n}(t) \land \mathcal{T}'^x_{t.x - n -1}(\sigma') \notin N_{S'} )$.}
    \begin{proof}
    \pf\ We prove the statement by proving both directions of the implication:
    \step{4-3-1}{\assume{$n \in \mathbb{N}_0$, $\sigma \in N_{S}$, $\sigma_0 = \hat{T}^x_{t.x - n}(\ddot{s})$, $\hat{\mathcal{T}}^x_{t.x - n-1}(\sigma) \notin N_{S}$}
                           \prove{$\E \sigma' \in N_{S'} : \sigma'_0 = T'^x_{t.x - n}(t) \land \mathcal{T}'^x_{t.x - n -1}(\sigma') \notin N_{S'} )$}}
      \begin{proof}
      \pf\ Define a behavior $\sigma' \in N_{S'}$ by adding the variables $mappedClock^x$ to each state of a copy $\sigma'$ of $\sigma$ with, for each state $\sigma'_i$ of the behavior $\sigma'$ it holds that $\sigma'_i.mappedClock^x = s.mappedClock^x$.
      Because the variables $mappedClock^x$ do not affect the conditions of $NextI_{S'}$ and are not changed by $NextI_{S'}$ and because all other variables are changed by $NextI_{S'}$ as by $NextI_{S}$, it holds that each step in $\sigma'$ is a step of $NextI_{S'}$ which means that $\sigma' \in N_{S'}$.
      
      By definition of $AdvanceTime_{S'}$, the states $s$ and $t$ differ only by the variables $mappedClock^x$ and $x$.
      Because $\hat{T}^x_{t.x - n}(\ddot{s}).x = t.x - n = T'^x_{t.x - n}(t).x$, it follows that $\sigma'_0 = T'^x_{t.x -n}(t)$.
      By the definition of $\sigma'$ and by the assumption $\hat{\mathcal{T}}^x_{t.x - n-1}(\sigma) \notin N_{\hat{S}}$, it holds that $\mathcal{T}'^x_{t.x - n -1}(\sigma') \notin N_{S'}$.
      \end{proof}
    \step{4-3-2}{\assume{$n \in \mathbb{N}_0$, $\sigma' \in N_{S'}$, $\sigma'_0 = T'^x_{t.x - n}(t)$, $\mathcal{T}'^x_{t.x - n -1}(\sigma') \notin N_{S'}$}
                           \prove{$\E \hat{\sigma} \in N_{\hat{S}} : \hat{\sigma}_0 = \hat{T}^x_{t.x - n}(\ddot{s}) \land \hat{\mathcal{T}}^x_{t.x - n-1}(\hat{\sigma}) \notin N_{\hat{S}}$}}
      \begin{proof}
      \pf\ Define a behavior $\hat{\sigma} \in N_{\hat{S}}$ by removing the variables $mappedClock^x$ from each state in $\sigma$ and duplicating the first state, i.e. $\hat{\sigma} = \langle \ddot{\sigma_0}', \ddot{\sigma_0}', \ddot{\sigma_1}', \ddot{\sigma_2}', ... \rangle$.
      Because the first step is a stuttering step and $NextI_{S'} = NextI_{\hat{S}}$, it holds that $\hat{\sigma} \in N_{\hat{S}}$.
      
      By definition of $AdvanceTime_{S'}$, all variables of the states $s$ and $t$ are equal except the variables $mappedClock^x$ and $x$.
      It holds that $\ddot{T}'^x_{t.x - n}(t) = \hat{T}'^x_{t.x - n}(\ddot{s})$ because in both states the value of clock $x$ is $t.x -n$ and the variables $mappedClock^x$ do not appear in both states.
      Because $\hat{\sigma}_0 = \ddot{\sigma_0}' = \ddot{T}'^x_{t.x - n}(t)$, it holds that $\hat{\sigma}_0 = \hat{T}^x_{t.x - n}(s)$.

      By definition of $\hat{\sigma}$ and by the assumption $\mathcal{T}'^x_{t.x - n -1}(\sigma') \notin N_{S'}$, it holds that $\hat{\mathcal{T}}^x_{t.x - n-1}(\hat{\sigma}) \notin N_{\hat{S}}$.
      \end{proof}
    \qedstep
      \begin{proof}
      By \stepref{4-3-1} and \stepref{4-3-2}.
      \end{proof}
    \end{proof}
  \qedstep
    \begin{proof}
    After rearranging the left part (\stepref{4-1}), expanding the definition of $relETP$ (\stepref{4-2}), and expanding the definition of $n(t)$ (\stepref{4-4}), the conclusion follows by \stepref{4-3}.
    \end{proof}
  \end{proof}
  \qedstep
    \begin{proof}
    The lemma follows from \stepref{3} by using \stepref{4}.
    \end{proof}
\end{proof}

\begin{definition}[$Inv_{S'}$]
\label{def-inv}
Define an invariant $Inv_{S'}$ of a state $s$ of specification $Spec_{S'}$ as:
\[
Inv_{S'}(s) \equiv \A x \in \mathcal{X} : s.mappedClock^x \geq s.x - n^x(s)
\]
\end{definition}

\begin{lemma}
\label{lemma-inv}
$Inv_{S'}$ is an invariant of specification $Spec_{S'}$:
\[
Spec_{S'} \Rightarrow \square  Inv_{S'}
\]
\end{lemma}

\begin{proof}
\step{1}{$Init_{S'} \Rightarrow Inv_{S'} $}
  \begin{proof}
  Let $x \in \mathcal{X}$ be arbitrary but fixed.
  By definition of $Init_{S'}$, it holds that $s.mappedClock^x = s.x$.
  Because it holds that $n^x(s) \geq 0$ by definition of $n^x$ and by definition of $min$, it holds that $s.mappedClock^x - s.x \geq - n^x(s)$.
  \end{proof}
\step{2}{$Inv_{S'} \land [Next_{S'}]_{v_{S'}} \Rightarrow Inv'_{S'} $}
  \begin{proof}
  \step{2-1}{$Inv_{S'} \land (v'_{S'} = v_{S'}) \Rightarrow Inv'_{S'} $}
    \begin{proof}
    \pf\ Assume that $Inv_{S'}$ and $v'_{S'} = v_{S'}$.
    It directly follows that $Inv'_{S'}$.
    \end{proof}
  \step{2-2}{$Inv_{S'} \land AdvanceTime_{S'} \Rightarrow Inv'_{S'} $}
    \begin{proof}
    \pf\ We refer to the step's starting state as $s$ and to the ending state as $t$.
    Let $x \in \mathcal{X}$ be arbitrary but fixed.
    \Cref{lemma-advance-time} states that it follows from $AdvanceTime_{S'}$ that $t.mappedClock^x \geq max(s.mappedClock^x, t.x - n^x(t))$.
    It follows that $t.mappedClock^x \geq t.x - n^x(t)$
    This is equal to $Inv_{S'}(t)$ which can be written as $Inv'_{S'}$.
    \end{proof}
  \step{2-3}{$Inv_{S'} \land NextI_{S'} \Rightarrow Inv'_{S'} $}
    \begin{proof}
    \pf\ We refer to the step's starting state as $s$ and to the ending state as $t$.
    Let $x \in \mathcal{X}$ be arbitrary but fixed.
    We assume that $Inv_{S'} \land NextI_{S'}$ holds and we prove that $Inv'_{S'}$.
    \step{2-3-1}{$n^x(s) \leq n^x(t)$}
      \begin{proof}
      \pf\ In the case that $n^x(s) = 0$, it follows that $n^x(s) \leq n(t)$ because it follows from the definition of $n^x(t)$ that $n^x(t) \geq 0$.
      
      In the following, we assume that $n^x(s) > 0$ which is the only other possible case because $n^x(s) \geq 0$ by definition of $n^x(s)$.
      \step{2-3-1-1}{$\A n \in \mathbb{N}_0$ it holds that, if $n < n^x(s)$, then $\A \sigma \in N_{S'} : \sigma_0 = T'^x_{s.x - n}(s) \Rightarrow \mathcal{T}'^x_{s.x - n -1}(\sigma) \in N_{S'}$}
        \begin{proof}
        \pf\ Follows from the definition of $n^x(s)$.
        $n^x(s)$ is defined as $min(\{n \in \mathbb{N}_0 : f(n) \}$ for a boolean-valued function $f$. The definition of $min$ implies that $\A n \in \mathbb{N}_0 : n < n^x(s) \Rightarrow \lnot f(n)$.
        \end{proof}

      \step{2-3-1-7}{$\A n \in \mathbb{N}_0$ it holds that, if $n < n^x(s)$, then the step from $T'^x_{s.x - n}(s)$ to $T'^x_{s.x - n}(t)$ is a $NextI_{S'}$ step.}
        \begin{proof}
        \pf\ 
        We show this by induction over $n \in \mathbb{N}_0$.
        For $n = 0$, we have to show that the step from $T'^x_{s.x - 0}(s)$ to $T'^x_{s.x - 0}(t)$ is a $NextI_{S'}$ step.
        This holds of \stepref{2-3} because $T'^x_{s.x}(s) = s$ and the step from $s$ to $t$ is a $NextI_{S'}$ step
        We assume that if $n < n^x(s)$, then the step from $T'^x_{s.x - n}(s)$ to $T'^x_{s.x - n}(t)$ is a $NextI_{S'}$ step.
        We prove that if $n+1 < n^x(s)$, then the step from $T'^x_{s.x - (n+1)}(s)$ to $T'^x_{s.x - (n+1)}(t)$ is a $NextI_{S'}$ step.
        Because the step from $T'^x_{s.x - n}(s)$ to $T'^x_{s.x - n}(t)$ is a $NextI_{S'}$ step, there exists a behavior $\sigma = \langle T'^x_{s.x - n}(s), T'^x_{s.x - n}(t), ... \rangle$.
        By \stepref{2-3-1-1}, the behavior $\mathcal{T}'^x_{s.x-(n+1)}(\sigma)$ is a behavior in $N_{S'}$ which means that its first step is a $NextI_{S'}$ step from $T'^x_{s.x - (n+1)}(s)$ to $T'^x_{s.x - (n+1)}(t)$.
        \end{proof}
      
      \step{2-3-1-8}{$\A n \in \mathbb{N}_0$ it holds that, if $n < n^x(s)$, then $\A \sigma \in N_{S'} : \sigma_0 = T'^x_{s.x - n}(t) \Rightarrow \mathcal{T}'^x_{s.x - n -1}(\sigma) \in N_{S'}$}
        \begin{proof}
        \pf\ By \stepref{2-3-1-7}, it holds for all $n \in \mathbb{N}_0$ if $n < n^x(s)$, then for every behavior $\sigma^t$ that starts in state $\sigma^t_0 = T'^x_{s.x - n}(t)$, a behavior $\sigma^s$ that starts in $\sigma^s_0 = T'^x_{s.x - n}(s)$ can be constructed because the step from $T'^x_{s.x - n}(s)$ to $T'^x_{s.x - n}(t)$ is a $NextI_{S'}$ step. Because by \stepref{2-3-1-1} it holds that $\mathcal{T}'^x_{s.x - n -1}(\sigma^s) \in N_{S'}$, it follows that $\mathcal{T}'^x_{s.x - n -1}(\sigma^t) \in N_{S'}$.
        \end{proof}

      \step{2-3-1-9}{$min(\{ n \in \mathbb{N}_0 : \E \sigma \in N_{S'} : \sigma_0 = T'^x_{t.x - n}(t) \land \mathcal{T}'^x_{t.x - n -1}(\sigma) \notin N_{S'} \} ) \geq n^x(s)$}
        \begin{proof}
        \pf\ By \stepref{2-3-1-8} and because $s.x = t.x$.
        \end{proof}

      \qedstep
        \begin{proof}
        \pf\ By \stepref{2-3-1-9} and the definition of $n^x(t) = min(\{ n \in \mathbb{N}_0 : \E \sigma \in N_{S'} : \sigma_0 = T'^x_{t.x - n}(t) \land \mathcal{T}'^x_{t.x - n -1}(\sigma) \notin N_{S'} \} )$.
        \end{proof}

      \end{proof}
    \qedstep
      \begin{proof}
      Because $NextI_{S'}$ leaves the variables $mappedClock^x$ and $x$ unchanged, it holds that $t.mappedClock^x = s.mappedClock^x$ and $t.x = s.x$.
      By $Inv_{S'}$ it follows that $t.mappedClock^x \geq t.x - n^x(s)$ and, by \stepref{2-3-1}, it follows that $t.mappedClock^x \geq t.x - n^x(t)$.
      \end{proof}
    \end{proof}
  \qedstep
    \begin{proof}
    \pf\ Because \stepref{2-1}, \stepref{2-2}, \stepref{2-3}, and $[Next_{S'}]_{v_{S'}} = (AdvanceTime_{S'} \lor NextI_{S'}) \lor (v'_{S'} = v_{S'})$.
    \end{proof}
  \end{proof}
\qedstep
  \begin{proof}
  By \stepref{1}, \stepref{2}, and the definition of $Spec_{S'}$.
  \end{proof}
\end{proof}

\begin{lemma}
\label{lemma-init}
\[
Init_{S'} \Rightarrow \overline{Init_{\hat{S}}}
\]
\end{lemma}

\begin{proof}
\pf\ Recall that, by definition, $Init_{S'} = Init_S \land \A x \in \mathcal{X} : mappedClock^x = x$
and $\overline{Init_{\hat{S}}} = (Init_{\hat{S}} \textsc{ with } x_1 \leftarrow mappedClock^{x_1}, x_2 \leftarrow mappedClock^{x_2}, x_3 \leftarrow mappedClock^{x_3},  ...$
and $Init_{\hat{S}} = Init_S$.

Because $\A x \in \mathcal{X} : mappedClock^x = x$, it trivially holds that $Init_{S'} \Rightarrow \overline{Init_{\hat{S}}}$.

\end{proof}

\begin{lemma}
\label{lemma-inv-applied}
For every state $s \in \Sigma'$ it holds that if $Inv_{S'}(s)$ then
\[
\A x \in \mathcal{X} : \A \sigma \in N_{S'} : \sigma_0 = s \Rightarrow \mathcal{T}'^x_{s.mappedClock^x}(\sigma) \in N_{S'}
\]
and $Inv_{S'}(T'^x_{s.mappedClock^x}(s))$.
\end{lemma}

\begin{proof}
\pf\ Let $s \in \Sigma'$. We assume that $Inv_{S'}(s)$ holds.
\step{1}{$\A x \in \mathcal{X} : \A \sigma \in N_{S'} : \sigma_0 = s \Rightarrow \mathcal{T}'^x_{s.mappedClock^x}(\sigma) \in N_{S'}$}
  \begin{proof}
  \pf\ Let $x \in \mathcal{X}$ be arbitrary but fixed.
  \step{1-1}{$n^x(s) > s.x - s.mappedClock^x - 1$}
    \begin{proof}
    \pf\ By assumption, $Inv_{S'}(s)$ holds, i.e. $s.mappedClock^x \geq s.x - n^x(s)$.
    By rearranging, it follows that $n^x(s) \geq s.x - s.mappedClock^x$.
    Because time is discrete and increases in integer steps, we can replace the $\geq$ by $>$ and it follows that $n^x(s) > s.x - s.mappedClock^x - 1$
    \end{proof}
  \step{1-2}{$\A n \in \mathbb{N}_0 : n < n^x(s) \Rightarrow \A \sigma \in N_{S'} : \sigma_0 = s \Rightarrow \mathcal{T}'^x_{s.x - n -1}(\sigma) \in N_{S'}$}
    \begin{proof}
    \pf\ This follows from the definition of $n^x(s)$.
    By definition, $n^x(s) = min(\{ n \in \mathbb{N}_0 : \E \sigma \in N_{S'} : \sigma_0 = T'^x_{s.x - n}(s) \land \mathcal{T}'^x_{s.x - n -1}(\sigma) \notin N_{S'} \} )$.
    In words, $n^x(s)$ is defined as the smallest value by that the value of clock $x$ in state $s$ must be decreased so that a valid behavior exists that starts in state $T'^x_{s.x - n^x(s)}(s)$ but the behavior $\mathcal{T}'^x_{s.x - n^x(s)-1}(\sigma)$ in which $x$ is one step of time lower is not a valid behavior.
    Therefore, for all values $n$ smaller than $n^x(s)$, it holds that each behavior that starts in state $s$ can be translated to a valid behavior where the clock $x$ is set to $s.x - n$.
    Formally, this can be expressed as:
    \step{1-2-1}{$\A n \in \mathbb{N}_0 : n < n^x(s) \Rightarrow \A \sigma \in N_{S'} : \sigma_0 = T'^x_{s.x - n}(s) \Rightarrow \mathcal{T}'^x_{s.x - n -1}(\sigma) \in N_{S'}$}
      \begin{proof}
      \pf\ From the definition of $n^x(s)$ follows by definition of $min$ that $\A n \in \mathbb{N}_0 : n < n^x(s) \Rightarrow \neg (\E \sigma \in N_{S'} : \sigma_0 = T'^x_{s.x - n}(s) \land \mathcal{T}'^x_{s.x - n -1}(\sigma) \notin N_{S'})$.

      Moving the negation inwards: $\A n \in \mathbb{N}_0 : n < n^x(s) \Rightarrow \A \sigma \in N_{S'} : \neg (\sigma_0 = T'^x_{s.x - n}(s) \land \mathcal{T}'^x_{s.x - n -1}(\sigma) \notin N_{S'})$

      Applying the equivalence $\neg (a \land \neg b) = a \Rightarrow b$ from boolean algebra: $\A n \in \mathbb{N}_0 : n < n^x(s) \Rightarrow \A \sigma \in N_{S'} : \sigma_0 = T'^x_{s.x - n}(s) \Rightarrow \mathcal{T}'^x_{s.x - n -1}(\sigma) \in N_{S'}$
      \end{proof}
    Informally speaking, \stepref{1-2-1} states that for all $n < n^x(s)$ all behaviors starting in a state in which the clock $x$ is set to $s.x - n$ can be translated to a valid behavior in which the clock $x$ is set to $s.x -n -1$.
    However, the statement to prove requires that all behaviors starting in a state in which the clock $x$ is unchanged, i.e.,  set to $s.x - 0$, can be translated to a valid behavior in which the clock $x$ is set to $s.x -n -1$.
    We use induction over $n \in \mathbb{N}_0$, to show that all the steps reducing the clock $x$ by a single time step can be combined.

    \qedstep
      \begin{proof}
      \step{1-2-2}{If $n = 0$, then $n < n^x(s) \Rightarrow \A \sigma \in N_{S'} : \sigma_0 = s \Rightarrow \mathcal{T}'^x_{s.x - n -1}(\sigma) \in N_{S'}$}
        \begin{proof}
        \pf\ Follows from \stepref{1-2-1} by inserting $n = 0$.
        \end{proof}
      Let $n \in \mathbb{N}_0$ be arbitrary but fixed.
      \step{1-2-3}{\assume{$n < n^x(s) \Rightarrow \A \sigma \in N_{S'} : \sigma_0 = T'^x_{s.x}(s) \Rightarrow \mathcal{T}'^x_{s.x - n -1}(\sigma) \in N_{S'}$}
                    \prove{$n + 1 < n^x(s) \Rightarrow \A \sigma \in N_{S'} : \sigma_0 = T'^x_{s.x}(s) \Rightarrow \mathcal{T}'^x_{s.x - (n+1) - 1}(\sigma) \in N_{S'}$}}
        \begin{proof}
        \pf\ 
        \step{1-2-3-1}{$n + 1 < n^x(s) \Rightarrow \A \sigma \in N_{S'} : \sigma_0 = T'^x_{s.x}(s) \Rightarrow \mathcal{T}'^x_{s.x - (n+1)}(\sigma) \in N_{S'}$}
          \begin{proof}
          It holds that $n + 1 < n^x(s) \Rightarrow n < n^x(s)$ and, by writing $s.x - n -1$ as $s.x - (n +1)$, the statement follows from the assumption.
          \end{proof}
        \step{1-2-3-2}{$n+1 < n^x(s) \Rightarrow \A \tau \in N_{S'} : \tau_0 = T'^x_{s.x - (n+1)}(s) \Rightarrow \mathcal{T}'^x_{s.x - (n+1) -1}(\tau) \in N_{S'}$}
          \begin{proof}
          \pf\ By \stepref{1-2-1} and replacing $n$ by $n+1$.
          \end{proof}
        \step{1-2-3-3}{$n + 1 < n^x(s) \Rightarrow \A \sigma \in N_{S'} : \sigma_0 = T'^x_{s.x}(s) \Rightarrow \mathcal{T}'^x_{s.x - (n+1)}(\sigma) \in N_{S'} \Rightarrow \mathcal{T}'^x_{s.x - (n+1) -1}(\mathcal{T}'^x_{s.x - (n+1)}(\sigma)) \in N_{S'}$}
          \begin{proof}
          \pf\ 
          For all $\sigma \in N_{S'}$ with $\sigma_0 = s$ and all $n \in \mathbb{N}_0$ it holds by definition of $\mathcal{T}'^x$ that the first state of $\mathcal{T}'^x_{s.x - (n+1)}(\sigma) \in N_{S'}$ is $T'^x_{s.x - (n+1)}(s)$.
          Thus, we can apply \stepref{1-2-1} with $n$ set to $n+1$ and the statement follows.
          \end{proof}
        \qedstep
          \begin{proof}
          Because $\mathcal{T}'^x_{s.x - (n+1) -1}(\mathcal{T}'^x_{s.x - (n+1)}(\sigma)) = \mathcal{T}'^x_{s.x - (n+1) -1}(\sigma)$, it follows from \stepref{1-2-3-3} that $n + 1 < n^x(s) \Rightarrow \A \sigma \in N_{S'} : \sigma_0 = T'^x_{s.x}(s) \Rightarrow \mathcal{T}'^x_{s.x - (n+1) - 1}(\sigma) \in N_{S'}$.
          \end{proof}
        \end{proof}
      \end{proof}
    \end{proof}
  \qedstep
    \begin{proof}
    \step{1-3-1}{$s.x - s.mappedClock^x - 1 \geq 0$}
      \begin{proof}
      \pf\ If $s.x = s.mappedClock^x$, then \cref{lemma-inv-applied} is trivial. Thus, we assume that $s.x \neq s.mappedClock^x$.
      By definition of $AdvanceTime_{S'}$, it holds that $s.x > s.mappedClock^x$.
      If follows that $s.x - s.mappedClock^x > 0$ and, because time increases in discrete steps, it holds that $s.x - s.mappedClock^x \geq 1$.
      \end{proof}
    From \stepref{1-1} and \stepref{1-3-1} it follows that we can insert $s.x - s.mappedClock^x - 1$ for $n$ in \stepref{1-2}.
    It holds that $\A \sigma \in N_{S'} : \sigma_0 = s \Rightarrow \mathcal{T}'^x_{s.x - (s.x - s.mappedClock^x - 1) -1}(\sigma) \in N_{S'}$.
    This is equal to $\A \sigma \in N_{S'} : \sigma_0 = s \Rightarrow \mathcal{T}'^x_{s.mappedClock^x}(\sigma) \in N_{S'}$.
    \end{proof}
  \end{proof}

To simplify notation, we define $w(x, s) = T'^x_{s.mappedClock^x}(s)$ and we omit the parameters, i.e. we write $w$ instead of $w(x,s)$.
\step{2}{$\A x \in \mathcal{X} : w.mappedClock^x \geq w.x - n^x(w)$}
  \begin{proof}
  \pf\ Let $x \in \mathcal{X}$ be an arbitrary but fixed.
  \step{2-1}{$n^x(w) \geq 0$}
    \begin{proof}
    \pf\ By definition of $n^x$, $n^x$ maps to values in $\mathbb{N}_0$ or $\infty$.
    \end{proof}
  \step{2-2}{$w.mappedClock^x = w.x$}
    \begin{proof}
    \pf\
    \step{2-2-1}{$w.x = s.mappedClock^x$}
      \begin{proof}
      \pf\ By definition of $T'$, $T'$ sets the value of $w.x$ to $s.mappedClock^x$.
      \end{proof}
    \step{2-2-2}{$w.mappedClock^x = s.mappedClock^x$}
      \begin{proof}
      \pf\ By definition of $T'$, $T'$ sets the value of $w.mappedClock^x$ to the respective value in $s$, i.e. $s.mappedClock^x$.
      \end{proof}
    \qedstep
      \begin{proof}
      \pf\ By \stepref{2-2-1} and \stepref{2-2-2}
      \end{proof}
    \end{proof}
  \qedstep
    \begin{proof}
    \pf\ From \stepref{2-1} follows that $0 \geq - n^x(w)$.
    Adding $w.x$ on both sides of the equation results in $w.x \geq w.x - n^x(w)$.
    The statement follows by replacing $w.x$ by $w.mappedClock^x$ because of \stepref{2-2}.
    \end{proof}
  \end{proof}
\end{proof}

Goal: We want to show that the specification $S'$ implements the time-optimized specification $\hat{S}$.
This means we show that the following theorem holds:

\begin{theorem}
\label{theo-spec-impl}
$Spec_{S'} \Rightarrow \overline{Spec_{\hat{S}}}$
\end{theorem}

\begin{proof}
In the proof, we use the shorthand $L_{S'}$ for $Liveness_{S'}$ and $L_{\hat{S}}$ for $Liveness_{\hat{S}}$.
\step{1}{$Init_{S'} \land \square [Next_{S'}]_{v_{S'}} \land L_{S'} \Rightarrow \overline{\square [Next_{\hat{S}}]_{v_{\hat{S}}} \land L_{\hat{S}}}$}
  \begin{proof}
  \step{1-1}{$Inv_{S'} \land [Next_{S'}]_{v_{S'}} \Rightarrow \overline{[Next_{\hat{S}}]_{v_{\hat{S}}}}$}
    \begin{proof}
    \step{1-1-1}{$Inv_{S'} \land (v_{S'}' = v_{S'}) \Rightarrow \overline{Next_{\hat{S}}} \lor (\overline{v_{\hat{S}}'} = \overline{v_{\hat{S}}})$}
      \begin{proof}
      \pf\ Because $Inv_{S'} \land (v_{S'}' = v_{S'}) \Rightarrow (v_{S'}' = v_{S'}) \Rightarrow (\overline{v_{\hat{S}}'} = \overline{v_{\hat{S}}})$.
      \end{proof}
    \step{1-1-2}{$Inv_{S'} \land Next_{S'} \Rightarrow \overline{Next_{\hat{S}}} \lor (\overline{v_{\hat{S}}'} = \overline{v_{\hat{S}}})$}
      \begin{proof}
      \step{1-1-2-1}{$Inv_{S'} \land AdvanceTime_{S'} \Rightarrow \overline{Next_{\hat{S}}} \lor (\overline{v_{\hat{S}}'} = \overline{v_{\hat{S}}})$}
        \begin{proof}
        \step{4-1}{\suffices{\assume{$Inv_{S'} \land AdvanceTime_{S'}$}
                             \prove{$\overline{Next_{\hat{S}}} \lor (\overline{v_{\hat{S}}'} = \overline{v_{\hat{S}}})$}}}
          \begin{proof}
          \pf\ Obvious.
          \end{proof}
        \step{4-2}{\case{$\A x \in \mathcal{X} : mappedClock'^{x} = mappedClock^x$}}
          \begin{proof}
          \pf\ In this case, $AdvanceTime_{S'}$ does not change the value of the variables $mappedClock^x$.
          Because, by definition, $AdvanceTime_{S'}$ leaves all variables except $mappedClock^x$ for all $x \in \mathcal{X}$ unchanged, 
          only the clocks are changed.
          As the values of the clocks do not affect the value of $\overline{v_{\hat{S}}}$, it follows that $\overline{v_{\hat{S}}'} = \overline{v_{\hat{S}}}$.
          \end{proof}
        \step{4-3}{\case{$\E x \in \mathcal{X} : mappedClock'^x > mappedClock^x$}}
          \begin{proof}
          \pf\ We refer to the step's starting state as $s$ and to the ending state as $t$.
          Let $\mathcal{Y} \subseteq \mathcal{X}$ be the set of clocks $y$ for which $mappedClock'^y > mappedClock^y$.
          We show that the step from state $\overline{s}$ to state $\overline{t}$ is an $AdvanceTime_{\hat{S}}$ step.
          To show that, we show that the properties of $AdvanceTime_{\hat{S}}$ defined in \cref{def_advance_time_hat} are fulfilled.
          \step{4-3-1}{$\A y \in \mathcal{Y} : \overline{t.y > s.y}$}
            \begin{proof}
            \pf\ Because $mappedClock'^y > mappedClock^y$ which can also be written as $t.mappedClock^y > s.mappedClock^y$ and $t.mappedClock^y = \overline{t.y}$.
            \end{proof}
          \step{4-3-1b}{$\A x \in \mathcal{X} \setminus \mathcal{Y} : \overline{t.x = s.x}$}
            \begin{proof}
            \pf\ By definition of $AdvanceTime_{S'}$, the value of $mappedClock^x$ cannot decrease.
            Because $x \notin \mathcal{Y}$, it follows that $t.mappedClock^x = s.mappedClock^x$ which equals $\overline{t.x = s.x}$.
            \end{proof}
          \step{4-3-2}{$\A y \in \mathcal{Y} : \overline{t.y \in relETP^y(s)}$}
            \begin{proof}
            \pf\ By definition of $AdvanceTime_{S'}$, it holds that $mappedClock'^y = max( \{mappedClock^y\} \cup \{t \in relETP^y : t \leq y' \})$.
            Because $mappedClock'^y > mappedClock^y$, it follows that $mappedClock'^y = max( \{t \in relETP^y : t \leq y' \})$.
            We can also write this as $t.mappedClock^y = max( \{t \in relETP^y(s) : t \leq t.y \}) \in relETP^y(s)$.
            \end{proof}
          \step{4-3-3}{$\overline{\A v \in v_{\hat{S}} : t.v = s.v}$}
            \begin{proof}
            \pf\ For all $v \in v_{\hat{S}}$ it holds that $v' = v$ and, because $\overline{v} = v$ it holds that $\overline{v' = v}$.
            \end{proof}
          \step{4-3-4}{$\overline{\A x \in \mathcal{X} : \A b \in B^x(s) : b \geq s.x \Rightarrow t.x \leq b}$}
            \begin{proof}
            \pf\
            Let $x \in \mathcal{X}$ be arbitrary but fixed.
            Let $b \in \overline{B^x(s)}$ be arbitrary but fixed.

            In this proof, we assume that $\overline{b \geq s.x}$ and we show that it follows that $\overline{t.x \leq b}$.
            The assumption can also be written as $b \geq s.mappedClock^x$.

            \step{4-3-4-2-1}{$b \geq s.x$}
              \begin{proof}
              \pf\ We prove this by contradiction.
              Assume that $b < s.x$.
              Then $(b+1) \leq s.x$ and $s.mappedClock^x \geq (b+1) > b$ because $(b+1) \in relETP^x(s)$.
              This is a contradiction to $b \geq s.mappedClock^x$.
              \end{proof}
            
            By \cref{assumption-b-independent}, it holds that $\A d \in \mathbb{N}_0, s \in \hat{\Sigma} : B^x(s) = B^x(\hat{T}^x_d(s))$ which means that $B^x(s) = \overline{B^x(s)}$.
            By $AdvanceTime_{S'}$, it holds for all $b' \in B^x(s)$ that $s.x \leq b' \Rightarrow t.x \leq b'$.
            Because $B^x(s) = \overline{B^x(s)}$, it also holds for $b$ that $s.x \leq b \Rightarrow t.x \leq b$.
            
            From \stepref{4-3-4-2-1} it follows that $t.x \leq b$.

            By definition, it holds that $t.mappedClock^x \leq t.x$ and, therefore,  $t.mappedClock^x \leq b$.

            This can also be written as $\overline{t.x \leq b}$
            \end{proof}
          \qedstep
            \begin{proof}
            By \stepref{4-3-1}, \stepref{4-3-2}, and \stepref{4-3-3} it holds that $\overline{AdvanceTime_{\hat{S}}}$.
            \end{proof}
          \end{proof}
        \qedstep
          \begin{proof}
          \pf\ Because, by definition of $AdvanceTime_{S'}$, it holds that $AdvanceTime_{S'} \Rightarrow mappedClock'^x \geq mappedClock^x$, the steps \stepref{4-2} and \stepref{4-3} cover all possible cases. Thereby and by \stepref{4-1} follows the statement to prove.
          \end{proof}
        \end{proof}
      \step{1-1-2-2}{$Inv_{S'} \land NextI_{S'} \Rightarrow \overline{Next_{\hat{S}}} \lor (\overline{v_{\hat{S}}'} = \overline{v_{\hat{S}}})$}
        \begin{proof}
        \pf\ We refer to the step's starting state as $s$ and to the ending state as $t$.
        Informally speaking, we show that every $NextI$ step that is allowed in specification $S'$ is also allowed in specification $\hat{S}$.
        \step{5-1}{\suffices{\assume{$Inv_{S'} \land NextI_{S'}$}
                             \prove{$\overline{Next_{\hat{S}}} \lor (\overline{v_{\hat{S}}'} = \overline{v_{\hat{S}}})$}}}
          \begin{proof}
          \pf\ Obvious.
          \end{proof}
        \step{5-5}{$\A x \in \mathcal{X} : \A \sigma \in N_{S'} : \sigma_0 = s \Rightarrow \mathcal{T}'^{x}_{s.mappedClock^x}(\sigma) \in N_{S'}$}
          \begin{proof}
          \pf\ By \cref{lemma-inv-applied} because $s$ is a state of specification $Spec_{S'}$.
          \end{proof}
        \step{5-6}{$\A \sigma \in N_{S'} : \sigma_0 = s \Rightarrow \bigcomp_{x \in \mathcal{X}}(\mathcal{T}'^{x}_{s.mappedClock^x})(\sigma) \in N_{S'}$}
          \begin{proof}
          \pf\ By \cref{lemma-inv-applied}, each behavior in $N_{S'}$ that starts in a state of specification $Spec_{S'}$, is mapped by $\mathcal{T}'^x$ to a behavior in $N_{S'}$ that starts in a state of specification $Spec_{S'}$.
          Therefore, executing the mappings for each clock consecutively, results in a behavior in $N_{S'}$ that starts in a state of specification $Spec_{S'}$.
          \end{proof}
        \step{5-7}{$\astep{\bigcomp_{x \in \mathcal{X}}(T'^x_{s.mappedClock^x})(s)}{\bigcomp_{x \in \mathcal{X}}(T'^x_{s.mappedClock^x})(t)}{NextI_{S'}}$}
          \begin{proof}
          \pf\ By \stepref{5-6} and because a behavior $\sigma$ that starts with the state $s$ followed by state $t$ is in $N_{S'}$.
          \end{proof}
        \step{5-8}{$\astep{\overline{s}}{\overline{t}}{NextI_{S'}}$}
          \begin{proof}
          \pf\ By \stepref{5-7} and because, by definition, $\overline{s} = \bigcomp_{x \in \mathcal{X}}(T'^x_{s.mappedClock^x})(s)$.
          \end{proof}
        \step{5-9}{$\overline{NextI_{S'}}$}
          \begin{proof}
          \pf\ By \stepref{5-8} written differently.
          \end{proof}
        \qedstep
          \begin{proof}
          \pf\ Because $NextI_{S'} = NextI_{\hat{S}}$ and because $Next_{\hat{S}} = NextI_{\hat{S}} \lor AdvanceTime_{\hat{S}}$, it follows from \stepref{5-9} that $\overline{Next_{\hat{S}}}$ and also that $\overline{Next_{\hat{S}}} \lor (\overline{v_{\hat{S}}'} = \overline{v_{\hat{S}}})$.
          \end{proof}
        \end{proof}
      \qedstep
        \begin{proof}
        \pf\ By \stepref{1-1-2-1} and \stepref{1-1-2-2} and because $Next_{S'} = AdvanceTime_{S'} \lor NextI_{S'}$.
        \end{proof}
      \end{proof}
    \qedstep
      \begin{proof}
      \pf\ By \stepref{1-1-1} and \stepref{1-1-2} and because $[Next_{S'}]_{v_{S'}} = Next_{S'} \lor (v_{S'}' = v_{S'})$.
      \end{proof}
    \end{proof}
  \qedstep
    \begin{proof}
    \pf\ By \stepref{1-1} and \cref{lemma-inv} and because $L_{S'} = L_{\hat{S}}$.
    \end{proof}
  \end{proof}
\step{2}{$Init_{S'} \land \square [Next_{S'}]_{v_{S'}} \land L_{S'} \Rightarrow \overline{Init_{\hat{S}} \land \square [Next_{\hat{S}}]_{v_{\hat{S}}} \land L_{\hat{S}}}$}
  \begin{proof}
  \pf\ By \stepref{1} and \cref{lemma-init}.
  \end{proof}
\qedstep
  \begin{proof}
  \pf\ By definitions of $Spec_{S'}$ and $Spec_{\hat{S}}$.
  \end{proof}
\end{proof}

\subsection{Proof for Application of Generalized Time Skip Theorem to Formalization of Lightning}
\label{sec-appendix-time-skip-application-I-II}

To prove that \cref{theorem-general-time-skip} can be applied to the formalization of Lightning, we need to prove that, for each clock $x \in \mathcal{X}$, the set $relETP^x$ defined in the specification meets the requirements of \cref{assumptions-relETP} and that the time bounds used in the specification meet the requirements of \cref{assumption-b-independent}. 

We start by proving that the time bounds specified in SpecificationI.tla meet \cref{assumption-b-independent}.
\begin{proof}
\pf\ $B^{LedgerTime}$ is defined by $TimeBounds$ of $PaymentChannelUser$.
As the variable $LedgerTime$ does not occur in the definition of $TimeBounds$, the definition of $TimeBounds$ is independent of $LedgerTime$ and, thus, \cref{assumption-b-independent} holds for $B^{LedgerTime}$.

Similarly, $B^{TxAge}$ is defined by $TxTimeBounds$ of $PaymentChannelUser$.
As the variable $TxAge$ does not occur in the definition of $TxTimeBounds$, the definition of $TxTimeBounds$ is independent of $TxAge$ and, thus, \cref{assumption-b-independent} holds for $B^{TxAge}$.

In conclusion, \cref{assumption-b-independent} holds for all clocks of the specification.
\end{proof}

We prove that $relETP^x$ meets \cref{assumptions-relETP}.
The set $relETP^{LedgerTime}$ is defined as $relETP$ in SpecificationII.

\begin{proof}
\pf\ 
First, we prove the second statement:
\step{1}{$\A x \in \mathcal{X}, s \in \Sigma : \{ b +1 \mid b \in B^x(s) \} \subseteq relETP^x(s)$}
  \begin{proof}
  \pf\
  For $LedgerTime$, the statement directly follows from the definition of $relETP$ in SpecificationII that defines $relETP$ as union of $\{t + 1 : t \in TimeBounds\}$ and another set.
  For the $TxAge$ clocks, the set of time bounds can only be the empty set or $\{TO\_SELF\_DELAY - 1\}$. 
  Because the set $relETP^{TxAge}$ contains $TO\_SELF\_DELAY$, it holds for all states $s \in \Sigma$ that $\{ b +1 \mid b \in B^{TxAge}(s) \} = \{TO\_SELF\_DELAY\} \subseteq relETP^{TxAge}(s)$.
  \end{proof}
\step{2}{$\A x \in \mathcal{X}, s \in \Sigma : ETP^x(s) \subseteq relETP^x(s)$}
  \begin{proof}
  \pf\ We have to show that for every $n \in ETP^x(s)$, it holds that $n \in relETP^x(s)$.
  This means that if there is a newly possible behavior at time $n$, then time $n$ must be in $relETP^x(s)$.

  Our proof strategy is based upon the following observation:
  For each $n \in ETP^x(s)$, there exists a newly possible behavior $\sigma$ for clock $x \in \mathcal{X}$ and time $n \in \mathbb{N}$.
  By definition of $ETP^x(s)$, the behavior $\sigma$ contains a step $\sstep{s}{t}$ that is not possible at time $n-1$, i.e., the step $\sstep{\hat{T}^x_{n-1}(s)}{\hat{T}^x_{n-1}(t)}$ is not allowed in the specification.
  In the step $\sstep{s}{t}$, the clock $x$ is unchanged because $\sigma \in N^x_{\hat{S}}$.
  Define a mapping $a$ from $ETP^x(s)$ to the subactions of the specification, that assigns to each $n \in ETP^x(s)$ a subaction $A$ so that $\astep{s}{t}{A}$ but not $\astep{\hat{T}^x_{n-1}(s)}{\hat{T}^x_{n-1}(t)}{A}$, i.e., a step of action $a(n)$ is not possible if clock $x$ is set to $n-1$ and becomes possible at time $n$.
  For the step $\sstep{\hat{T}^x_{n-1}(s)}{\hat{T}^x_{n-1}(t)}$, not to be possible, there must be a condition that contains the clock $x$ because the value of the clock $x$ is the only difference between the states $\hat{T}^x_{n-1}(s)$ and $s$.
  Therefore, to find all $n \in ETP^x(s)$, we inspect all subactions of the $NextI$ action of SpecificationI and select the actions that depend on the value of the clock $x$.
  We need to prove for each subaction of $NextI$ that all points in time at which a step of this action becomes possible, are included in $relETP^x$.
  Therefore, our strategy is as follows:
  We go through all subactions $A$ of $NextI$, we find the conditions under which an $A$ step exists that is enabled at time $n$ but not at time $n-1$, and verify that in all states $s$ from which a state that meets these conditions can be reached by $NextI$ steps it holds that $n \in relETP^x(s)$.

  Actions of module PaymentChannelUser:
  There are nine actions in the module that depend on the value of $LedgerTime$ 
  and three actions that depend on $TxAge$. 

  \step{2-1}{$\A s \in \Sigma : \A n \in ETP^{LedgerTime}(s) : a(n) = \mathrm{SendSignedCommitment} \Rightarrow n \in relETP^{LedgerTime}(s)$}
    \begin{proof}
    The action SendSignedCommitment commits to new HTLCs by including them in the new commitment transaction that is sent. Outgoing HTLCs are included if they are in state \textsc{new} and if their timelock is greater than the current value of $LedgerTime$. For two consecutive points in time $n-1$ and $n$, a step of the action SendSignedCommitment is possible at time $n$ but not at time $n-1$ if a SendSignedCommitment step is allowed at time $n-1$ that adds an HTLC but the same HTLC could not be added at time $n$ because $n-1$ is smaller than the HTLC's timelock but $n$ is equal to the HTLC's timelock.
    It follows that a long as there is an HTLC that is in state \textsc{new}, the set $relETP^{LedgerTime}(s)$ must include the HTLC's timelock because at this point in time a new step might become possible that sends a commitment transaction without including this HTLC.
    This condition is fulfilled by the definition of TimelockRegions of the module PaymentChannelUser.

    Further, the set $relETP^{LedgerTime}(s)$ must include the timelock of each HTLC that might come into state \textsc{new}. The only way for an HTLC to reach the state \textsc{new} is to be created by the action AddAndSendOutgoingHTLC in the module HTLCUser. This action creates an HTLC for an outgoing payment with the payment's timelock.
    Thus, the set $relETP^{LedgerTime}(s)$ must include the timelock of each outgoing payment. 
    This condition is fulfilled by the definition of TimelockRegions of the module HTLCUser.
    
    An outgoing payment is created if an incoming HTLC is received that is to be forwarded. The timelock of the newly created payment is taken from the dataForNextHop field in the message that was prepared by the original sender of the payment.
    Thus, the set $relETP^{LedgerTime}(s)$ must include all the timelocks included in the dataForNextHop fields for payments that are already in process and, for multi-hop payments that are not yet in process, the set $relETP^{LedgerTime}(s)$ must include the timelocks as they will appear in the dataForNextHop field.
    This condition is fulfilled by the definition of TimelockRegions of the module HTLCUser.
    It is possible to predict these timestamps because they are calculated only based on a payment's timelock and the payment's route which are known already in the initial states of the specification.

    Given that all these values are included in the set $relETP^{LedgerTime}(s)$, the set $relETP^{LedgerTime}(s)$ contains all points in time at which a SendSignedCommitment step becomes possible because the action SendSignedCommitment depends on LedgerTime only for choosing the HTLCs to add.
    \end{proof}
  \step{2-2}{$\A s \in \Sigma : \A n \in ETP^{LedgerTime}(s) : a(n) = \mathrm{ReceiveSignedCommitment} \Rightarrow n \in relETP^{LedgerTime}(s)$}
    \begin{proof}
    Analogously to \stepref{2-1}, ReceiveSignedCommitment accepts a commitment transaction that is being received only if the incoming HTLCs that are committed have a timelock that is greater than LedgerTime.
    An incoming HTLC can be committed only if it is in state \textsc{new}.
    As described above, the set $relETP^{LedgerTime}(s)$ contains an HTLC's timelock already for all HTLCs that are or can come into state \textsc{new}.
    Under this condition, the set $relETP^{LedgerTime}(s)$ contains all points in time at which a ReceiveSignedCommitment step becomes possible because the action ReceiveSignedCommitment depends on LedgerTime only for deciding which HTLCs are allowed to be added in a commitment transaction being received.
    \end{proof}
  \step{2-3}{$\A s \in \Sigma : \A n \in ETP^{LedgerTime}(s) : a(n) = \mathrm{ReceiveRevocationKey} \Rightarrow n \in relETP^{LedgerTime}(s)$}
    \begin{proof}
    The action ReceiveRevocationKey depends on LedgerTime for defining which HTLCs are marked as \textsc{committed}.
    By definition of ReceiveRevocationKey, an incoming HTLC is only committed if it is in state \textsc{sent-commit} and has a timelock greater than LedgerTime.
    As described above, the set $relETP^{LedgerTime}(s)$ contains an HTLC's timelock already for all HTLCs that are or can come into state \textsc{new}. Additionally, the set $relETP^{LedgerTime}(s)$ must also contain an HTLC's timelock if the HTLC is in state \textsc{sent-commit} or in one of the states \textsc{recv-commit} and \textsc{pending-commit} which are preceding states of \textsc{sent-commit}.
    This condition is fulfilled by the definition of TimelockRegions of the module PaymentChannelUser.
    Under this condition, the set $relETP^{LedgerTime}(s)$ contains the timelocks of all HTLCs that are in state \textsc{sent-commit} or can come into this state.
    \end{proof}
  \step{2-4}{$\A s \in \Sigma : \A n \in ETP^{LedgerTime}(s) : a(n) = \mathrm{CloseChannel} \Rightarrow n \in relETP^{LedgerTime}(s)$}
    \begin{proof}
    The action CloseChannel chooses a transaction for closing based on LedgerTime.
    The channel can always be closed using the latest commitment transaction.
    The channel can also be closed using a pending new commitment transaction unless there is an HTLC that is in state \textsc{recv-commit} and has timed out, i.e., the HTLC has a timelock lower than or equal to the value of LedgerTime.
    A step of CloseChannel that closes using a pending new commitment transaction becomes possible if a timedout HTLC that is in state \textsc{recv-commit} is advanced to another state.
    Because the timelocks of all HTLCs that are in state \textsc{recv-commit} or any preceding state are included in the set $relETP^{LedgerTime}(s)$, the set $relETP^{LedgerTime}(s)$ contains all points in time at which a step of CloseChannel becomes possible.
    \end{proof}
  \step{2-5}{$\A s \in \Sigma : \A n \in ETP^{LedgerTime}(s) : a(n) \in \{\mathrm{Cheat}, \mathrm{Punish}, \mathrm{RedeemHTLCAfterClose}\} \Rightarrow n \in relETP^{LedgerTime}(s)$}
    \begin{proof}
    The actions Cheat, Punish, and RedeemHTLCAfterClose all use a formula that finds transactions that can be used to spend outputs of published transactions and one of the transactions found is published.
    This formula depends on $LedgerTime$ as well as $TxAge$ because outputs might be timelocked and only spendable after an absolute timelock or after a certain time after the transaction containing the output was published.
    If $LedgerTime$ reaches the absolute timelock of an output or $TxAge$ reaches the relative timelock of an output, there is a new transaction in the set of publishable transactions and, thus, a step publishing this new transaction becomes possible.
    Consequently, the set $relETP^{LedgerTime}(s)$ must contain absolute timelocks of outputs that might be spent by such a new transaction.
    This condition is fulfilled by the definition of TimelockRegions of the module PaymentChannelUser because TimelockRegions contains all absolute timelocks of outputs in the ledger.

    To account for the points in time at which an output of a transaction becomes spendable that has not been published yet, the set $relETP^{LedgerTime}(s)$ must also include the absolute timelocks of transactions that are stored in the users' Inventory variables.
    This condition is fulfilled by the definition of TimelockRegions of the module PaymentChannelUser.
    Such absolute timelocks in transactions in the users' Inventory variables are created based on the absolute timelocks of uncommitted HTLCs. As shown in the steps above, the timelocks of these HTLCs are included by the definition of TimelockRegions of the module PaymentChannelUser.
    \end{proof}
  \step{2-8}{$\A s \in \Sigma : \A n \in ETP^{LedgerTime}(s) : a(n) = \mathrm{NoteThatHTLCFulfilledOnChainByOtherUser} \Rightarrow n \in relETP^{LedgerTime}(s)$}
    \begin{proof}
    The action NoteThatHTLCFulfilledOnChain finds an HTLC that has been fulfilled on-chain, adds the preimage to the user's variables, and marks the HTLC as \textsc{persisted}.
    If the HTLC that has been fulfilled has timed out for longer than the grace period, i.e., LedgerTime is greater than the HTLC's timelock plus the duration of the grace period,
    the HTLC's preimage is added to the set LatePreimages.
    A step of NoteThatHTLCFulfilledOnChain in which a preimage is added to the set LatePreimages becomes possible if LedgerTime equals the HTLC's timelock $+ G + 1$ where $G$ is the duration of the grace period.
    Therefore, the set $relETP^{LedgerTime}(s)$ must contain for each HTLC that is in a state that can be followed by state \textsc{persisted} the value of the HTLC's timelock $+ G + 1$.
    This condition is fulfilled by the definition of TimelockRegions of the module PaymentChannelUser and HTLCUser because for every HTLC that is in a state that can be followed by the state \textsc{persisted} or that can be created, these sets contain the value of the HTLC's timelock $ + G + 1$.
    \end{proof}
  \step{2-9}{$\A s \in \Sigma : \A n \in ETP^{LedgerTime}(s) : a(n) = \mathrm{NoteThatHTLCFulfilledOnChainInOtherChannel} \Rightarrow n \in relETP^{LedgerTime}(s)$}
    \begin{proof}
    With respect to the use of $LedgerTime$, the definition of NoteThatHTLCFulfilledOnChainInOtherChannel equals the definition of NoteThatHTLCFulfilledOnChainByOtherUser and, thus, the conclusion follows analogously to \stepref{2-8}.
    \end{proof}
  \step{2-10}{$\A s \in \Sigma : \A n \in ETP^{TxAge}(s) : a(n) \in \{\mathrm{Cheat}, \mathrm{Punish}, \mathrm{RedeemHTLCAfterClose}\} \Rightarrow n \in relETP^{TxAge}(s)$}
    \begin{proof}
    The actions Cheat, Punish, and RedeemHTLCAfterClose all use a formula that finds transactions that can be used to spend outputs of published transactions and one of the transactions found is published.
    This formula depends on $LedgerTime$ as well as $TxAge$ because outputs might be timelocked and only spendable after an absolute timelock or after a certain time after the transaction containing the output was published.
    If $TxAge$ reaches the relative timelock of an output, there is a new transaction in the set of publishable transactions and, thus, a step publishing this new transaction becomes possible.
    Consequently, the set $relETP^{TxAge}(s)$ must contain the relative timelocks of outputs that might be spent by such a new transaction.
    Because all relative timelocks in the specification equal the constant TO\_SELF\_DELAY, this condition is fulfilled by the definition of AdvanceLedgerTime which assumes that $relETP^{TxAge}(s) = \{ TO\_SELF\_DELAY \}$.
    \end{proof}

  The module HTLCUser contains five actions that depend on LedgerTime and no actions that depend on TxAge:
  
  \step{2-20}{$\A s \in \Sigma : \A n \in ETP^{LedgerTime}(s) : a(n) = \mathrm{AddAndSendOutgoingHTLC} \Rightarrow n \in relETP^{LedgerTime}(s)$}
    \begin{proof}
    The action AddAndSendOutgoingHTLC adds a payment with the lowest timelock of all payments whose timelock is greater than the value of LedgerTime.
    Thus, a new step of AddAndSendOutgoingHTLC becomes possible if a payment cannot be added anymore because the value of LedgerTime equals the payment's timelock and a payment with the next upcoming timelock becomes the payment with the lowest timelock of all payments whose timelock is greater than the value of LedgerTime.
    Because the definition of TimelockRegions of the module HTLCUser contains the timelock of each existing or future payment, all points in time at which a new step of AddAndSendOutgoingHTLC becomes possible are included in $relETP^{LedgerTime}(s)$.
    \end{proof}
  \step{2-21}{$\A s \in \Sigma : \A n \in ETP^{LedgerTime}(s) : a(n) = \mathrm{SendHTLCPreimage} \Rightarrow n \in relETP^{LedgerTime}(s)$}
    \begin{proof}
    The action SendHTLCPreimage depends on LedgerTime and sends the preimage to an HTLC only if the HTLC's timelock plus the constant grace period $G$ is greater than LedgerTime.
    By this dependency, a SendHTLCPreimage step cannot become possible because, if a SendHTLCPreimage is possible at one point in time, then at any preceding point in time is also smaller than the HTLC's timelock plus the grace period $G$.
    \end{proof}
  \step{2-22}{$\A s \in \Sigma : \A n \in ETP^{LedgerTime}(s) : a(n) = \mathrm{ReceiveHTLCPreimage} \Rightarrow n \in relETP^{LedgerTime}(s)$}
    \begin{proof}
    Analogously to the action NoteThatHTLCFulfilledOnChainByOtherUser in the module PaymentChannelUser, the action ReceiveHTLCPreimage adds the preimage of an HTLC to the variable LatePreimages if the HTLC's timelock $+ G$ is smaller than LedgerTime.
    Thus, a new step becomes possible at an HTLC's timelock $+ G + 1$.
    Because the set $relETP^{LedgerTime}(s)$ contains the HTLC's timelock $+ G + 1$ for every HTLC that can be fulfilled, this condition is fulfilled.
    \end{proof}
  \step{2-23}{$\A s \in \Sigma : \A n \in ETP^{LedgerTime}(s) : a(n) = \mathrm{SendHTLCFail} \Rightarrow n \in relETP^{LedgerTime}(s)$}
    \begin{proof}
    The action SendHTLCFail depends on LedgerTime to fail an HTLC that has been committed if the HTLC has timed out.
    Thus, a SendHTLCFail becomes possible if the value of LedgerTime equals the timelock of an HTLC that is in state \textsc{committed}.
    As the set $relETP^{LedgerTime}(s)$ contains the timelock of all HTLCs that are committed and can come into state \textsc{committed}, the set $relETP^{LedgerTime}(s)$ contains all points in time at which a step of SendHTLCFail becomes possible.
    \end{proof}
  \step{2-24}{$\A s \in \Sigma : \A n \in ETP^{LedgerTime}(s) : a(n) = \mathrm{ReceiveHTLCFail} \Rightarrow n \in relETP^{LedgerTime}(s)$}
    \begin{proof}
    The action ReceiveHTLCFail marks an HTLC only as failed if the HTLC's timelock is smaller than or equal to LedgerTime.
    This is the same condition as for the action SendHTLCFail and, by \stepref{2-23}, the set $relETP^{LedgerTime}(s)$ contains all points in time at which a step of ReceiveHTLCFail becomes possible.
    \end{proof}

  The module MultiHopMock contains one action that depends on LedgerTime and no actions that depend on $TxAge$.
  \step{2-30}{$\A s \in \Sigma : \A n \in ETP^{LedgerTime}(s) : a(n) = \mathrm{ReceivePreimageForIncomingHTLC} \Rightarrow n \in relETP^{LedgerTime}(s)$}
    \begin{proof}
    The action ReceivePreimageForIncomingHTLC mocks the reception of a preimage of an HTLC.
    The action describes steps that add the preimage to the variable LatePreimages and steps that do not add the preimage to LatePreimages.
    The preimage must be added to LatePreimages if the HTLC's timelock is smaller than or equal to LedgerTime.
    Thus, at an HTLC's timelock, a step that does not add the preimage to LatePreimages becomes impossible but no new step becomes possible because a step that adds the preimage to LatePreimages is already possible when LedgerTime is smaller than the HTLC's timelock.
    \end{proof}

  \end{proof}
\end{proof}

\subsection{Proof that Specification $(II)$ refines Specification $(III)$}
\label{sec-appendix-proofs}

Let $C$ be the set of all channels that exist in a specification and $U$ be the set of all users of a specification and $U_c$ the users that are part of channel $c \in C$.

The states of specification $(I)$ are defined by the following variables. The variables that are specific for a user and/or a channel are prefixed with an identifier of the specific user and/or channel.
In the proofs, we omit the prefixes when they are irrelevant or it is clear from the context which user or channel is referred to.
For each user $u \in U$: $v_u = \{$ $u$PreimageInventory, $u$LatePreimages, $u$PaymentSecretForPreimage, $u$NewPayments, $u$Payments, $u$ChannelBalance, $u$ExtBalance, $u$Honest $\}$.
For each channel $c \in C$: $v_c = \{$ $c$Messages, $c$UsedTransactionIds, $c$PendingBalance $\}$.
For each channel $c \in C$ and for each user $u$ of channel $c$: $v_{c,u} = \{$ $c,u$State, $c,u$Balance, $c,u$Vars, $c,u$DetailVars, $c,u$Inventory $\}$.
Global variables: $v_\mathrm{g} = \{$ LedgerTime, TxAge, Messages, LedgerTx $\}$.

Specification $(II)$ has the same variables as specification $(I)$.
Specification $(IIa)$ has the variables that specification $(II)$ has and additionally a variable $u$RequestedInvoices for each user $u$ in which data about the mocked environment is stored. Specifically, it is stored for which payments an invoice was already requested.

The states of specification $(III)$ are defined by the following variables:
For each user $u$: $v_u = \{$ $u$PreimageInventory, $u$LatePreimages, $u$PaymentSecretForPreimage, $u$NewPayments, $u$Payments, $u$ChannelBalance, $u$ExtBalance, $u$Honest $\}$
For each channel $c \in C$: $v_c = \{$ $c$Messages, $c$PendingBalance $\}$.
For each channel $c$ and for each user $u$ of channel $c$: $v_{c,u} = \{$ $c,u$State, $c,u$Balance, $c,u$Vars $\}$
Global variables: $v_\mathrm{g} = \{$ LedgerTime, Messages $\}$

\subsubsection{$(IIa) \Rightarrow (IIIa)$}
\label{sec-appendix-proof-IIa-IIIa}

\begin{definition}
Define $f_c$ to be the refinement mapping from specification $(IIa)$ to specification $(IIIa)$ for channel $c$.
The refinement mapping $f_c$ is defined using TLA\textsuperscript{+} in the file SpecificationIIatoIIIa.tla.
\end{definition}

\begin{lemma}
\label{lem-ref-iia-iiia}
\label{lem-refinement-mapping}
Specification $(IIa)$ for channel $c \in C$ is a refinement of specification $(IIIa)$ with the refinement mapping $f_c$.
\end{lemma}

We check the correctness of this lemma using model checking (see \cref{sec-model-checking-results}).
We model check specification $(IIa)$ and verify that applying the mapping to each state leads to steps of specification $(IIIa)$.

\begin{lemma}
\label{lem-mapping-module-steps-in-fc}
The refinement mapping $f_c$ maps steps of PaymentChannelUser to steps of IdealChannel,
steps of HTLCUser to steps of HTLCUser,
steps of MultiHopMock to steps of MultiHopMock,
and steps of LedgerTime to steps of LedgerTime or stuttering steps.
\end{lemma}

\begin{proof}
\pf\ 
By design, steps of an action $A$ of HTLCUser and MultiHopMock are mapped to steps of the same action $A$ in specification $(IIIa)$.
This mapping of steps to steps of the same action is implemented in the refinement mapping by leaving the variables or the fields of variables unchanged by the refinement mapping that are updated by the actions of HTLCUser and MultiHopMock.
Because the values of the refinement mapping of all variables except LedgerTime do not depend on LedgerTime, these values are unchanged, if the value of LedgerTime changes.
Because the refinement mapping does not change the value of LedgerTime and no action in another module describes a change in LedgerTime, a step of LedgerTime is mapped to a step of LedgerTime.
Steps of the module PaymentChannelUser are mapped to steps of the module IdealChannel.
\end{proof}

\subsubsection{$(II) \Rightarrow (III)$}
\label{sec-appendix-proof-II-III}

We prove that specification $(II)$ is a refinement of specification $(III)$ by defining a mapping $g$ from the state space of specification $(II)$ to the state space of specification $(III)$ and by proving that the mapping $g$ is a refinement mapping.
The mapping $g$ works as follows:
In a first step, the mapping $g$ defines for each channel $c$ in specification $(II)$ a state of specification $(IIa)$. 
This state of specification $(IIa)$ is mapped to a state of specification $(IIIa)$ using the refinement mapping $f_c$.
The mapping $g$ assigns to a state of specification $(II)$ a state of specification $(III)$ with the channel-specific variables of specification $(IIIa)$ for each channel $c$ and the user-specific and general variables of the given state of specification $(II)$.

An instance of specification $(IIa)$ defines the behavior of one single channel.
To define an instance of specification $(IIa)$ for each channel $c \in C$ we use the function $m_c$ to define a state of specification $(IIa)$ based on a state of specification $(II)$.
To this end, specification $(IIa)$ has the same global variables $v_\mathrm{g}$ used by specification $(II)$ but only the variables $v_{u}$ for the state of the two participating users of the channel $c$ and the variables $v_{c,u}$ and $v_c$ that are specifically related to the channel $c$.

\begin{definition}[$m_c$]
\label{definition-m_c}
Define a mapping $m_c: \Sigma_{(II)} \rightarrow \Sigma_{(IIa)}$ so that for a state $s \in \Sigma_{(II)}$, the state $m_c(s)$ assigns to all variables in $\Sigma_{(IIa)}$ the following values:

Each channel-specific variable $v \in v_{c,u} \cup v_c$ in specification $(IIa)$ is assigned the value of the corresponding variable in state $s$ with $c$ prepended to the variable's name. E.g., the variable ChannelMessages in specification $(IIa)$ is assigned the value $s.c\mathrm{Messages}$.

The user-specific variables $v \in v_u$ in specification $(IIa)$ are assigned the following values.
As a helper we define the set of relevant preimages $R_{u,c}$ for each user $u$ of the channel $c$ as the domain of the function $u$PaymentSecretForPreimage combined with the set of all preimages for hashes of HTLCs stored in $c,u$Vars.
\begin{itemize}
  \item $u$PreimageInventory is assigned the value of $s.u\mathrm{PreimageInventory} \cap R_{u,c}$.
  \item $u$LatePreimages is assigned the value of $s.u\mathrm{LatePreimages} \cap R_{u,c}$.
  \item $u$PaymentSecretForPreimage is assigned the value of $s.u\mathrm{PaymentSecretForPreimage}$.
  \item $u$NewPayments is assigned the set of all payments in $s.u\mathrm{NewPayments}$ that are initiated by user $u$, for which the next hop is channel $c$, or for which there exists an incoming HTLC in $c,u$Vars.
  \item $u$ChannelBalance and $u$Payments are assigned the corresponding values in state $s$.
  \item $u$Honest and $u$ExtBalance are assigned the corresponding values in state $s$.
  \item $u$RequestedInvoices is assigned the set of payments of users that are not part of channel $c$ for which an invoice was requested from user $u$.
\end{itemize}

The global variables in specification $(IIa)$ are assigned the following values:
\begin{itemize}
  \item Messages is assigned a set that contains all messages in $s.\mathrm{Messages}$ for which either the sender or the recipient is a user that is member of channel $c$.
  \item LedgerTx is the set of all transactions in $s.LedgerTx$ that are related to channel $c$, i.e., the transaction with the funding output and any directly or indirectly spending transactions.
  \item LedgerTime is assigned the greatest value of the set $c$TimelockRegions that is lower than
  $s.LedgerTime$. The set $c$TimelockRegions consists of the value that was assigned to LedgerTime for channel $c$ in the previous step\footnote{This value is obtained by adding a helper variable for each channel $c$ to specification $(II)$ that stores the mapped value of LedgerTime for channel $c$.} combined with the union of the TimelockRegions sets of PaymentChannelUser, HTLCUser, and MultiHopMock for the users of channel $c$ .
\end{itemize}

\end{definition}

\begin{lemma}
\label{lemma-multi-hop-mock-refinement}
Given a state of specification $(II)$, then in all channels $c \in C$ created by the mapping $m_c$ it holds that for every step $\sstep{s}{t}$ of an action of an instance of PaymentChannelUser or HTLCUser for channel $c$, 
in all channels $c' \neq c$ the step $\sstep{m_{c'}(s)}{m_{c'}(t)}$  is either a stuttering step, a step of an instance of MultiHopMock for channel $c'$ or a or a step of HTLCUser for channel $c'$ and the same user as in channel $c$.
\end{lemma}

\begin{proof}
Let $\sstep{s}{t}$ be a step in specification $(II)$ of an action of the modules PaymentChannelUser or HTLCUser for channel $c$.

The variables $v_{c'}$ and $v_{c',u}$ that are specific to channel $c'$ and the variables of users $u$ who are not part of channel $c$ are unaffected by a step in channel $c$ and, thus, these variables are unchanged in the step $\sstep{m_{c'}(s)}{m_{c'}(t)}$.

Variables that are left to discuss because they are possibly changed in the step $\sstep{m_{c'}(s)}{m_{c'}(t)}$ are the global variables $v_\mathrm{g}$ and the user specific variables $v_u$ for users who are both in channel $c'$ and channel $c$.
These variables are: Messages, LedgerTx, LedgerTime, TxAge, $u$PreimageInventory, $u$LatePreimages, $u$PaymentSecretForPreimage, $u$ChannelBalance, $u$Payments, $u$NewPayments, $u$Honest, $u$ExtBalance, and $u$RequestedInvoices

The variables LedgerTx and TxAge are reduced by $m_{c'}$ to values that are specific to channel $c'$. Thus, a change in channel $c$ does not affect the value of LedgerTx and TxAge in the step $\sstep{m_{c'}(s)}{m_{c'}(t)}$ and these variables are unchanged.
The variable LedgerTime cannot be changed in step $\sstep{s}{t}$ which is a step of an action of the modules PaymentChannelUser or HTLCUser. Because the set $c$TimelockRegions can only become smaller, the value of LedgerTime in step $\sstep{m_{c'}(s)}{m_{c'}(t)}$ does not change.

The remaining variables are changed by the following actions of the modules PaymentChannelUser and HTLCUser: PublishFundingTransaction, RedeemHTLCAfterClose, NoteThatHTLCFulfilledOnChainByOtherUser, NoteThatHTLCFulfilledOnChainInOtherChannel, RequestInvoice, GenerateAndSendPaymentHash, ReceivePaymentHash, SendHTLCPreimage, ReceiveHTLCPreimage, AddAndSendOutgoingHTLC, and ReceiveUpdateAddHTLC.
In the following, we prove for each of these actions that a step $\sstep{s}{t}$ of these actions in channel $c$ is a step $\sstep{m_{c'}(s)}{m_{c'}(t)}$ in channel $c'$ and either a stuttering step, a step of an action of MultiHopMock or of HTLCUser for the same user.

\step{1}{If $\sstep{s}{t}$ is a step of PublishFundingTransaction of user $u$ in channel $c$, then $\sstep{m_{c'}(s)}{m_{c'}(t)}$ is either a stuttering step, a step of an action of MultiHopMock or of HTLCUser for user $u$.}
  \begin{proof}
  \pf\ The variables that are changed by PublishFundingTransaction and are not specific to channel $c$ are $u$ExtBalance and $u$ChannelBalance.
  By definition of $m_c$, both variables are unchanged in step $\sstep{m_{c'}(s)}{m_{c'}(t)}$ if user $u$ is not part of channel $c'$.
  In the step $\sstep{s}{t}$ the channel balance of user $u$ is increased by the funding amount which is the user's external balance.
  This change is described by the action UserOpensPaymentChannel of the module MultiHopMock which decreases the balance of user $u$ by to zero and adds it to the channel balance of user $u$ stored in the variable $u$ChannelBalance.
  Thus, these variables of user $u$ change in the same way as described by the action PublishFundingTransaction.
  The information that this balance is stored in an external channel from the view of channel $c'$ is stored by adding a record to the variable $u$RequestedInvoices tat stores the amount that is stored in the other channel.
  As all other variables remain unchanged, $\sstep{m_{c'}(s)}{m_{c'}(t)}$ is a step of UserOpensPaymentChannel if user $u$ is part of channel $c'$ and a stuttering step otherwise.
  \end{proof}
\step{2}{If $\sstep{s}{t}$ is a step of RedeemHTLCAfterClose or SendHTLCPreimage of user $u$ in channel $c$, then $\sstep{m_{c'}(s)}{m_{c'}(t)}$ is either a stuttering step, a step of an action of MultiHopMock or of HTLCUser for user $u$.}
  \begin{proof}
  \pf\ The variables that are changed by RedeemHTLCAfterClose and SendHTLCPreimage and are not specific to channel $c$ are $u$Payments and $u$ChannelBalance.
  By definition of $m_c$, both variables are unchanged in step $\sstep{m_{c'}(s)}{m_{c'}(t)}$ if user $u$ is not part of channel $c'$.
  Both variables are only changed if a payment is processed. This can only be the case if $u$ is the receiver of the payment.
  This change of the variables in step $\sstep{m_{c'}(s)}{m_{c'}(t)}$ is described by the action ProcessOtherPayment of MultiHopMock which describes that any payment of which the user $u$ is the receiver may be processed and the respective balance be received.
  \end{proof}
\step{3}{If $\sstep{s}{t}$ is a step of NoteThatHTLCFulfilledOnChainByOtherUser or ReceiveHTLCPreimage of user $u$ in channel $c$, then $\sstep{m_{c'}(s)}{m_{c'}(t)}$ is either a stuttering step, a step of an action of MultiHopMock or of HTLCUser for user $u$.}
  \begin{proof}
  \pf\ The variables that are changed by NoteThatHTLCFulfilledOnChainByOtherUser and ReceiveHTLCPreimage and are not specific to channel $c$ are $u$PreimageInventory, $u$LatePreimages, $u$Payments and $u$ChannelBalance.
  By definition of $m_c$, these variables are unchanged in step $\sstep{m_{c'}(s)}{m_{c'}(t)}$ if user $u$ is not part of channel $c'$.
  Further, by definition of $m_c$, the variables $u$PreimageInventory and $u$LatePreimages are unchanged if the preimage that is being read from the blockchain is not in the set $R_{u,c'}$.
  If user $u$ is the receiver of the payment associated with the fulfilled HTLC, then the HTLC is not part of another channel of user $u$ and only the variables $u$Payments and $u$ChannelBalance are changed.
  This change in the step $\sstep{m_{c'}(s)}{m_{c'}(t)}$ is described by the action ProcessOtherPayment of MultiHopMock which describes that any payment of which the user $u$ is the sender may be processed and the respective balance be deducted from $u$ChannelBalance.
  If user $u$ is not the receiver of the payment associated with the fulfilled HTLC, then the variables $u$Payments and $u$ChannelBalance are unchanged in step $\sstep{s}{t}$ and, by definition of $m_{c'}$, they are unchanged in step $\sstep{m_{c'}(s)}{m_{c'}(t)}$.
  The variables $u$PreimageInventory and $u$LatePreimages change only in step $\sstep{m_{c'}(s)}{m_{c'}(t)}$ if the preimage of the fulfilled HTLC is in $R_{u,c'}$.
  As channels $c'$ and $c$ are different channels and, if $u$PreimageInventory is changed by NoteThatHTLCFulfilledOnChainByOtherUser or ReceiveHTLCPreimage, the received preimage must be for an outgoing HTLC in channel $c$, the HTLC for which the preimage is received must be an incoming HTLC in channel $c'$.
  Thus, the change to the variables $u$PreimageInventory and $u$LatePreimages in step $\sstep{m_{c'}(s)}{m_{c'}(t)}$ that a preimage is received is described by the action ReceivePreimageForIncomingHTLC of the module MultiHopMock.
  \end{proof}
\step{4}{If $\sstep{s}{t}$ is a step of NoteThatHTLCFulfilledOnChainInOtherChannel of user $u$ in channel $c$, then $\sstep{m_{c'}(s)}{m_{c'}(t)}$ is either a stuttering step, a step of an action of MultiHopMock or of HTLCUser for user $u$.}
  \begin{proof}
  \pf\ The variables that are changed by NoteThatHTLCFulfilledOnChainInOtherChannel and are not specific to channel $c$ are $u$PreimageInventory and $u$LatePreimages.
  By definition of $m_c$, both variables are unchanged in step $\sstep{m_{c'}(s)}{m_{c'}(t)}$ if user $u$ is not part of channel $c'$.
  By definition of NoteThatHTLCFulfilledOnChainInOtherChannel, the HTLC that is being fulfilled must be an incoming HTLC.
  Thus, the change to the variables $u$PreimageInventory and $u$LatePreimages in step $\sstep{m_{c'}(s)}{m_{c'}(t)}$ that a preimage is received is described by the action ReceivePreimageForIncomingHTLC of the module MultiHopMock.
  \end{proof}
\step{5}{If $\sstep{s}{t}$ is a step of RequestInvoice of user $u$ in channel $c$, then $\sstep{m_{c'}(s)}{m_{c'}(t)}$ is either a stuttering step, a step of an action of MultiHopMock or of HTLCUser for user $u$.}
  \begin{proof}
  \pf\ The action RequestInvoice is not specific to a channel but only specific to user $u$.
  If user $u$ is in channel $c'$, step $\sstep{m_{c'}(s)}{m_{c'}(t)}$ is also described by RequestInvoice of the module HTLCUser.
  If user $u$ is not in channel $c'$, the global variable Messages might be changed in step $\sstep{m_{c'}(s)}{m_{c'}(t)}$ if a message is sent from user $u$ to a user of channel $c'$.
  Such a step $\sstep{m_{c'}(s)}{m_{c'}(t)}$ is described by the action RequestInvoice of the module MultiHopMock.
  If no user of channel $c'$ is the recipient of the message sent by $u$, then, by definition of $m_{c'}$, the variable Messages is unchanged in $\sstep{m_{c'}(s)}{m_{c'}(t)}$ and $\sstep{m_{c'}(s)}{m_{c'}(t)}$ is a stuttering step.
  \end{proof}
\step{6}{If $\sstep{s}{t}$ is a step of GenerateAndSendPaymentHash of user $u$ in channel $c$, then $\sstep{m_{c'}(s)}{m_{c'}(t)}$ is either a stuttering step, a step of an action of MultiHopMock or of HTLCUser for user $u$.}
  \begin{proof}
  \pf\ The action GenerateAndSendPaymentHash is not specific to a channel but only specific to user $u$.
  If user $u$ is in channel $c'$, step $\sstep{m_{c'}(s)}{m_{c'}(t)}$ is also described by GenerateAndSendPaymentHash of the module HTLCUser.
  If user $u$ is not in channel $c'$, the global variable Messages might be changed in step $\sstep{m_{c'}(s)}{m_{c'}(t)}$ if user $u$ replies to a message sent from a user of channel $c'$.
  Such a step $\sstep{m_{c'}(s)}{m_{c'}(t)}$ is described by the action GenerateAndSendPaymentHash of the module MultiHopMock.
  The value included in the reply is deterministic as it can be deducted from the payment id for which an invoice is requested.
  Thus, the action GenerateAndSendPaymentHash of the module MultiHopMock can define the reply that an actual user would send.
  If user $u$ replies to a message that was not sent by a user from channel $c'$, by definition of $m_{c'}$, the variable Messages is unchanged in $\sstep{m_{c'}(s)}{m_{c'}(t)}$ and $\sstep{m_{c'}(s)}{m_{c'}(t)}$ is a stuttering step.
  \end{proof}
\step{7}{If $\sstep{s}{t}$ is a step of ReceivePaymentHash of user $u$ in channel $c$, then $\sstep{m_{c'}(s)}{m_{c'}(t)}$ is either a stuttering step, a step of an action of MultiHopMock or of HTLCUser for user $u$.}
  \begin{proof}
  \pf\ The action ReceivePaymentHash is not specific to a channel but only specific to user $u$.
  If user $u$ is in channel $c'$, step $\sstep{m_{c'}(s)}{m_{c'}(t)}$ is also described by ReceivePaymentHash of the module HTLCUser.
  If user $u$ is not in channel $c'$, the global variable Messages might be changed in step $\sstep{m_{c'}(s)}{m_{c'}(t)}$ if user $u$ receives a message sent by a user from channel $c'$.
  Such a step $\sstep{m_{c'}(s)}{m_{c'}(t)}$ is described by the action ReceivePaymentHash of the module MultiHopMock that describes for the user in channel $c'$ that the received message is removed from the Messages variable.
  If user $u$ receives a message that was not sent by a user from channel $c'$, by definition of $m_{c'}$, the variable Messages is unchanged in $\sstep{m_{c'}(s)}{m_{c'}(t)}$ and $\sstep{m_{c'}(s)}{m_{c'}(t)}$ is a stuttering step.
  \end{proof}
\step{8}{If $\sstep{s}{t}$ is a step of AddAndSendOutgoingHTLC of user $u$ in channel $c$, then $\sstep{m_{c'}(s)}{m_{c'}(t)}$ is either a stuttering step, a step of an action of MultiHopMock or of HTLCUser for user $u$.}
  \begin{proof}
  \pf\ The only variable that is changed by AddAndSendOutgoingHTLC and is not specific to channel $c$ is $u$NewPayments.
  By definition of $m_c$, this variable is unchanged in step $\sstep{m_{c'}(s)}{m_{c'}(t)}$ if user $u$ is not part of channel $c'$.
  The change to the variable $u$NewPayments is that a payment is removed.
  This change is described by the action AddOutgoingHTLCToOtherChannel of the module MultiHopMock.
  \end{proof}
\step{9}{If $\sstep{s}{t}$ is a step of ReceiveUpdateAddHTLC of user $u$ in channel $c$, then $\sstep{m_{c'}(s)}{m_{c'}(t)}$ is either a stuttering step, a step of an action of MultiHopMock or of HTLCUser for user $u$.}
  \begin{proof}
  \pf\ The only variable that is changed by ReceiveUpdateAddHTLC and is not specific to channel $c$ is $u$NewPayments.
  By definition of $m_c$, this variable is unchanged in step $\sstep{m_{c'}(s)}{m_{c'}(t)}$ if user $u$ is not part of channel $c'$.
  The change to the variable $u$NewPayments is that a payment is added that should be forwarded.
  This change is described by the action AddNewForwardedPayment of the module MultiHopMock.
  The action AddNewForwardedPayment describes the payments that can be added based on the initial payments of other users and calculates the parameters of the payment to be forwarded based on the position of the channel $c'$ in the path of the payment.
  \end{proof}
\qedstep
  \begin{proof}
  \pf\ The steps above discussed all actions that change variables that have an effect on the variables in step $\sstep{m_{c'}(s)}{m_{c'}(t)}$.
  Thus, for all steps $\sstep{s}{t}$ of other actions, the step $\sstep{m_{c'}(s)}{m_{c'}(t)}$ is a stuttering step.
  \end{proof}

\end{proof}

\begin{lemma}
\label{lemma-ii-to-iia}
For each channel $c \in C$, the mapping $m_c$ maps a state of specification $(II)$ to a valid state of specification $(IIa)$ of channel $c$.
Formally: $\A s \in \Sigma_{(II)}, c \in C : m_c(s) \in \Sigma_{(IIa)}$
\end{lemma}

\begin{proof}
We prove the lemma by induction:

Let $F_i$ be the set of initial states of specification $S_i$.

\step{1}{$\A s \in F_{(II)}, c \in C : m_c(s) \in F_{(IIa)}$}
  \begin{proof}
  By definition of the initial states, applying the selection by $m_c$ to an initial state of specification $(II)$ results in an initial state of specification $(IIa)$.
  \end{proof}

\step{2}{\assume{$s \in \Sigma_{(II)}$ and $m_c(s) \in \Sigma_{(IIa)}$ and $\sstep{s}{t}$ is a step of specification $(II)$}
          \prove{$m_c(t) \in \Sigma_{(IIa)}$}
        }
  \begin{proof}
  Assume that state $s$ is a valid state of specification $(II)$ and $m_c(s)$ is a valid state of specification $(IIa)$.
  We show that for a step $\sstep{s}{t}$ to a state $t$ of specification $(II)$, it holds that $m_c(t)$ is a valid state of specification $(IIa)$.
  A step of specification $(II)$ is either a step of a (1) PaymentChannelUser, (2) HTLCUser, (3) LedgerTime, or (4) a final withdraw step.
  \step{2-1}{If $\sstep{s}{t}$ is a step of PaymentChannelUser, then $m_c(t) \in \Sigma_{(IIa)}$}
    \begin{proof}
    If the step $\sstep{s}{t}$ is a step of PaymentChannelUser, it is either (1.1) a step of an action of PaymentChannelUser for channel $c$ or (1.2) for another channel.
    \step{2-1-1}{If $\sstep{s}{t}$ is a step of PaymentChannelUser for channel $c$, then $m_c(t) \in \Sigma_{(IIa)}$}
      \begin{proof}
      If the step $\sstep{s}{t}$ is a step of an action of PaymentChannelUser for channel $c$ then $\sstep{m_c(s)}{m_c(t)}$ is a step of specification $(IIa)$ because specification $(IIa)$ allows exactly the same action of PaymentChannelUser that changes the variables selected by the mapping $m_c$.
      Because $\sstep{m_c(s)}{m_c(t)}$ is a step of specification $(IIa)$, $m_c(t)$ is a state of specification $(IIa)$.
      \end{proof}
    \step{2-1-2}{If $\sstep{s}{t}$ is a step of PaymentChannelUser for a channel other than channel $c$, then $m_c(t) \in \Sigma_{(IIa)}$}
      \begin{proof}
      If the step $\sstep{s}{t}$ is a step of an action of PaymentChannelUser for a channel other than channel $c$, the variables in specification $(II)$ that are specific to channel $c$ do not change between states $s$ and $t$. However, the variables for the users participating in channel $c$ might change in the step $\sstep{s}{t}$ or global variables might change. For example, a user of channel $c$ might receive a preimage. To account for such changes, specification $(IIa)$ contains the module MultiHopMock.
      The module MultiHopMock abstracts all actions that can happen in other payment channels.
      By \cref{lemma-multi-hop-mock-refinement}, it holds that $m_c(t) \in \Sigma_{(IIa)}$ and, thus, $m_c(t) \in \Sigma_{(IIa)}$.
      \end{proof}
    \qedstep
      \begin{proof}
      By \stepref{2-1-1} and \stepref{2-1-2}.
      \end{proof}
    \end{proof}
  \step{2-2}{If $\sstep{s}{t}$ is a step of HTLCUser, then $m_c(t) \in \Sigma_{(IIa)}$}
    \begin{proof}
    A step of HTLCUser is either a step of an action of HTLCUser for channel $c$ or for another channel.
    \step{2-2-1}{If $\sstep{s}{t}$ is a step of HTLCUser for channel $c$, then $m_c(t) \in \Sigma_{(IIa)}$}
      \begin{proof}
      If the step $\sstep{s}{t}$ is a step of an action of HTLCUser for channel $c$ then $\sstep{m_c(s)}{m_c(t)}$ is a step of specification $(IIa)$ because specification $(IIa)$ allows exactly the same action of HTLCUser that changes the variables selected by the mapping $m_c$.
      Because $\sstep{m_c(s)}{m_c(t)}$ is a step of specification $(IIa)$, $m_c(t)$ is a state of specification $(IIa)$.
      \end{proof}
    \step{2-2-2}{If $\sstep{s}{t}$ is a step of HTLCUser for a channel other than channel $c$, then $m_c(t) \in \Sigma_{(IIa)}$}
      \begin{proof}
      If the step $\sstep{s}{t}$ is a step of an action of HTLCUser for a channel other than channel $c$, the variables in specification $(II)$ specifically for channel $c$ do not change between states $s$ and $t$. However, the variables for the users participating in channel $c$ might change in the step $\sstep{s}{t}$ or global variables might change.
      As in step \stepref{2-1}, we use the module MultiHopMock to account for this. For every action in HTLCUser that changes a variable that is also used in another channel, there is an action in MultiHopMock that describes the changes to the variables.
      By \cref{lemma-multi-hop-mock-refinement}, it holds that $m_c(t) \in \Sigma_{(IIa)}$ and, thus, $m_c(t) \in \Sigma_{(IIa)}$.
      \end{proof}
    \qedstep
      \begin{proof}
      By \stepref{2-2-1} and \stepref{2-2-2}.
      \end{proof}
    \end{proof}
  \step{2-3}{If $\sstep{s}{t}$ is a step of LedgerTime, then $m_c(t) \in \Sigma_{(IIa)}$}
    \begin{proof}
    If the step $\sstep{s}{t}$ is a step of LedgerTime, the only change between states $s$ and $t$ is an increase in LedgerTime. By the definition of $m_c$, LedgerTime changes in the step $\sstep{m_c(s)}{m_c(t)}$ as defined by the module LedgerTime.
    Thus, the step $\sstep{m_c(s)}{m_c(t)}$ is a step of specification $(IIa)$.
    Thus, $m_c(t) \in \Sigma_{(IIa)}$.
    \end{proof}
  \step{2-4}{If $\sstep{s}{t}$ is a final withdraw step, then $m_c(t) \in \Sigma_{(IIa)}$}
    \begin{proof}
    Assume that $\sstep{s}{t}$ is a final withdraw step in specification $(II)$.
    The final withdraw action is defined in specification $(IIa)$ as in specification $(II)$ with the exception of the new value of $u$ExtBalance for each user $u$.
    Thus, we only need to show that the change of $u$ExtBalance in $\sstep{m_c(s)}{m_c(t)}$ conforms to the final withdraw action of specification $(IIa)$.
    The balance of dishonest users does not change in step $\sstep{s}{t}$ which matches the definition of the final withdraw action in specification $(IIa)$.
    In step $\sstep{s}{t}$, the value of $u$ExtBalance for each honest user $u$ is increased by user $u$'s on-chain balance.
    In specification $(IIa)$ the visible on-chain balance might be less than in specification $(II)$ because only the on-chain transaction for channel $c$ are visible.
    However, because specification $(IIa)$ requires that the value of $u$ExtBalance for an honest user $u$ is at least the user's previous balance increased by the on-chain balance from channel $c$, the user $u$'s new external balance can be larger than the balances visible in specification $(IIa)$ and, therefore, $\sstep{m_c(s)}{m_c(t)}$ is a step of specification $(IIa)$.
    It follows that $m_c(t) \in \Sigma_{(IIa)}$.
    \end{proof}
  \qedstep
    \begin{proof}
    By \stepref{2-1}, \stepref{2-2}, \stepref{2-3}, and \stepref{2-4} because a step in specification $(II)$ can only be either a step of PaymentChannelUser, HTLCUser, LedgerTime, or a final withdraw step.
    \end{proof}
  \end{proof}
\qedstep
  \begin{proof}
  By induction using \stepref{1} and \stepref{2}.
  \end{proof}
\end{proof}

Let $F$ be the set of the refinement mappings $f_c$ for all channels $c \in C$ from specification $(IIa)$ to specification $(IIIa)$ that we defined in the TLA\textsuperscript{+} code of the formalization in \textit{SpecificationIIatoIIIa.tla}.

\begin{definition}[$g$]
Define the function $g$ from $\Sigma_{(II)}$, the state space of specification $(II)$, to $\Sigma_{(III)}$, the state space of specification $(III)$, as follows.
For a state $s$ of specification $(II)$,
let $g(s)$ be a state of specification $(III)$ with the following value to variable assignments:
\begin{itemize}
    \item The global variables $v_\mathrm{g}$ are assigned as follows:
    \begin{itemize}
        \item Messages is assigned $s.\mathrm{Messages}$ which equals the value $\bigcup_{c \in C} f_{c}(m_{c}(s)).\mathrm{Messages}$, the union of all global messages.
        \item LedgerTime is assigned the value $s.\mathrm{LedgerTime}$ which equals the value $max_{c \in C}(f_c(m_c(s)).\mathrm{LedgerTime})$.
    \end{itemize}
    \item The variables $v_u$ that concern one specific user $u$ are assigned the value $s.v$ because these variables are left unchanged by $f_c$.
    \item In the set of variables $v_c$ that concern one specific channel $c \in C$, there are the variables $c$Messages and $c$PendingBalance. This variables are assigned the value $f_c(m_c(s)).\mathrm{ChannelMessages}$ and $f_c(m_c(s)).\mathrm{PendingBalance}$ respectively.
    \item For the variables $v \in v_{c,u}$ that concern one specific user $u$ and channel $c$, the value of the variable $v$ in state $g(s)$ is defined as $f_c(m_c(s)).v$.
\end{itemize}
\end{definition}

\begin{lemma}
\label{lemma-iiia-additional-values}
A step of IdealChannel or HTLCUser in specification $(IIIa)$ is also possible if the variables that are reduced by $m_c$ (i.e., $u$PreimageInventory, $u$LatePreimages, $u$NewPayments, Messages) contain additional values that are filtered out by $m_c$ (see \cref{definition-m_c}).
\end{lemma}

\begin{proof}
\pf\ 
\step{1}{All conditions on the first state of a step of IdealChannel or HTLCUser are valid even if the variables that are reduced by $m_c$ contain additional values that are filtered out by $m_c$.}
  \begin{proof}
  \pf\ It can be checked for every action of IdealChannel and HTLCUser and every variable that is reduced by $m_c$ that this property holds.
  For this proof, we discuss all relevant cases. The remaining cases are either trivial or analogous to the discussed cases.

  The action RequestInvoice of HTLCUser requests an invoice only if there are no payments with lower timelock in $u$NewPayments that are initiated by the user $u$.
  Because $m_c$ retains all payments that are initiated by the user $u$ in $u$NewPayments, the additional payments that are filtered out by $m_c$ are not relevant for this condition.

  The action GenerateAndSendPaymentHash of HTLCUser adds a preimage for the requested payment to $u$PreimageInventory and a payment secret to $u$PaymentSecretForPreimage. The action contains a condition that the generated preimage may not already be in $u$PreimageInventory.
  Because the value chosen for the preimage in the formalization depends on and is unique for the payment that the preimage is for, if the preimage has already been generated then it must have been generated by this action and must have been added to the domain of $u$PaymentSecretForPreimage and, therefore, it is selected by $m_c$.

  The action AddAndSendOutgoingHTLC selects a payment that has the lowest timelock of all payments for which the next hop is this channel.
  Because $m_c$ retains all payments in $u$NewPayments for which the next hop is this channel, additional payments that are not filtered out by $m_c$ do not affect this condition.

  In IdealChannel, some actions have checks whether the preimage of an HTLC is in $u$PreimageInventory.
  Because these checks are only for preimages for HTLCs of the users and $m_c$ retains all preimages for which an HTLC exists, these checks return the same result if $u$PreimageInventory contains more values.
  \end{proof}
\step{2}{All conditions on the second state of a step of IdealChannel or HTLCUser are valid even if the variables that are reduced by $m_c$ contain additional values that are filtered out by $m_c$.}
  \begin{proof}
  \pf\ 
  The values of the concerned variables in the second state of a step are always described by the actions of IdealChannel and HTLCUser by describing a change and not directly the new value.
  More concretely, the actions define that a value is added or removed to the set and all remaining values of the set remain unchanged.
  Thus, a step of IdealChannel or HTLCUser is still valid if the concerned variables contain additional values.
  \end{proof}
\qedstep
  \begin{proof}
  \pf\ 
  Because the conditions on the first and the second state are valid, a step of IdealChannel or HTLCUser is also a valid step if the variables that are reduced by $m_c$ contain additional values that are filtered out by $m_c$.
  \end{proof}
\end{proof}

\begin{theorem}
\label{theorem-ref-II-III}
The function $g$ is a refinement mapping from specification $(II)$ to specification $(III)$.
\end{theorem}

\begin{proof}
\step{1}{For all $s \in \Sigma_{(II)} : \Pi(g(s)) = \Pi(s)$}
    \begin{proof}
    \pf\ The externally visible variables are for all users $u \in U$: $u$Payments, $u$ExtBalance, $u$ChannelBalance, and $u$Honest.
    These variables are by definition unchanged by $g$.
    \end{proof}
\step{2}{For all initial states $s$ of specification $(II)$, $g(s)$ is an initial state of specification $(III)$}
    \begin{proof}
    \pf\  All variables in state $s$ are mapped by $g$ either by directly assigning the value of the variable in state $s$ or by applying functions $f_c \in F$.
    By definition of $f_c$, $f_c$ does not change values in the initial states.
    By definition of Init of specifications $(II)$ and $(III)$, it follows that all variables of a state of specification $(II)$ are mapped to values of an initial state of specification $(III)$.
    \end{proof}
\step{3}{For all steps $\sstep{s}{t}$ of specification $(II)$, $\sstep{g(s)}{g(t)}$ is a step of specification $(III)$.}
    \begin{proof}
    \pf\ 

    A step of specification $(II)$ can be one of the following.
    
    \step{3-1}{\case{$\sstep{s}{t}$ is a step of an action of the module LedgerTime in specification $(II)$.}
                \prove{$\sstep{g(s)}{g(t)}$ is a step of the module LedgerTime in specification $(III)$.}}
        \begin{proof}
        The only variable that is updated in step is the variable LedgerTime.
        Because the way that the refinement mappings $f_c$ update variables does not depend on the value of LedgerTime, all variables other than LedgerTime are unchanged in step $\sstep{g(s)}{g(t)}$.
        It follows that $\sstep{g(s)}{g(t)}$ is a step of the module LedgerTime in specification $(III)$ because all variables except LedgerTime are unchanged and $g(s).\mathrm{LedgerTime} = g(t).\mathrm{LedgerTime}$ and specification $(III)$ allows LedgerTime to increase by at least the values that are allowed in specification $(II)$ because specification $(III)$ is a regular (unoptimized) real-time specification.
        \end{proof}
    \step{3-2}{\case{$\sstep{s}{t}$ is a step of an action of the module HTLCUser or PaymentChannelUser in specification $(II)$.}
                \prove{$\sstep{g(s)}{g(t)}$ is a step of HTLCUser or IdealChannel in specification $(III)$.}}
        \begin{proof}
        Let $c \in C$ be the channel and $u$ the user for which the step $\sstep{s}{t}$ is a step of an action of HTLCUser or PaymentChannelUser.
        By the proof of \cref{lemma-ii-to-iia}, the step $\sstep{m_c(s)}{m_c(t)}$ is a step of an action of HTLCUser or PaymentChannelUser for user $u$ in the instance of specification $(IIa)$ that is described by the mapping $m_c$.
        By \cref{lemma-multi-hop-mock-refinement}, the step $\sstep{m_{c'}(s)}{m_{c'}(t)}$ is a step of HTLCUser, MultiHopMock or a stuttering step in the instance of specification $(IIa)$ for every channel $c' \neq c$ that is described by the mapping $m_{c'}$.
        By \cref{lem-refinement-mapping}, the step $\sstep{f_c(m_c(s))}{f_c(m_c(t))}$ is a step of HTLCUser or IdealChannel in specification $(IIIa)$ and the step $\sstep{f_{c'}(m_{c'}(s))}{f_{c'}(m_{c'}(t))}$ is a step of HTLUser, MultiHopMock or a stuttering step in specification $(IIIa)$.
        The channel-specific variables of all channels $c' \neq c$ and the user-specific variables of all users $u' \neq u$ are unchanged in the step $\sstep{s}{t}$ and, by definition of $g$, these variables are also unchanged in step $\sstep{g(s)}{g(t)}$.
        Therefore, only the global variables, user-specific variables for user $u$ and channel-specific variables for channel $c$ might be changed in step $\sstep{g(s)}{g(t)}$.
        Because the step $\sstep{f_c(m_c(s))}{f_c(m_c(t))}$ is a step of HTLCUser or IdealChannel in specification $(IIIa)$, by \cref{lemma-iiia-additional-values} and the definition of $g$, the step $\sstep{g(s)}{g(t)}$ is a step of HTLCUser or IdealChannel in specification $(III)$.
        \end{proof}
    \end{proof}
\qedstep
    \begin{proof}
    \pf\ By \stepref{1}, \stepref{2}, \stepref{3}, and the definition of refinement mappings and because specification $(III)$ contains the same fairness condition as specification $(II)$ that only behaviors are valid that end in a state in which all users have withdrawn their balance.
    \end{proof}
\end{proof}

\subsection{Proof of Application of Generalized Time Skip Theorem to Specification $(III)$}
\label{sec-appendix-time-skip-III-IV}

Analogously to \cref{sec-appendix-time-skip-application-I-II}, to prove that \cref{theorem-general-time-skip} can be applied to specification $(III)$, we need to prove that, for each clock $x \in \mathcal{X}$, the set $relETP^x$ defined in the specification meets the requirements of \cref{assumptions-relETP} and that the time bounds used in the specification meet the requirements of \cref{assumption-b-independent}. 
In contrast to specification $(I)$, specification $(III)$ contains only one clock which is LedgerTime.

We start by proving that the time bounds specified in SpecificationIII.tla meet \cref{assumption-b-independent}.
\begin{proof}
\pf\ $B^{LedgerTime}$ is defined by $TimeBounds$ of $IdealChannel$.
As the variable $LedgerTime$ does not occur in the definition of $TimeBounds$, the definition of $TimeBounds$ is independent of $LedgerTime$ and, thus, \cref{assumption-b-independent} holds for $B^{LedgerTime}$.
\end{proof}

We prove that $relETP^x$ meets \cref{assumptions-relETP}.
The set $relETP^{LedgerTime}$ is defined as $relETP$ in SpecificationIV.

\begin{proof}
\pf\ 
First, we prove the second statement:
\step{1}{$\A x \in \mathcal{X}, s \in \Sigma : \{ b +1 \mid b \in B^x(s) \} \subseteq relETP^x(s)$}
  \begin{proof}
  \pf\
  The statement directly follows from the definition of $relETP$ in SpecificationIIa that defines $relETP$ as union of $\{t + 1 : t \in TimeBounds\}$ and another set.
  \end{proof}
\step{2}{$\A x \in \mathcal{X}, s \in \Sigma : ETP^x(s) \subseteq relETP^x(s)$}
  \begin{proof}
  \pf\ We have to show that for every $n \in ETP^x(s)$, it holds that $n \in relETP^x(s)$.
  This means that if there is a newly possible behavior at time $n$, then time $n$ must be in $relETP^x(s)$.

  We use the same proof strategy as \cref{sec-appendix-time-skip-application-I-II}.
  We go through all subactions $A$ of $NextI$, we find the conditions under which an $A$ step exists that is enabled at time $n$ but not at time $n-1$, and verify that in all states $s$ from which a state that meets these conditions can be reached by $NextI$ steps it holds that $n \in relETP^x(s)$.
  
  Define a mapping $a$ from $ETP^x(s)$ to the subactions of the specification, that assigns to each $n \in ETP^x(s)$ a subaction $A$ so that $\astep{s}{t}{A}$ but not $\astep{\hat{T}^x_{n-1}(s)}{\hat{T}^x_{n-1}(t)}{A}$, i.e., a step of action $a(n)$ is not possible if clock $x$ is set to $n-1$ and becomes possible at time $n$.

  The module HTLCUser contains five actions that depend on LedgerTime that have already been discussed in \cref{sec-appendix-time-skip-application-I-II}. We do not repeat the proofs for these actions here.

  The module IdealChannel contains seven actions that depend on LedgerTime:
  \step{2-1}{$\A s \in \Sigma : \A n \in ETP^{LedgerTime}(s) : a(n) = \mathrm{UpdatePaymentChannel} \Rightarrow n \in relETP^{LedgerTime}(s)$}
    \begin{proof}
    \pf\ The action UpdatePaymentChannel depends on the value of LedgerTime to choose the HTLCs to add and to remove.
    For this, the condition that is used is the condition that an HTLCs timelock is greater than LedgerTime.
    Thus, the set $relETP^{LedgerTime}(s)$ must contain an HTLC's timelock because at the timelock of an HTLC a new step becomes possible that does not add a specific HTLC or that removes an HTLC that has timedout.
    This is fulfilled by the definition of TimelockRegions of the module IdealChannel that is included by $relETP^{LedgerTime}(s)$.

    The action UpdatePaymentChannel also contains a condition that under certain conditions the value of LedgerTime is smaller than an HTLC's absolute timelock plus the grade period $G$. For larger values of LedgerTime steps become impossible but this condition does not allow for new steps to become possible.
    \end{proof}
  \step{2-2}{$\A s \in \Sigma : \A n \in ETP^{LedgerTime}(s) : a(n) = \mathrm{SetOnChainHTLCsAndCheater} \Rightarrow n \in relETP^{LedgerTime}(s)$}
    \begin{proof}
    \pf\ As UpdatePaymentChannel, the action SetOnChainHTLCsAndCheater contains a condition that under certain conditions the value of LedgerTime is smaller than an HTLC's absolute timelock plus the grade period $G$. For larger values of LedgerTime steps become impossible but this condition does not allow for new steps to become possible.

    Another condition requires a user to have at least the balance of incoming HTLCs that cannot be persisted because the user was dishonest and the value of LedgerTime is larger than or equal to the HTLC's timelock + G.
    This condition might lead to steps becoming impossible but there are no steps of SetOnChainHTLCsAndCheater that can become possible.
    \end{proof}
  \step{2-3}{$\A s \in \Sigma : \A n \in ETP^{LedgerTime}(s) : a(n) = \mathrm{FulfillIncomingHTLCsOnChain} \Rightarrow n \in relETP^{LedgerTime}(s)$}
    \begin{proof}
    \pf\ The action FulfillIncomingHTLCsOnChain can fulfill an HTLC on-chain as long as the HTLC's timelock + G is greater than LedgerTime.
    Because the action chooses a subset of fulfillable HTLCs as the HTLCs to fulfill, if the value of LedgerTime reaches an HTLC's timelock + G no new step becomes possible but only all steps that include fulfilling this HTLC become impossible.

    Additionally, the action has the same condition for the balance of users as SetOnChainHTLCsAndCheater that also does not lead to steps becoming possible.
    \end{proof}
  \step{2-4}{$\A s \in \Sigma : \A n \in ETP^{LedgerTime}(s) : a(n) = \mathrm{NoteFulfilledHTLCsOnChain} \Rightarrow n \in relETP^{LedgerTime}(s)$}
    \begin{proof}
    \pf\ The action NoteFulfilledHTLCsOnChain can note a fulfilled HTLC on-chain as long as the HTLC's timelock + G is greater than LedgerTime.
    Because the action chooses a subset of fulfillable HTLCs as the HTLCs to fulfill, if the value of LedgerTime reaches an HTLC's timelock + G no new step becomes possible but only all steps that include fulfilling this HTLC become impossible.

    If the preimage for an HTLC is learned when the value of LedgerTime is greater than the HTLC's timelock + G, the preimage is added to the set $u$LatePreimages.
    Thus, for each HTLC that can be fulfilled, the set $relETP^{LedgerTime}(s)$ must include the HTLC's timelock + G + 1.
    This is fulfilled by the definition of TimelockRegions of the module IdealChannel that is included by $relETP^{LedgerTime}(s)$.
    \end{proof}
  \step{2-5}{$\A s \in \Sigma : \A n \in ETP^{LedgerTime}(s) : a(n) = \mathrm{CommitHTLCsOnChain} \Rightarrow n \in relETP^{LedgerTime}(s)$}
    \begin{proof}
    \pf\ The action CommitHTLCsOnChain can persist an HTLC on-chain as long as the HTLC's timelock + G is greater than LedgerTime.
    Because the action chooses a subset of persistable HTLCs as the HTLCs to persist, if the value of LedgerTime reaches an HTLC's timelock + G no new step becomes possible but only all steps that include persisting this HTLC become impossible.

    The action CommitHTLCsOnChain also contains a condition that under certain conditions the value of LedgerTime is smaller than an HTLC's absolute timelock plus the grade period $G$. For larger values of LedgerTime steps become impossible but this condition does not allow for new steps to become possible.

    If the preimage for an HTLC is persisted when the value of LedgerTime is greater than the HTLC's timelock + G, the preimage is added to the set $u$LatePreimages of the user who learns the preimage.
    Thus, for each HTLC that can be persisted, the set $relETP^{LedgerTime}(s)$ must include the HTLC's timelock + G + 1.
    This is fulfilled by the definition of TimelockRegions of the module IdealChannel that is included by $relETP^{LedgerTime}(s)$.

    Another condition requires a user to have at least the balance of incoming HTLCs that cannot be persisted because the user was dishonest and the value of LedgerTime is larger than or equal to the HTLC's timelock + G.
    This condition might lead to steps becoming impossible but there are no steps of CommitHTLCsOnChain that can become possible.
    \end{proof}
  \step{2-6}{$\A s \in \Sigma : \A n \in ETP^{LedgerTime}(s) : a(n) = \mathrm{FulfillHTLCsOnChain} \Rightarrow n \in relETP^{LedgerTime}(s)$}
    \begin{proof}
    \pf\ The action FulfillHTLCsOnChain can fulfill an HTLC on-chain as long as the HTLC's timelock + G is greater than LedgerTime.
    Because the action chooses a subset of fulfillable HTLCs as the HTLCs to fulfill, if the value of LedgerTime reaches an HTLC's timelock + G no new step becomes possible but only all steps that include fulfilling this HTLC become impossible.

    If the preimage for an HTLC is learned when the value of LedgerTime is greater than the HTLC's timelock + G, the preimage is added to the set $u$LatePreimages.
    Thus, for each HTLC that can be fulfilled, the set $relETP^{LedgerTime}(s)$ must include the HTLC's timelock + G + 1.
    This is fulfilled by the definition of TimelockRegions of the module IdealChannel that is included by $relETP^{LedgerTime}(s)$.
    \end{proof}
  \step{2-7}{$\A s \in \Sigma : \A n \in ETP^{LedgerTime}(s) : a(n) = \mathrm{ClosePaymentChannel} \Rightarrow n \in relETP^{LedgerTime}(s)$}
    \begin{proof}
    \pf\ The action ClosePaymentChannel contains a condition in the function ValidMapping that verifies that an HTLC can only be timed out if the value of LedgerTime is at least the HTLCs timelock. Thus, steps in which the HTLC is timed out are only valid from a state on in which the value of LedgerTime is at least the HTLC's timelock.
    Therefore, the $relETP^{LedgerTime}(s)$ must include the HTLC's timelock.
    This is fulfilled by the definition of TimelockRegions of the module IdealChannel that is included by $relETP^{LedgerTime}(s)$.
    The function ValidMapping contains a check that an HTLC can only be persisted as long as the value of LedgerTime is lower than the HTLC's timelock + G. An increasing value of LedgerTime does not enable new steps to become possible.

    If the preimage for an HTLC is learned when the value of LedgerTime is greater than the HTLC's timelock + G, the preimage is added to the set $u$LatePreimages.
    Thus, for each HTLC that can be fulfilled, the set $relETP^{LedgerTime}(s)$ must include the HTLC's timelock + G + 1.
    This is fulfilled by the definition of TimelockRegions of the module IdealChannel that is included by $relETP^{LedgerTime}(s)$.
    \end{proof}
  \qedstep
    \begin{proof}
    By \stepref{2-1}, \stepref{2-2}, \stepref{2-3}, \stepref{2-4}, \stepref{2-5}, \stepref{2-6}, and \stepref{2-7} we have proven for all subactions of $NextI$ that $\A x \in \mathcal{X}, s \in \Sigma : ETP^x(s) \subseteq relETP^x(s)$.
    \end{proof}
  \end{proof}
\qedstep
  \begin{proof}
  By \stepref{1} and \stepref{2} both statements of \cref{assumption-b-independent} are proven.
  \end{proof}
\end{proof}

\end{document}